\pdfoutput=1
\documentclass[11pt]{lmcs} %%% last changed 2014-08-20
\pdfoutput=1

\usepackage{lastpage}
\lmcsdoi{21}{3}{17}
\lmcsheading{}{\pageref{LastPage}}{}{}%
{Jul.~03,~2023}{Aug.~07,~2025}{}


\usepackage{proof}
\usepackage{cmll}
\usepackage{ amssymb }
\usepackage{upgreek}
\usepackage{mathrsfs}
\usepackage{stmaryrd}
\usetikzlibrary{calc}
\usepackage{hyperref}
\usepackage[utf8]{inputenc} %FIXED
\newcommand{\obseq}{\equiv_{\mathrm{obs}}}
\newcommand{\appeq}{\equiv_{\mathrm{app}}}
\newcommand{\MM}{\mathcal{M}}
\newcommand{\DD}{\mathscr{D}}
\newcommand{\Nat}{\mathbb{N}}
\newcommand{\Prog}[1]{\mathcal{P}_{#1}}

\newcommand{\redcr}{\to_\mathsf{cr}}
\newcommand{\redsc}{\to_\mathsf{pr}}
\newcommand{\reddsc}{\redd[\mathsf{pr}]}
\newcommand{\E}{\mathsf{E}}
\newcommand{\imp}{\,\Rightarrow\,}
\newcommand{\comb}[1]{\mathbf{#1}}
\newcommand{\Om}{\boldsymbol{\Upomega}}
\newcommand{\bsub}{\begin{enumerate}}
\newcommand{\esub}{\end{enumerate}}
\renewcommand{\bar}{\begin{array}}
\newcommand{\ear}{\end{array}}
\newcommand{\nat}{{\mathbb N}}
\newcommand{\st}{\mid}
\newcommand\set[1]{\{#1\}}
\newcommand{\lam}{\ensuremath{\lambda}} 
\newcommand{\tuple}[1]{\langle #1 \rangle}
\newcommand{\Tuple}[1]{\left\langle\begin{aligned} #1 \end{aligned}\right\rangle}
\newcommand{\cI}{\mathcal{I}}
\newcommand{\cR}{\mathcal{R}}
\newcommand{\cD}{\mathcal{D}}
\newcommand{\und}{\ \land \ }
\newcommand{\ToE}{\Downarrow^\E}

\newcommand{\Lam}{\Lambda}
\newcommand{\Lame}{\Lambda^{\E}}

\newcommand{\C}{\mathrm{C}}
\newcommand{\Ev}{\mathrm{E}}
\newcommand{\Val}[1][]{\mathrm{Val}^{#1}}
\newcommand{\redd}[1][]{\twoheadrightarrow_{#1}}
\newcommand{\Var}{\mathrm{Var}}
\newcommand{\PCF}{{\sf PCF}}
\newcommand{\tint}{\mathsf{int}}
\newcommand{\pred}{\mathbf{pred}\,}
\let\succc\succ
\renewcommand{\succ}{\mathbf{succ}\,}
\newcommand{\ifterm}[3]{\mathbf{ifz}(#1,#2,#3)}
\newcommand{\fix}{\mathbf{fix}\,}
\newcommand{\num}{\underline}
\newcommand{\subst}[2]{[#2/#1]}
\newcommand{\FV}[1]{\mathrm{FV}(#1)}

\newcommand{\esubst}[2]{\langle #2/#1  \rangle}

\newcommand{\EPCF}{{\sf EPCF}}

\newcommand{\wh}{\mathtt{wh}}
\newcommand{\dom}{\mathrm{dom}}

\newcommand{\ifc}[2]{\mathtt{Ifc}(#1,#2)}
\newcommand{\headsize}[1]{\lfloor #1 \rfloor}
\newcommand{\uf}[1]{#1^\dagger}

\newcommand{\redwh}{\to_\wh}
\newcommand{\reddwh}{\twoheadrightarrow_\wh}

\newcommand{\PCFvalrule}{(\mathrm{val})}

\newcommand{\predrule}{(\mathrm{pr})}
\newcommand{\succrule}{(\mathrm{sc})}
\newcommand{\predzrule}{(\mathrm{pr_0})}
\newcommand{\ifzzrule}{(\mathrm{ifz_0})}
\newcommand{\ifzrule}{(\mathrm{ifz_{>0}})}
\newcommand{\betarule}{(\beta_v)}
\newcommand{\fixrule}{(\mathrm{fix})}
\newcommand{\Enumrule}{(\mathrm{nat}^\E)}
\newcommand{\Elamrule}{(\lambda^\E)}

\newcommand{\Evarrule}{(\mathrm{var}^\E)}
\newcommand{\Epredrule}{(\mathrm{pr}^\E)}
\newcommand{\Esuccrule}{(\mathrm{sc}^\E)}
\newcommand{\Epredzrule}{(\mathrm{pr_0}^\E)}
\newcommand{\Eifzzrule}{(\mathrm{ifz_0}^\E)}
\newcommand{\Eifzrule}{(\mathrm{ifz_{>0}}^\E)}
\newcommand{\Ebetarule}{(\beta_v^\E)}
\newcommand{\Efixrule}{(\mathrm{fix}^\E)}

\newcommand{\Types}{\mathbb{T}}

\newcommand{\Addrs}{\mathbb{A}}
\newcommand{\Tapes}[1][\Addrs]{\mathcal{T}_{#1}}
\newcommand\cM[1][\Addrs]{\mathcal{M}_{#1}}
\newcommand{\appT}[2]{#1\,@\,#2\,}
\newcommand{\append}[2]{\appT{#1}{{[}#2{]}}}
\newcommand{\Lookinv}[1]{\#^{-1}(#1)}
\newcommand{\Lookup}[1]{\#{#1}}
\newcommand{\App}[2]{#1\cdot #2}
\newcommand{\Cons}[2]{#1::#2}
\newcommand{\nnat}[1][+]{\nat^{#1}}
\newcommand{\repl}[2]{[#1:=#2]}
\newcommand{\trans}[2][]{|#2|_{#1}}
\renewcommand{\int}[2][]{\llbracket #2\rrbracket^{#1}}
\newcommand{\rtrans}[3][\Delta]{\llparenthesis #2 \rrparenthesis^{#1}_{#3}}
\newcommand{\mPred}{\mach{Pred}}
\newcommand{\mSucc}{\mach{Succ}}
\newcommand{\mIfz}{\mach{Ifz}}
\newcommand{\mProj}[2]{\mach{Pr}_{#2}^{#1}}
\newcommand{\mAppn}[2]{\mach{Apply}_{#1}^{#2}}
\newcommand{\mY}{\mach{Y}}

\newcommand{\mach}{\mathsf} % machine
\newcommand{\Null}{\varnothing}
\newcommand{\len}[1]{\vert #1\vert}
\newcommand{\val}[1]{\oc #1}
\newcommand{\mM}{\mach{M}}
\newcommand{\mN}{\mach{N}}

\newcommand{\ins}[1]{\mathtt{#1}}
 \newcommand{\Call}[1]{\ins{Call}~#1}
\newcommand{\Load}[1]{\ins{Load}~#1}
\newcommand{\Ifz}[4]{#4\shortleftarrow \ins{Test}(#1,\,#2,\,#3)}
\newcommand{\Pred}[2]{#2\shortleftarrow \ins{Pred}(#1)}
\newcommand{\Succ}[2]{#2\shortleftarrow \ins{Succ}(#1)}
 \newcommand{\Apply}[3]{#3\shortleftarrow \ins{App}(#1,\,#2)}

\newcommand{\redh}{\to_{\mach{c}}}
\newcommand{\reddh}{\twoheadrightarrow_{\mach{c}}}
\newcommand{\convh}{\leftrightarrow_\mach{c}}
\renewcommand{\conv}[1][]{\leftrightarrow_{#1}}
\newcommand{\convg}{\succc_\mach{c}}

\theoremstyle{plain} %\crefname{satz}{Satz}{S\"atze}

\begin{document}

\title[A Fully Abstract Model of PCF Based on EAMS]{A Fully Abstract Model of PCF Based on Extended Addressing Machines}
%FIXED Need a short Title

% affiliations are numbered automatically with a, b, c (see below)
% use the optional argument to indicate the affiliation(s) of each author
% omit the argument if there is only one author, or only one affiliation
\author[B.~Intrigila]{Benedetto Intrigila\lmcsorcid{0000-0002-9754-7007}}[a]
\author[G.~Manzonetto]{Giulio Manzonetto\lmcsorcid{0000-0003-1448-9014}}[b]
\author[N.~M\"unnich]{Nicolas M\"unnich\lmcsorcid{0009-0000-0498-2118}}[c]

\thanks{Nicolas M\"unnich is partly supported by ANR JCJC Project CoGITARe, ANR-18-CE25-0001.}

% affiliation 1 (automatically numbered a)
\address{Dipartimento di Ingegneria dell'Impresa, University of Rome ``Tor Vergata'', Italy}	%optional
% write emails for all authors having that affiliation
\email{benedetto.intrigila@uniroma2.it}  %optional

\address{Université Paris Cité, CNRS, IRIF, F-75013, Paris, France}	%optional
\email{gmanzone@irif.fr}  %optional

% affiliation 2 (automatically numbered b)
\address{USPN,  LIPN, UMR 7030, CNRS, F-93430 Villetaneuse, France}	%optional
\email{manzonetto@univ-paris13.fr, munnich@lipn.univ-paris13.fr}  %optional

%% etc.

%% required for running head on odd and even pages, use suitable
%% abbreviations in case of long titles and many authors:

%%%%%%%%%%%%%%%%%%%%%%%%%%%%%%%%%%%%%%%%%%%%%%%%%%%%%%%%%%%%%%%%%%%%%%%%%%%

%% the abstract has to PRECEDE the command \maketitle:
%% be sure not to issue the \maketitle command twice!

\begin{abstract}
  \noindent Extended addressing machines (EAMs) have been introduced to represent higher-order sequential computations.
Previously, we have shown that they are capable of simulating---via an easy encoding---the operational semantics of PCF, extended with explicit substitutions.
In this paper we prove that the simulation is actually an equivalence: a PCF program terminates in a numeral exactly when the corresponding EAM terminates in the same numeral. 
It follows that the model of PCF obtained by quotienting typable EAMs by a suitable logical relation is adequate. 
From a definability result stating that every EAM in the model can be transformed into a PCF program with the same observational behavior, we conclude that the model is fully abstract for PCF.
\end{abstract}

\maketitle

%% start the paper here:
\section{Introduction}
% !TEX root = ../LMCS.tex
%!TEX spellcheck = en-US

The problem of determining when two programs are equivalent is central in computer science. For instance, it is necessary to verify that the optimizations performed by a compiler actually preserve the meaning of the input program. 
In \lam-calculi like Plotkin's $\PCF$~\cite{Plotkin77}, two programs $P_1,P_2$ are considered observationally equivalent, written $P_1\obseq P_2$, whenever it is possible to plug either $P_1$ or $P_2$ into any context $\C\square$ of appropriate type without noticing any difference in the global behavior, i.e., $\C[P_1]$ of type $\tint$ produces a numeral $\num k$ exactly when $\C[P_2]$ does.
The Full Abstraction Problem, raised by Milner in~\cite{Milner77}, aims at finding a syntax-independent description of $\obseq$ by means of denotational models. The quest for a fully abstract (\emph{FA}, for short) model of \PCF{} kept researchers busy for decades~\cite{Curien07}, as the usual Scott-continuous models did not enjoy this property, and culminated with FA models based on game semantics~\cite{AbramskyMJ94,Nickau94,HylandO00}.
Subsequently, a FA model of \PCF{} based on realizability techniques was presented in \cite{MarzRS99}.

\paragraph{Addressing machines.} In 2020, the first author initiated a research program focused on modeling higher-order sequential  computations through \emph{addressing machines} (AMs), originally introduced by Della~Penna in~\cite{DellaPennaTh}.
The intent is not to develop a denotational model grounded in abstract mathematical structures but rather to propose a computational model, alternative to the von Neumann architecture, where computation is driven by communication between machines instead of local operations.
Addressing machines owe their name to the fact that they operate exclusively by manipulating the addresses of other machines (just like pure \lam-terms operate on other \lam-terms rather than on ordinary data types).
Higher-order computations are realized by passing functional programs via these addresses.
An AM can read an address from its tape, store the result of applying an address to another and pass the execution to another machine by calling its address. 
A challenge in defining addressing machines (AMs) lies in the fact that an address uniquely identifies a machine in a specific \emph{state of execution}, meaning its registers and tape may contain the addresses of other AMs.
The problem is finding a way to assign addresses to machines, and to compute the address resulting from applying an address $a$ to an address $b$, while avoding a circular definition. 
The definition of AMs resolves this circularity issue by initially defining the machines over an arbitrary set $\Addrs$ of addresses, and then employing external maps---whose existence is guaranteed by set-theoretical principles---to establish these associations.

\paragraph{Addressing machines as a model of computation}
Before delving deeper into the formalism, we would like to justify AMs as a computational model and to explain the motivation behind their inception. 
We begin by observing that many significant results in recursion theory are obtained by passing a Turing Machine as input to another Turing Machine (possibly itself).
For example, this technique is a key component in formulating the Halting Problem and proving its undecidability~\cite{Davis1958,Turing1936}. 
Since Turing Machines accept only natural numbers as input, this phenomenon of self-referentiation is achievable only through appropriate encodings.
These encodings are also necessary to represent higher order computations, but difficult to handle in practice.
In fact, while efficient encodings of the \lam-calculus into Turing machines are well-documented in the literature~\cite{Accattoli23}, the reverse encoding is part of the \emph{folklore} and considered excessively verbose to be presented in detail.
By making the addresses first class citizen and by explicitly and uniformly relying on external mechanisms, the framework of AMs reduces to a minimum the technicalities concerning the encodings. 
This is the main feature that allowed Della Penna \emph{et~al.\ }to define a reasonably elegant translation from \lam-terms to AMs, and to show that the associated machine faithfully represents the behavior of the original \lam-term~\cite{DellaPennaIM21}.

A corollary of the \lam-definability result in~\cite{DellaPennaIM21} is that all numerical partial computable functions can be represented via AMs.
As previously mentioned, the definition of addressing machines depends however on the existence of an external function, called \emph{address table map}, bijectively associating each AM with its address. 
Since the set of machines and the set of addresses are both countably infinite there exist a \emph{continuum} of address table maps, most of which cannot be effective. 
On the one hand, choosing a non-computable function as address table map can lead to define AMs representing non-computable functions. 
This situation is reminiscent of Ershov's notion of \emph{pre-complete enumeration}~\cite{Visser80}, where a function $f:A\to B$ is considered computable w.r.t.\ two enumerations $\nu: \nat\to A,\mu:\nat\to B$ of its domain and codomain, and taking non-effective enumerations can lead to represent non-computable functions.
On the other hand, it is possible to construct effective maps by employing classical Gödelization techniques, therefore the formalism of addressing machines actually deserve to be considered a model of computation.

\paragraph{Extended addressing machines.}
In \cite{DellaPennaIM21}, Della Penna \emph{et~al.\ }successfully built a model of untyped \lam-calculus based on AMs, but also realized that performing calculations on natural numbers in this formalism is as awkward as using Church numerals.
To avoid this problem, we recently introduced \emph{extended addressing machines} (EAMs)~\cite{IntrigilaMM22} having additional instructions for performing basic arithmetic operations on addresses of `numeral' machines, a recursor machine representing a fixed point combinator, and a typing algorithm. 
A recursor can even be programmed in the original formalism, but we avoid any dependency on {\em self-application} by manipulating the addressing mechanism to grant it access to {\em its own address}. This can be seen as a very basic version of the reflection principle of programming languages. 
In the present paper we construct a fully abstract model of \PCF{} based on EAMs (Theorem~\ref{thm:FA}).
We start by recalling the definition of EAMs (Section~\ref{sec:EAMs}), their operational semantics and type system. 
Then, we define a translation (Definition~\ref{def:trans}) constructing an EAM from every typable \PCF{} term and prove that such a translation is type-preserving (Theorem~\ref{thm:transtyping}), and the resulting EAM correctly represents the behavior of the original term: 
\begin{quote}
A \PCF{} program of type $\tint$ reduces to a numeral $\num k$ exactly when the corresponding EAM reduces to the machine $\mach{k}$ representing $k\in\nat$ (Theorem~\ref{thm:simulation}).
\end{quote}
This equivalence is achieved using as intermediate language \EPCF, i.e.\ a \PCF{} extended with explicit substitutions, since their operational semantics coincide on closed programs of type $\tint$.
The result can be seen as a strengthening of the main theorem in \cite{IntrigilaMM22}, where only one direction was proven.
As a consequence, one gets that the model constructed by equating all typable EAMs having the same observational behavior is adequate for \PCF{} (Theorem~\ref{thm:adequacy}). 
We achieve completeness---the other side of full abstraction---by defining a `reverse' translation that associates a \PCF{} program with every typable EAM, and proving that the two translations are inverses up to observational equivalence on EAMs (Theorem~\ref{thm:transrevtrans}).
We emphasize that the reverse translation is made possible by the presence of the type system: an EAM is well-typed if there exists a finite derivation that recursively associates a type with each of its components. This phenomenon suppresses all non-definable elements.

\paragraph{The notion of model and full abstraction.} 
The model of PCF that we construct in this paper is based on AMs and therefore does not fit into the notion of denotational model considered in the literature, which is usually based on category theory and domain theory~\cite{Curien07}. 
Our model is an instance of Milner's set-theoretical definition~\cite{Milner77} and provides a syntax-independent description of PCF terms.
Our full abstraction result is achieved through a quotient that somewhat mimics the observational equivalence. 
This approach is reminiscent of the construction of fully abstract models in categories of games, which also rely on a quotient. 
However, while game models use the quotient solely to achieve completeness, we leverage the same quotient to establish both soundness and completeness.
The interest of this kind of results mostly lies in the \emph{definability property} that does not hold automatically in the quotiented model.\footnote{Even in our case, in the untyped world, there are EAMs that are not definable by any PCF program. 
The machines $\mM_n$ introduced in Examples~\ref{rem:sillyEAMs}\eqref{rem:sillyEAMs2} are witnesses of this situation.} 
We believe that our approach, based on a notion of deterministic machine, brings an original, computationally-oriented, perspective in \PCF{} semantics: 
e.g., in our setting no limit construction is required to handle the fixed point combinator since an EAM accessing its own address is easily definable from a mathematical perspective, and can be seen as an abstract view of the usual implementation of recursion. 

% !TEX root = ../LMCS.tex
%!TEX spellcheck = en-US
\paragraph{Related and future works.} A preliminary version of addressing machines was introduced in \cite{DellaPennaTh} to model computation as communication between distinguished processes by means of their addresses. 
They were subsequently refined in~\cite{DellaPennaIM21} with the theoretical purpose of constructing a model of \lam-calculus. 

Compared to other machine-based formalisms, addressing machines occupy a unique space with a fair amount of individually overlapping features, but distinct overall. Many abstract machines feature subroutines which can be called and returned from. Traditional abstract machines (KAM \cite{Krivine07}, SECD \cite{Landin64}, TIM~\cite{FairbairnW87}, Lazy KAM~\cite{Cregu91,Lang07}), ZINC~\cite{Leroy90} etc.) employ a stack to temporarily store the current state of the program during the execution of a subroutine. AMs adopt a different approach, utilizing separate AMs for each subroutine and an external address table map to allow communication among them. Evaluating a subroutine is done simply by executing the corresponding AM. Thus, to perform meaningful computations, it is necessary to consider the whole set of AMs, rather than a single machine. 
By delegating the complexity of subroutines to the external map, AMs are difficult to compare with other formalisms in terms of implementational complexity and efficiency.
Due to space limitations, any implementation can represent only a finite number of AMs and relies on a finite address table map, which is therefore decidable.
From a theoretical perspective, we believe that a translation from AMs into (multiple) traditional abstract machines could be defined without much difficulty.
Additionally, we emphasize that computational efficiency can be enhanced by optimizing the finite address table map through the use of suitable hash tables. Addressing this optimization challenge is a goal we plan to explore in future work.

The presence of explicit substitutions may also remind of more ``modern'' abstract machines arising from the linear substitution calculus, namely the MAM~\cite{Accattoli14,Accattoli17} and related machines. The MAM also avoids the concept of a closure, through use of $\alpha$-renaming all environments which would be local to a subroutine can instead be stored in a single global environment. One could argue that this is a special method of implementing subroutines, while AMs retain a local environment and sidestep the issue entirely.

The communication aspect of addressing machines may also remind the reader of formalisms such as Interactive Turing Machines~\cite{ITM,AdvancedITM} and the $\pi$-calculus ~\cite{Milner99,Sangiorgi01} which both feature communication between processes. Both of these formalisms use the communication aspect for the purpose of parallel/concurrent computations. By design, AMs do not support concurrent computations and instead model asymmetric communication rather than bilateral communication, so the similarity ends there.

Compared with models of \PCF{} based on Scott-continuous functions~\cite{Milner77,BerryCL85,Curien07}, EAMs provide a more operational interpretation of a program and naturally avoid parallel features that would lead to the failure of FA as in the continuous semantics. 
Compared with Curien's sequential algorithms \cite{Curien92} and categories of games~\cite{AbramskyMJ94,Nickau94,HylandO00} they share the intensionality of the denotations of the programs, while presenting an original way of modelling sequential computation.
The model based on AMs also bares \emph{some} similarities with the categories of assembly used to model PCF \cite{LongleyTh}, mostly on a philosophical level, in the sense that these models are based on the `codes' (rather than addresses) of recursive functions realizing a formula ($\cong$ type). 

A disclaimer concerning explicit substitutions~\cite{AbadiCCL90,AbadiCCL91,CurienHL96,SeamanI96,LevyM99}: since the full abstraction property is defined on \emph{closed} \PCF{} terms,  we assume a closedness condition on explicit substitutions, as in~\cite{AlvesFFM07}, to reduce the level of technicalities. We do not claim, or believe, that our presentation of \EPCF{} gives a satisfying theory of explicit substitution for `open' \PCF. 
In future works, we plan to build from Accattoli and Kesner's works~\cite{AccattoliK12,Accattoli18} to develop a version of \EPCF{} at the correct level of generality.
\section{PCF with Explicit Substitutions}\label{sec:PCF}
% !TEX root = ../LMCS.tex
%!TEX spellcheck = en-US

We begin by describing the syntax and operational semantics of Plotkin's \PCF{} \cite{Plotkin77,Ong95}.

\begin{defi}\bsub
\item Let us consider fixed a countably infinite set $\Var$ of variables that are denoted $x,y,z,\dots$ The set $\Lam^\PCF$ of \emph{$\PCF$ terms} is defined by induction as follows:
\[
	P,Q,Q'\ ::= x \mid \lam x.P \mid P\cdot Q \mid \fix P\mid \mathbf{0} \mid \pred P \mid \succ P\mid  \ifterm{P}{Q}{Q'} \tag{$\Lam{}^\PCF{}$}
\]
where $\lam x.P$ represents the abstraction, $P\cdot Q$ the application, $\fix$ a fixed point combinator, $\mathbf{0}$ the constant zero, $\pred$ and $\succ$ the predecessor and successor (respectively), and $\mathbf{ifz}$ the conditional test on zero (namely, `is zero?-then-else').
\item A \emph{\PCF{} value} is either a numeral or an abstraction. The \emph{set of \PCF{} values} is defined by:
\[
	\Val[\PCF] = \set{\num n \st n\in\nat}\cup\set{\lam x.P\st P\in\Lam^\PCF},\textrm{ where }\num n = \mathbf{succ}^n(\mathbf{0}).
\]
\item The set $\FV{P}$ of \emph{free variables of $P\in\Lam^{\PCF}$} and \emph{$\alpha$-conversion} are defined as usual.
\item A term $P\in\Lam^\PCF$ is called a \emph{\PCF{} program} if it is closed, i.e.\ $\FV{P} = \emptyset$.
\item Given $P,Q\in\Lam^\PCF$, we denote by $P\subst{x}{Q}$ the \emph{capture-free substitution} of $Q$ for all free occurrences of $x$ in $P$.
\esub
\end{defi}
\noindent %FIXED indent
Hereafter, \PCF{} terms are considered up to $\alpha$-conversion. 
We now introduce \PCF{} `contexts', namely \PCF{} terms possibly containing occurrences of an algebraic variable, which is traditionally called \emph{hole}, and here denoted by $\square$.

\begin{defi}\label{def:contexts}
\bsub
\item\label{def:contexts1} \emph{\PCF{} contexts} $\C\square$ are generated by the following grammar:
\[
	\C\square\ ::= \square \mid x \mid \lam x.\C \mid \C_1\cdot \C_2 \mid \fix \C\mid \mathbf{0} \mid \pred \C \mid \succ \C\mid  \ifterm{\C}{\C_1}{C_2} 
\]
\item Given a $\PCF$ context $\C\square$ and $P\in\Lam^\PCF$, $\C[P]$ stands for the \PCF{} term obtained by substituting $P$ for all occurrences of $\square$ in $\C$, possibly with capture of free variables in $P$.
\item\label{def:contexts3} \PCF{} \emph{evaluation contexts} are \PCF{} contexts generated by the grammar (for $P,Q\in\Lam^{\PCF}$):
\[
	\Ev\square ::= \square \mid \Ev \cdot P \mid \pred \Ev \mid \succ \Ev \mid \ifterm {\Ev} {P}{Q}
\]

\esub
\end{defi}
\noindent %FIXED indent
We introduce the `small step' and `big step' call-by-name operational semantics of $\PCF.$
\begin{defi}
\bsub
\item The \emph{weak head reduction} $\to_\PCF\ \subseteq\Lam^\PCF\times\Lam^\PCF$ is defined as follows:
\[
	\bar{ll}
	\Ev[(\lam x.P)Q]\to_\PCF \Ev[P\subst{x}{Q}],
	&	
	\Ev[\fix(P)]\to_\PCF \Ev[P\cdot (\fix(P))],\\[3pt]
	\Ev[\pred(\succ(\num n))]\to_\PCF \Ev[\num n],
	&
	\Ev[\pred(\mathbf{0})]\to_\PCF \Ev[\mathbf{0}],	
	\\[3pt]
	\Ev[\ifterm{\mathbf{0}}{P_1}{P_2}]\to_\PCF \Ev[P_1],
	&
	\Ev[\ifterm{\num{n+1}}{P_1}{P_2}]\to_\PCF \Ev[P_2],
	\ear
\]
for all evaluation contexts $\E\square$.
\item The \emph{multistep reduction} $\redd[\PCF]$ is defined as the reflexive and transitive closure of $\to_{\PCF}$.
\item The \emph{big step reduction} $\Downarrow\ \subseteq\, \Lam^\PCF\times\Lam^\PCF$ takes a \PCF{} term to a \PCF{} value:
\[
	\bar{ccc}
	\infer[\PCFvalrule]{U\Downarrow U}{ U\in \Val}
	&
	\infer[\predzrule]{\pred P \Downarrow \mathbf{0}}{P \Downarrow \mathbf{0}}	
	&
	\infer[\predrule]{\pred P \Downarrow \num n}{P \Downarrow \num{n+1}}	
	\\[1ex]
	\infer[\ifzzrule]{\ifterm P{Q}{Q'} \Downarrow U_1}{P\Downarrow \mathbf{0} & Q\Downarrow U_1 }
	&\qquad
	\infer[\ifzrule]{\ifterm P{Q}{Q'} \Downarrow U_2}{P\Downarrow \num{n+1} & Q'\Downarrow U_2 }\qquad
	&
	\infer[\succrule]{\succ P \Downarrow \num {n+1}}{P \Downarrow \num{n}}	
	\\[1ex]
	\infer[\fixrule]{\fix P \Downarrow U}{P \cdot (\fix P) \Downarrow U}
	&
	\infer[\betarule]{P\cdot Q \Downarrow U}{P\Downarrow \lam x.P' & P'\subst{x}{Q}\Downarrow U}
	\ear
\]
\esub
\end{defi}
\noindent %FIXED indent
It is well-known that the big-step and small-step semantics of \PCF{} are equivalent on \PCF{} programs. 
For a proof of this equivalence, see e.g.~\cite{Ong95}. 

\begin{exas}\label{ex:PCFterms} 
The following \PCF{} programs will be used as running examples:
\bsub
\item $\comb{I} = \lam x.x$, representing the identity.
\item $\Om = \fix(\comb{I})$, representing the paradigmatic looping program.
\item $\mathbf{succ1} = \lam x.\succ x$, representing the successor function.
\item $\mathbf{succ2}=(\lam sn.s\cdot(s\cdot n))\cdot \mathbf{succ1}$, representing the function $f(x) = x+2$.
\esub
\end{exas}

\begin{defi}\label{def:simpletypes}\bsub
\item\label{def:simpletypes1} The set $\Types$ of \emph{(simple) types} over a \emph{ground type} $\tint$ is defined by:
\begin{equation}\tag{$\Types$}
	\alpha,\beta\ ::=\ \tint \mid \alpha\to \beta
\end{equation}
The arrow operator associates to the right, in other words we write $\alpha_1\to\cdots\to\alpha_n\to\beta$ for $\alpha_1\to(\cdots(\alpha_n\to\beta)\cdots)$ ($= \vec\alpha\to\beta$, for short). If $n=0$ then $\alpha_1\to\cdots\to\alpha_n\to\beta = \beta$.
\item\label{def:simpletypes2} A \emph{typing environment} is defined as a finite map $\Gamma : \Var\to\Types$. 
\item
We write  $x_1:\alpha_1,\dots,x_n : \alpha_n$ for the unique environment $\Gamma$ satisfying $\dom(\Gamma) = \set{x_1,\dots,x_n}$ and $\Gamma(x_i) = \alpha_i$, for all~$i\in\set{1,\dots,n}$.
\item \emph{Typing judgements} are triples, denoted $\Gamma\vdash P : \alpha$, where $\Gamma$ is a typing environment, $P$ is a \PCF{} term and $\alpha\in\Types$.
\item \emph{Typing derivations} are finite trees built bottom-up in such a way that the root has shape $\Gamma\vdash P : \alpha$ and every node is an instance of a rule from Figure~\ref{fig:typing}. In the rule $(\to_\mathrm{I})$ we assume wlog that $x\notin\dom(\Gamma)$, by $\alpha$-conversion.
\item 
When writing $\Gamma\vdash P : \alpha$, we mean that this typing judgement is derivable.
\item We say that $P$ \emph{is typable} if $\Gamma\vdash P : \alpha$ is derivable for some $\Gamma,\alpha$.
\esub
\end{defi}
\begin{figure*}[t]
$\arraycolsep=20pt\def\arraystretch{2.2}
\bar{c}
\bar{ccc}
\infer[(0)]{\Gamma\vdash\mathbf{0}:\tint}{}&
\infer[(\mathrm{ax})]{\Gamma,x:\alpha\vdash x:\alpha}{}&
\infer[(\mathrm{Y})]{\Gamma\vdash \fix P:\alpha}{\Gamma\vdash P : \alpha\rightarrow \alpha}\\
\infer[(-)]{\Gamma\vdash\pred P:\tint}{\Gamma\vdash P:\tint}&\infer[(+)]{\Gamma\vdash\succ P:\tint}{\Gamma\vdash P:\tint}&
{\infer[(\to_\mathrm{I})]{\Gamma\vdash\lambda x.P :\alpha\to\beta}{\Gamma,x:\alpha\vdash P : \beta}}
\ear\\
\bar{cc}
\infer[(\to_\mathrm{E})]{\Gamma\vdash P\cdot Q:\beta}{\Gamma\vdash P:\alpha\rightarrow \beta &\Gamma\vdash Q:\alpha}&
\infer[(\mathrm{ifz})]{\Gamma\vdash \ifterm {P}{Q}{Q'}:\alpha}{
	\Gamma\vdash P:\tint
	&
	\Gamma\vdash Q :\alpha
	&
	\Gamma\vdash Q' :\alpha
}
\ear
\ear
$
\caption{\PCF{} type assignment system.}\label{fig:typing}
\end{figure*}

\begin{exas} The following are examples of derivable typing judgments.
\bsub
\item $\vdash \lam x.x : \alpha \rightarrow \alpha$, for all $\alpha \in \Types$.
\item $\vdash (\lam sn.s\cdot (s\cdot n))\cdot(\lam x.\succ x):\tint\rightarrow\tint$.
\item $\vdash \fix (\lam f x y.\ifterm y {x} {f \cdot(\succ x)\cdot (\pred y)}):\tint\rightarrow\tint\rightarrow\tint$.
\item $\vdash \Om : \alpha$, for all $\alpha\in\Types$. In particular, we have $\vdash \Om :\tint$.\pagebreak
\esub
\end{exas}

\subsection{\PCF{} with explicit substitutions}

We introduce the language \EPCF, a version of \PCF{} endowed with (closed) explicit substitutions.

\begin{defi}\bsub\item The set $\Lame$ of \EPCF{} \emph{terms} is defined by the grammar (for $x\in\Var$):
	\[
		\bar{lcl}
		L,M,N &::=&x\mid M\cdot N \mid \lam x.M \mid \fix M \mid M\esubst{x}{N} \mid\\
		&&\phantom{x}\mid \mathbf{0} \mid \pred M\mid \succ M\mid \ifterm L{M}{N}
		\ear\tag{$\Lame$}
	\]
	The only change compared to standard \PCF{} is the addition of a constructor $M\esubst{x}{N}$ representing the term $M$ receiving the \emph{explicit substitution} $\esubst{x}{N}$ of $x$ by $N$.
	Notice that, while $M\subst{x}{N}$ is a meta-notation, explicit substitutions are actual constructors.
\item When considering a (possibly empty) list of explicit substitutions $\sigma \!=\! \esubst{x_1}{N_1}\cdots \esubst{x_n}{N_n}$, we assume that the variables $\vec x$ in $\sigma$ are pairwise distinguished. \noindent %FIXED indent Fixed Spacing for margin
	Given $\sigma$ as above, we let $\dom(\sigma) = \set{x_1,\dots,x_n}$ be the \emph{domain of $\sigma$} and write $\sigma(x_i) = N_i$, for $x_i\in\dom(\sigma)$.
	\item Given $M\in\Lame$, we write $M^\sigma$ for $M\esubst{x_1}{N_1}\cdots \esubst{x_n}{N_n} = (\cdots(M\esubst{x_1}{N_1})\cdots) \esubst{x_n}{N_n}$.
	\item The set $\Val[\E]$ of \EPCF{} \emph{values} contains abstractions under substitutions and numerals, i.e.
\[
	\Val[\E] = \set{\num n \st n\in\nat}\cup \set{ (\lam x.M)^\sigma \st M\in\Lame}
\]
	\item Given $M\in\Lame$, the set $\FV{M}$ of all \emph{free variables of $M$} and \emph{$\alpha$-conversion} are defined by induction on $M$ as expected, namely with the following additional cases (for $y$ fresh):
	\[
		\FV{M\esubst{x}{N}} = (\FV{M} - \{x\}) \cup \FV{N},\qquad\quad
		M\esubst{x}{N} =_\alpha M\subst{x}{y}\esubst{y}{N},
	\]
	where $y$ is a fresh variable and $M\subst{x}{L}$ denotes the \emph{capture-free substitution} of a term $L$ (in the case above, one takes $L = y$) for all free occurrences of $x$ in $M$.
\item An \EPCF{} term $M$ is called \emph{an \EPCF{} program}, written $M\in\Prog{\E}$, if it is closed 
and all its explicit substitutions are of shape $\esubst{x}{N}$ for some  $x\in\Var$ and $N$ closed.
\item \EPCF{} \emph{evaluation contexts} are defined identically to \PCF{} evaluation contexts (Def.~\ref{def:contexts}\eqref{def:contexts1}), but with $P,Q$ ranging over $\Lame$.
Similarly, given $M\in\Lame$ and an evaluation context $\Ev\square$, we write $\Ev[M]$ for the \EPCF{} term obtained by substituting $M$ for the hole $\square$ in $\Ev\square$.
\esub
\end{defi}
\noindent %FIXED indent
Notice that, in an \EPCF{} program, all subterms of the form $M\esubst{x}{N}$ must have $N\in\Prog{\E}$.
Clearly $\Lambda^\PCF\subsetneq\Lame$, moreover all \PCF{} programs belong to $\Prog{\E}$.
In this paper we focus on ({\sf E})\PCF{} programs and adopt Barendregt's \emph{variable convention} stating that bound variables are given maximally distinguished names.

\begin{exas}\label{ex:EPCFterms} 
All \PCF{} terms introduced previously are also \EPCF{} terms. Some examples of \EPCF{} programs that are not \PCF{} terms are $x\esubst{x}{\num 1}$ and $\lambda x.(\ifterm{x}{y}{z}\esubst{y}{\num 1}\esubst{z}{ \num2})$. 
\end{exas}

We endow \EPCF{} with a small-step call-by-name operational semantics capturing weak head reduction.
The reduction is `weak' since it does not reduce under abstractions.
\begin{defi}\label{def:rewriting}
\bsub

\item\label{def:rewriting11} The \emph{computation reduction} $\redcr$ on \EPCF{} terms is defined as:
\[
	\bar{r@{\hskip 3pt}c@{\hskip 3pt}lcr@{\hskip 3pt}c@{\hskip 3pt}l} %FIXED Reduced Spacing in array inline with below
	\Ev[(\lam x.M)^\sigma N]&\redcr& \Ev[M^\sigma\esubst{x}{N}],&&\Ev[\pred \mathbf{0}] &\redcr& \Ev[\mathbf{0}],\\
\Ev[	\ifterm{\mathbf{0}} M N] &\redcr& \Ev[M],&&\Ev[\pred(\num{n+1})]&\redcr& \Ev[\num n],\\
	\Ev[\ifterm{\num {n+1}} M N] &\redcr& \Ev[N],&&\Ev[\fix(M)]&\redcr&\Ev[M(\fix(M))].\\
	\ear
\]
\item\label{def:rewriting12} The \emph{percolation reduction} $\redsc$ on \EPCF{} terms is defined as (where $\sigma$ is non-empty):
\begin{minipage}{\textwidth}
\[
	\bar{r@{\hskip 3pt}c@{\hskip 3pt}lcr@{\hskip 3pt}c@{\hskip 3pt}l} %\bar{rclcrcl} Reduced spacing and put it into minipage to maintain
	\Ev[x^\sigma]&\redsc&\Ev[N], \textrm{ if } \sigma(x) =N,&&\Ev[\mathbf{0}^\sigma]&\redsc&\Ev[\mathbf{0}],\\
	\Ev[y^\sigma]&\redsc&\Ev[y],\ \textrm{ if }y\notin\dom(\sigma),&&\Ev[(M\cdot N)^\sigma]&\redsc&\Ev[M^\sigma\cdot  N^\sigma],\\	\Ev[(\pred(M))^\sigma]&\redsc&\Ev[\pred(M^\sigma)],&&\Ev[(\succ(M))^\sigma]&\redsc&\Ev[\succ(M^\sigma)],\\
	\Ev[(\ifterm L M N)^\sigma] &\redsc&\Ev[ \ifterm{L^\sigma}{M^\sigma}{N^\sigma}],&&\Ev[(\fix(M))^\sigma]&\redsc&\Ev[\fix(M^\sigma)].\\
	\ear
\]
\end{minipage}
\item\label{def:rewriting2} 
	The (one step) \emph{weak head (w.h.) reduction} $\redwh$ is defined as the union of $\redcr$ and $\redsc$.
\item As it is customary, we denote by $\reddwh$ the transitive-reflexive closure of $\redwh$.
\item We define $\conv[\wh]$ as the symmetric, transitive and reflexive closure of $\redwh$.
\esub
\end{defi}

\begin{rems}\label{rem:EPCFprops}
\bsub
\item The reduction $\redwh$ is designed to work on \EPCF{} programs. 
On closed terms that are not programs, we might obtain unsound reductions like the following:
\[
(\lambda x.(y\esubst{y}{x}))\mathbf{0} \redwh y\esubst{y}{x}\esubst{x}{\mathbf{0}} \redwh x
\]
\item The set $\Prog{\E}$ is closed under $\redwh$, hence under $\reddwh$ as well.
\item\label{rem:EPCFprops3} 
	The w.h.\ reduction is deterministic: $N_1\,{}_\wh\!\!\leftarrow M\redwh N_2$ implies $N_1=N_2$.
\item\label{rem:EPCFprops4} 
	An \EPCF{} program $M$ is w.h.-normalizing if and only if $M\reddwh V$, for some $V\in\Val[\E]$.
\esub
\end{rems}
\noindent %FIXED indent
Since $\redwh$ is deterministic, we can safely write $\len{M\reddwh N}$ to denote the length $n$ of the unique reduction sequence $M = M_1\redwh M_2\redwh\cdots \redwh M_n = N$.

\begin{defi}\label{def:epcfbigstep} From $\redwh$ and mirroring the big-step reduction of \PCF{}, we can derive a big-step reduction $\ToE\ \subseteq \Lam^\EPCF\times\Val[\E]$ relating an \EPCF{} term with an \EPCF{} value:
\[
\bar{c}
	\bar{ccc}
	\infer[\Enumrule]{(\num n)^\sigma\ToE \num n}{ \num n \in \Nat}
	&
	\infer[\Elamrule]{(\lam x. M)^{\sigma}\ToE (\lam x. M)^\sigma}{}
	&
	\infer[\Evarrule]{x^\sigma\ToE V}{\sigma(x) = N & N \ToE V}
	\\[1ex]
	\infer[\Epredzrule]{(\pred M)^\sigma \ToE \mathbf{0}}{M^\sigma \ToE \mathbf{0}}	
	&
	\infer[\Epredrule]{(\pred M)^\sigma \ToE \num n}{M^\sigma \ToE \num{n+1}}	
	&
	\infer[\Esuccrule]{(\succ M)^\sigma \ToE \num {n+1}}{M^\sigma \ToE \num{n}}	
	\\[1ex]
	\ear\\
	\bar{cc}
	\infer[\Eifzzrule]{(\ifterm L{M}{N})^\sigma \ToE V_1}{L^\sigma\ToE \mathbf{0} & M^\sigma\ToE V_1 }
	&
	\infer[\Eifzrule]{(\ifterm L{M}{N})^\sigma \ToE V_2}{L^\sigma\ToE \num{n+1} & N^\sigma\ToE V_2 }\qquad
	\\[1ex]
	\infer[\Efixrule]{(\fix M)^\sigma \ToE V}{M^\sigma \cdot \fix (M^\sigma) \ToE V}
	&
	\infer[\Ebetarule]{(M\cdot N)^\sigma \ToE V}{M^\sigma\ToE (\lam x.M')^{\sigma'} & (M')^{\sigma'}\esubst{x}{N^\sigma}\ToE V}
	\ear
	\ear
\]
\end{defi}

\begin{prop}\label{prop:epcfbigsmallstep}
Given an \EPCF{} program $M$ and an \EPCF{} value $V$, we have
\[
	M\ToE V	\iff M \reddwh V.
\]
\end{prop}
\begin{proof}[Proof (Appendix~\ref{app:regPCF})]
$(\Rightarrow)$ By induction on the height of a derivation of $M\ToE V$.

$(\Leftarrow)$ By induction on the length $\len{M \reddwh V}$.
\end{proof}

\EPCF{} terms can be typed in a manner very similar to \PCF{} terms. 

\begin{defi}
\bsub\item
An \EPCF{} typing judgement is, like in \PCF{}, a triple of $\Gamma\vdash M : \alpha$. 
\item The rules for typing an \EPCF{} term have been presented in Figure~\ref{fig:etyping}. 
\esub
\end{defi}
\noindent %FIXED indent
The only change compared to \PCF{} typing rules is the introduction of the $(\sigma)$ rule. 
In this rule, we rely on the fact that in a program of shape $M\esubst{x}{N}$, the subterm $N$ must be closed and remove the typing environment in $\vdash N : \beta$. This is not standard when considering more general notions of explicit substitutions, but simplifies our definitions later.
\begin{rem}
This approach of forcing explicit substitutions to be closed also sidesteps the issues caused by Melli\`{e}s' degeneracy~\cite{Mellies95}. 
\end{rem}
\begin{exa}
The running examples can be typed without the use of the $(\sigma)$ rule, as they are \PCF{} terms. For an \EPCF{} exclusive term, as an example, we type $x\esubst{x}{\succ\mathbf{0}}$:
\[
\infer[(\sigma)]{\vdash x\esubst{x}{\succ\mathbf{0}} : \tint}{
\infer[(\mathrm{ax})]{x : \tint \vdash x:\tint}{}
& \infer[(+)]{\vdash \succ \mathbf{0} : \tint}{
\infer[(0)]{\vdash \mathbf{0}:\tint}{}}
}
\]
\end{exa}
\begin{figure*}[t]
$\arraycolsep=4pt\def\arraystretch{2.2}
\bar{ccc}
\infer[(0)]{\Gamma\vdash\mathbf{0}:\tint}{}&
\infer[(\to_\mathrm{E})]{\Gamma\vdash M\cdot N:\beta}{\Gamma\vdash M:\alpha\rightarrow \beta &\Gamma\vdash N:\alpha}&
\infer[(\mathrm{ax})]{\Gamma,x:\alpha\vdash x:\alpha}{}\\
\infer[(+)]{\Gamma\vdash\succ M:\tint}{\Gamma\vdash M:\tint}&
\infer[(\sigma)]{\Gamma\vdash M\esubst{x}{N} : \alpha}{
	\Gamma, x : \beta\vdash M : \alpha
	&
	\vdash N : \beta
}&
\mathclap{\infer[(\to_\mathrm{I})]{\Gamma\vdash\lambda x.M :\alpha\to\beta}{\Gamma,x:\alpha\vdash M : \beta}}\\
\infer[(-)]{\Gamma\vdash\pred M:\tint}{\Gamma\vdash M:\tint}
&
\infer[(\mathrm{ifz})]{\Gamma\vdash \ifterm LMN:\alpha}{
	\Gamma\vdash L:\tint
	&
	\Gamma\vdash M:\alpha
	&
	\Gamma\vdash N:\alpha
}&
\infer[(\mathrm{Y})]{\Gamma\vdash \fix M:\alpha}{\Gamma\vdash M : \alpha\rightarrow \alpha}\\
\ear
$
\caption{\EPCF{} type assignment system. In the rule $(\sigma)$ we enforce $N$ to be closed.}\label{fig:etyping}
\end{figure*}

The following lemma summarizes the main properties of the type assignment system.
\begin{lem}\label{lem:epcftyping} Let $M\in\Lame$, $\alpha,\beta\in\Types$ and $\Gamma$ be a typing environment.
\bsub
\item \label{lem:epcftyping1} (Syntax directedness) Every derivable judgement $\Gamma\vdash M : \alpha$ admits a unique derivation.
\item \label{lem:epcftyping2} (Strengthening) $\Gamma,x:\beta\vdash M : \alpha$ and $x\notin\FV{M}$ hold if and only if $\Gamma\vdash M : \alpha$ does. Thus, an \EPCF{} program $M$ is typable if it is typable in the empty environment.
\item \label{lem:epcftyping4} (Subject reduction) 
For all $M\in\Prog{\E}$, $\vdash M : \alpha$ and $M \redwh N$ entail $\vdash N : \alpha$.
\esub
\end{lem}
\noindent %FIXED indent
We now move on to proving that \PCF{} and \EPCF{} operational semantics coincide on closed terms of type $\tint$. To this purpose, we first introduce the \emph{collapse} $\uf M$ of an \EPCF{} term $M$ defined by performing all of the internal explicit substitutions.

\begin{defi}
\bsub
\item Given an \EPCF{} term $M$, define a \PCF{} term $\uf M\in\Lam^\PCF$ as follows:
\[
	\bar{rclcrcl}
	\uf{x}&=& x &&	\uf{\mathbf{0}}&=& \mathbf{0} \\
	\uf{(M\cdot N)} &=& \uf{M}\cdot \uf{N} &&	\uf{(\pred M)}&=& \pred\uf{M} \\
	\uf{(\lam x.M)}&=& \lam x.\uf M&&\uf{(\succ M)}&=& \succ\uf{M} \\
	\uf{(\fix M)}&=& \fix(\uf M) &\qquad&\uf{(\ifterm L{M}{N})}&=& \ifterm {\uf L}{\uf M}{\uf N}\\
	\uf{(M\esubst{x}{N})}&=& \uf M\subst{x}{\uf N}\\
	\ear
\]
\item The \emph{head size} $\headsize{-} : \Lame\to\nat$ of an \EPCF{} term is defined as follows:
\[
\bar{rclcrcl}
\headsize{x}=\headsize{\mathbf{0}}&=& 1 & & \headsize{M\esubst{x}{N}}&=&\headsize{M}\cdot (\headsize{N} + 1)\\
\headsize{M\cdot N} &=& \headsize{M}+1 && \headsize{\lam x.M} &=& \headsize{M}+1\\\headsize{\pred M}&=&\headsize{M} + 1&&\headsize{\fix M} &=&  \headsize{M}+1 \\ \headsize{\succ M}&=&\headsize{M}+1&&\headsize{\ifterm L{M}{N}}&=&\headsize{L}+1
\ear
\]
\item The map $\headsize{-}$ is extended to explicit substitutions $\sigma = \esubst{x_1}{N_1}\cdots \esubst{x_n}{N_n}$, by setting 
\[
	\headsize{\sigma} = \prod_{i=1}^n(\headsize{N_i} + 1).
\]
\esub
\end{defi}
\noindent %FIXED indent
Recall that $\Prog{\E}$ denotes the set of $\EPCF$ programs, i.e.\ the set of closed $\EPCF$ terms $M$ whose explicit substitutions (if any) are of the form $\esubst{x}{N}$ for $N$ closed.

\begin{lem}\label{lem:trivial} 
\bsub\item\label{lem:trivial1} If $M\in\Prog{\E}$ then $\uf M$ is a \PCF{} program.
\item\label{lem:trivial2} If $P\in\Lam^\PCF$ then $\uf P = P$.
\esub
\end{lem}

\begin{proof} \bsub\item By structural induction on $M$.
\item By structural induction on $P$.\qedhere
\esub
\end{proof}
\noindent %FIXED indent
For a proof of the following proposition, we refer to the technical Appendix~\ref{app:regPCF}.\pagebreak

\begin{prop}\label{prop:MreddNthenMreddN}
\bsub
\item Let $M,N\in\Prog{\E}$ be such that $M \redcr N$.
Then $\uf M \to_\PCF \uf N$.\label{prop:McrNthenMredN} 
\item Let $M,N\in\Prog{\E}$ be such that $M \redsc N$.
Then $\uf M = \uf N$.\label{prop:MscNthenMisN} 
\esub
\end{prop}

\begin{cor}\label{cor:MreddwhthenreddPCF} 
Let $M\in\Prog{\E}$ be such that $M\reddwh \num n$. Then $\uf{M} \redd[\PCF]\num n$.
\end{cor}

\begin{proof} By induction on the length  of the reduction sequence $\ell = \len{M \reddwh \num n}$.

Case $\ell = 0$. Then $M = \num n = \uf M$, so this case follows by reflexivity of $\redd[\PCF]$.

Case $\ell > 0$. Then there exists $N\in\Prog{\E}$ such that $M\to_\wh N\reddwh \num n$ where $\len{N\reddwh \num n}<\ell$.
By Proposition~\ref{prop:MreddNthenMreddN}, we have $\uf M \redd[\PCF] \uf N$. 
By induction hypothesis, we obtain $\uf N\redd[\PCF] \num n$.
By transitivity, we conclude $\uf{M} \redd[\PCF]\num n$.
\end{proof}

\begin{lem}\label{lem:secSN}
The percolation reduction $\redsc$ on \EPCF{} terms is strongly normalizing. More precisely:
\[
	M \redsc N\imp \headsize{M} > \headsize{N}
\]
\end{lem}

\begin{proof}
By induction on a derivation of $M\redsc N$. In the following we consider a non-empty list of explicit substitutions $\sigma = \esubst{x_1}{N_1}\cdots \esubst{x_n}{N_n}$, i.e.\ $n > 0$.
It follows that $\headsize{\sigma} > 1$.
\begin{itemize}
\item Base cases.
\begin{itemize}
\item Case $M = x_i^\sigma$ and $N = \sigma(x_i) = N_i$. Then $\headsize{x_i^\sigma} = 1\cdot\headsize{\sigma} > \headsize{N_i} = \headsize{N}$.
\item Case $M = \mathbf{0}^\sigma$ and $N = \mathbf{0}$. Then $\headsize{\mathbf{0}^\sigma} = 1\cdot\headsize{\sigma} > 1 = \headsize{\mathbf{0}}$.
\item Case $M = y^\sigma$, with $y \notin\dom(\sigma)$, and $N = y$. Then $\headsize{y^\sigma} =1\cdot\headsize{\sigma} > 1 = \headsize{y}$.
\item Case $M = (M_1\cdot M_2)^\sigma$. Then $\headsize{(M_1\cdot M_2)^\sigma} = (\headsize{M_1} + 1)\cdot \headsize{\sigma} > \headsize{M_1}\cdot \headsize{\sigma} + 1 = \headsize{M_1^\sigma \cdot M_2^\sigma}$.
\item Case $M = (\pred M')^\sigma$. Then $\headsize{(\pred M')^\sigma} = (\headsize{M'} + 1)\cdot \headsize{\sigma} > \headsize{M'}\cdot \headsize{\sigma} + 1 = \headsize{\pred (M')^\sigma}$.
\item Case $M = (\succ M')^\sigma$. Then $\headsize{(\succ M')^\sigma} = (\headsize{M'} + 1)\cdot \headsize{\sigma} > \headsize{M'}\cdot \headsize{\sigma} + 1 = \headsize{\succ (M')^\sigma}$.
\item Case $M = (\ifterm {M_1}{M_2}{M_3})^\sigma$. Then $\headsize{(\ifterm {M_1}{M_2}{M_3})^\sigma} = (\headsize{M_1} + 1)\cdot \headsize{\sigma} > \headsize{M_1}\cdot \headsize{\sigma} + 1 = \headsize{\ifterm {M_1^\sigma}{M_2^\sigma}{M_3^\sigma}}$.
\item Case $M = (\fix M')^\sigma$. Then $\headsize{(\fix M')^\sigma} = (\headsize{M'} + 1)\cdot \headsize{\sigma} > \headsize{M'}\cdot \headsize{\sigma} + 1 = \headsize{\fix (M')^\sigma}$.
\end{itemize}
\item Induction case $M = M_1\cdot M_2$ and $N = N_1\cdot M_2$ with $M_1\redsc N_1$. By induction hypothesis, we have $\headsize{M_1} > \headsize{N_1}$. Therefore $\headsize{M_1}+1 > \headsize{N_1}+1$.
%\item Case $M=\ifterm{M_1}{M_2}{M_3}$ and $N=\ifterm{N_1}{M_2}{M_3}$ with $M_1\redsc N_1$. Analogous.
%\item Case $M=\pred M'$ and $N=\pred N'$ with $M'\redsc N'$. Analogous.
%\item Case $M=\succ M'$ and $N=\succ N'$ with $M'\redsc N'$. Analogous.\qedhere
\end{itemize}
The remaining cases follow analogously from the induction hypothesis.
\end{proof}

\begin{prop}\label{prop:infinitereduction}
Let $(M_n)_{n\in\nat}$ be a sequence of \EPCF{} programs such that $M_n\!\redwh M_{n+1}$. Then for all $i\in\nat$ there exists an index $j>i$ such that $M_i \reddwh M_j$ and $\uf {M_i} \to_\PCF \uf M_j$.
\end{prop}
\begin{proof} Consider an arbitrary $M_i$. By Lemma~\ref{lem:secSN} $M_i \reddsc M_{k}$, for some $k\ge i$ such that $M_k$ in $\redsc$-normal form.
Therefore $M_k\redwh M_{k+1}$ must be a computation step, i.e.\ $M_k\redcr M_{k+1}$.
By Proposition~\ref{prop:MreddNthenMreddN}, $\uf{M_i} = \uf{M_k}\to_\PCF \uf{M_{k+1}}$. 
Conclude by taking $j = k+1 > i$.
\end{proof}
\begin{thm}\label{thm:eqPCFandEPCF} 
For a \PCF{} program $P$ having type $\tint$, we have:
\[
P\redd[\PCF] \num n\iff P\reddwh \num n
\]
\end{thm}
\begin{proof} $(\Rightarrow)$ We prove the contrapositive. Assume that $P \not\reddwh \num n$. By Subject Reduction (Lemma~\ref{lem:epcftyping}\eqref{lem:epcftyping4}) $P$ must have an infinite $\redwh$ reduction path. By Proposition~\ref{prop:infinitereduction}, $\uf P$ must have an infinite $\to_\PCF$ reduction path. Conclude since, by Lemma~\ref{lem:trivial}\eqref{lem:trivial2}, $\uf P = P$.

$(\Leftarrow)$ Assume $P\reddwh \num n$. Since $P$ is a \PCF{} program it also belongs to $\Prog{\E}$.
By Corollary~\ref{cor:MreddwhthenreddPCF} we have $\uf P\redd[\PCF] \num n$.
Conclude since, by Lemma~\ref{lem:trivial}\eqref{lem:trivial2}, $\uf P = P$.
\end{proof}

\section{EAMs: Extended Addressing Machines}\label{sec:EAMs}
% !TEX root = ../LMCS.tex
%!TEX spellcheck = en-US

We extend the addressing machines from \cite{DellaPennaIM21} with instructions for performing arithmetic operations and conditional testing.
Natural numbers are represented by particular machines playing the role of numerals.

\subsection{Main definitions}

We consider fixed a countably infinite set $\Addrs$ of \emph{addresses} together with a distinguished countable subset $\mathbb{X} \subset \Addrs$, such that $\Addrs-\mathbb{X}$ remains infinite. 
Intuitively, $\mathbb{X}$ is the set of addresses that we reserve for the numerals, therefore hereafter we work under the hypothesis that $\mathbb{X} = \nnat[]$, an assumption that we can make without loss of generality.
\begin{defi} 
\bsub
\item Let $\Null\notin\Addrs$ be a ``null'' constant representing an uninitialised register. 
Set $\Addrs_\Null = \Addrs\cup\set{\Null}$.

\item 
	An \emph{$\Addrs$-valued tape} $T$ is a finite ordered list of addresses $T = [a_1,\dots,a_n]$ with $a_i\in\Addrs$ for all $i\,(1\le i \le n)$. When $\Addrs$ is clear from the context, we simply call $T$ a tape. We denote by $\Tapes$ the set of all $\Addrs$-valued tapes.	
\item
	Let $a\in\Addrs$ and $T,T'\in\Tapes$. We denote by $\Cons a {T}$ the tape having $a$ as first element and $T$ as tail. We write $\appT{T}{T'}$ for the concatenation of $T$ and $T'$, which is an $\Addrs$-valued tape itself.

\item
	Given a natural number $i \ge 0$, an $\Addrs_\Null$-valued \emph{register} $R_i$ is a memory-cell, accessible via its index $i$, and capable of storing either $\Null$ or an address $a\in\Addrs$. We write $\val{R_i}$ to represent the value stored in the register $R_i$. (The notation $\oc R_i$ is borrowed from ML, where $\oc$ represents an explicit dereferencing operator.)
\item
	Given $\Addrs_\Null$-valued registers $R_0,\dots,R_r$ for $r\ge 0$, an address $a\in\Addrs$ and an index $i\ge 0$, we write $\vec R\repl{R_i}{a}$ for the list of registers $\vec R$ where the value of $R_i$ has been updated by setting $\val{R_i} = a$.
Notice that, whenever $i > r$, we assume that the contents of $\vec R$ remains unchanged, i.e.\ $\vec R\repl{R_i}{a} = \vec R$.
\esub
\end{defi}
\noindent %FIXED indent
Intuitively, the contents of the registers $R_0,\dots,R_r$ constitute the \emph{state} of a machine, while the tape correspond to the list of its inputs.
The addressing machines from \cite{DellaPennaIM21} are endowed with only three instructions ($i,j,k,l$ range over indices of registers):
\begin{enumerate}[1.]
\item $\Load i$ : If the tape is non-empty, `pops' the address $a$ from the input tape $\Cons a {T}$ and stores $a$ in the register $R_i$.
If the tape is empty then the machine halts its execution.
\item $\Apply ijk$ : reads the addresses $a_1,a_2$ from $R_i$ and $R_j$ respectively, and stores in $R_k$ the address of the machine obtained by extending the tape of the machine of address $a_1$ with the address $a_2$. The resulting address is not calculated internally but rather obtained calling an external \emph{application map}.
\item $\Call i$ : transfers the computation to the machine having as address the value stored in $R_i$, whose tape is extended with the remainder of the current machine's tape.
\end{enumerate}
As a general principle, writing on a non-existing register does not cause issues as the value is simply discarded---this is in fact the way one can erase an argument. 
Attempts to read an uninitialized register can be avoided statically (see Lemma~\ref{lem:correction}).

We enrich the above set of instructions with arithmetic operations mimicking the ones present in \PCF:
\begin{enumerate}[1.,resume]
\item $\Ifz i j k l$: implements the ``\emph{is zero?}'' test on $!R_i$. Assuming that the value of $R_i$ is an address $n\in\nat$, the instruction stores in $R_l$ the value of $R_j$ or $R_k$, depending on whether $n = 0$.
\item $\Pred i j$: if $\val{R_i} \in\nat$, the value of $R_j$ becomes $\val{R_i}\ominus1 = \max(\val{R_i} - 1, 0)$.
\item $\Succ i j$: if $\val{R_i} \in\nat$, then the value of $R_j$ becomes $\val{R_i}+1$.
\end{enumerate}
Notice that the instructions above need $R_i$ to contain a natural number to perform the corresponding operation. However, they are also supposed to work on addresses of machines that compute a numeral. 
For this reason, the machine whose address is stored in $R_i$ must first be executed, and only if the computation terminates with a numeral is the arithmetic operation performed. 
If the computation terminates in an address not representing a numeral, then the machine halts. We will see that these terminations can be avoided using a type inference algorithm (see Proposition~\ref{prop:typing}, below).

\begin{defi}\label{def:progs}
\bsub
\item \label{def:progs1}
	A \emph{program} $P$ is a finite list of instructions generated by the following grammar, where $\varepsilon$ represents the empty string and $i,j,k,l$ are indices of registers:
	\[
	\begin{array}{lcl}
	\ins{P}&::=&\Load i;\, \ins{P}\mid \ins{A}\\[1ex]
	\ins{A}&::=& \Apply ijk;\, \ins{A}\mid \Ifz ijkl;\, \ins{A}\mid
	\Pred ij;\, \ins{A}\mid\Succ ij;\, \ins{A}\mid \ins{C}\\[1ex]
	\ins{C}&::=&\Call i \mid \varepsilon
	\end{array}
	\]
	Thus, a program starts with a list of $\ins{Load}$'s, continues with a list of $\ins{App}$, $\ins{Test}$, $\ins{Pred}$, $\ins{Succ}$, and possibly ends with a $\ins{Call}$. Each of these lists may be empty, in particular the empty program $\varepsilon$ can be generated.
	\item In a program, we write $\Load (i_1, \dots, i_n)$ as an abbreviation for the instruction sequence $\Load i_1;\,\cdots;\,\Load i_n$. For a given instruction $\ins{ins}$, we write $\ins{ins}^a$ for the repetition of $\ins{ins}$ $a$-many times.
\item\label{def:progs2}
	Let $P$ be a program, $r\ge 0$, and $\cI\subseteq \set{0,\dots,r-1}$ be a set of indices corresponding to the indices of initialized registers. 
	Define the relation $\cI\models^{r} P$, whose intent is to specify that $P$ does not read uninitialized registers, as the least relation closed under the rules:
\[
	\bar{cccc}
		\infer{\cI\models^{r}\varepsilon}{}
		&
		\infer{\cI\models^{r}\Call i}{i\in \cI}
		&
		\infer{\cI\models^{r} \Pred ij;\, \ins{A}}{\cI\cup\set{j}\models^{r}  \ins{A} & i\in \cI& j<r}
		\\[1ex]
		\infer{\cI\models^{r} \Load i;\, \ins{P}}{\cI\cup\set{i}\models^{r}  \ins{P} & i< r}		
		&
		\infer{\cI\models^{r} \Load i;\, \ins{P}}{\cI\models^{r}  \ins{P} & i\ge r}\qquad
		&
	\infer{\cI\models^{r} \Succ ij;\, \ins{A}}{\cI\cup\set{j}\models^{r}  \ins{A} & i\in \cI& j<r}
		\\[1ex]
	\infer{\cI\models^{r} \Apply ijk;\, \ins{A}}{\cI\cup\set{k}\models^{r}  \ins{A} & i,j\in \cI& k<r}	
	&
	\multicolumn{2}{c}{		
	\infer{\cI\models^{r} \Ifz ijkl;\, \ins{A}}{\cI\cup\set{l}\models^{r}  \ins{A} & i,j,k\in \cI& l<r}}
	\ear
\]
\item\label{def:progs3} 
	A program $P$ is \emph{valid with respect to $R_0,\dots,R_{r-1}$} if $\cR\models^{r} P$ holds for
	\[\cR = \set {i\st R_i \neq\Null \und 0\le i < r}.\]
\esub
\end{defi}
\noindent %FIXED indent
Notice that requiring all loads to occur at the beginning of the program is not an actual limitation. 
Once the operational semantics of AMs is discussed, we encourage the reader to verify that any AM whose program contains loads at arbitrary positions can be transformed into an equivalent AM, where all loads are placed at the beginning.

\begin{exas} For each of these programs, we specify its validity w.r.t.\ $R_0 = 7$, $R_1 = a$, $R_2 = \Null$ (i.e., $r = 3$).
\[
\bar{lcr}
P_1 = \Pred 02;\,\Call 2&&\textrm{(valid)}\\
P_2 = \mathrlap{\Load(2,8);\,\Ifz 0120;\, \Call 0}&&\textrm{(valid)}\\
P_3 = \Load(0, 2,8);\, \Call 8&&\textrm{(calling the uninitialized register $R_8$, thus not valid)}\\
P_4 = \Succ 2 0; \Call 1 &&\textrm{(reading from the uninitialized register $R_2$, thus not valid)}\\
P_5 = \Pred 0 8; \Call 0 &&\textrm{(storing in non-existent register $R_0$, thus not valid)}
\ear
\]
\end{exas}

\begin{lem}\label{lem:correction}
Given $\Addrs_\Null$-valued registers $\vec R$ and a program $P$ it is decidable whether $P$ is valid w.r.t.\ $\vec R$.
\end{lem}

\begin{proof} Decidability follows from the syntax directedness of Definition~\ref{def:progs}\eqref{def:progs2}, and the preservation of the invariant $\cI\subseteq\set{0,\dots,r-1}$, since $\cI$ is only extended with $k<r$.
\end{proof}

Hereafter, we will focus on EAMs having valid programs. 

\begin{defi}\label{def:AM}
\bsub
\item
	An \emph{extended addressing machine} (\emph{EAM}) $\mM$ over $\Addrs$ (having $r+1$ registers) is given by a tuple $\mM = \tuple{R_0,\dots,R_{r},P,T}$
where 
\begin{itemize}
	\item $\vec R$ are $\Addrs$-valued registers, 
	\item $P$ is a program valid w.r.t.\ $\vec R$, and
	\item $T\in\Tapes$ is an (input) tape.
\end{itemize}
\item We denote by $\cM$ the set of all extended addressing machines over $\Addrs$.
\item
	Given an EAM $\mM$, we write $\mM.\vec R$ for the list of its registers, $\mM.R_i$ for its $i$-th register, $\mM.P$ for the associated program and $\mM.T$ for its input tape.
\item 
	Given $\mM\in\cM$ and $T'\in\Tapes$, we write $\appT{\mM}{T'}$ for the machine 
	\[
	\tuple{\mM.\vec R,\mM.P,\appT{\mM.T}{T'}}.
	\]
\item For $n\in\nat$, the \emph{$n$-th numeral machine} is defined as $\mach{n}= \tuple{R_0,\varepsilon,[]}$, with $\val{R_0} = n$.
\item\label{def:AM6} 
	 For every $a\in\Addrs$, define the following extended addressing machine:
\[
	\mY^a = \tuple{R_0 = \Null,R_1=\Null ,\Load (0, 1);\Apply  0 1 0;\Apply 101; \Call 1,[a]}.
\] 
\esub
\end{defi} 
\noindent %FIXED indent
We now enter into the details of the addressing mechanism which constitutes the core of this formalism.
\begin{defi}\label{def:bijectivelookup+Y} Recall that $\nat$ stands for an infinite subset of $\Addrs$, here identified with the set of natural numbers, and that $\mY^a$ has been introduced in Definition~\ref{def:AM}(\ref{def:AM6}).
\begin{enumerate}
\item Since $\cM$ is countable, by a simple set-theoretical argument, we can fix a bijective function called an \emph{address table map} (ATM) \[\Lookup\, :\, \cM \to  \Addrs\] satisfying the following conditions:
	\bsub
	\item (Numerals)
	$\forall n \in \mathbb{N}\,.\, \Lookup{\mach{n}} = n$, where $\mach{n}$ is the $n$-th numeral machine;
\item\label{itemB}
	(Fixed point combinator) $\exists a\in\Addrs-\nnat[]\,.\,\Lookup(\mY^{a}) = a$.
	\end{enumerate}
\item\label{def:ymachine}
	We write $\mY$ for the EAM $\mY^a$ satisfying $\Lookup(\mY^a) = a$ (which exists by condition \eqref{itemB} above).	
\item Given $\mM\in\cM$, we call the element $\Lookup \mM$ the {\em address of the EAM $\mM$}. 
	Conversely, the EAM having as address $a\in\Addrs$ is denoted by $\Lookinv{a}$, i.e.\ $\Lookinv{a} = \mM\iff \Lookup\mM = a.$
\item
	Define the \emph{application map} $(\App{}{}) : \Addrs\times\Addrs\to \Addrs$ by setting 
	\[\App{a}{b} = \Lookup (\append{\Lookinv{a}}{b}).\]
	I.e., the \emph{application} of $a$ to $b$ is the unique address $c$ of the EAM obtained by appending $b$ at the end of the input tape of the EAM $\Lookinv{a}$.	  	
\esub
\end{defi}

\begin{rem}\label{rem:blackhole} 
In general, there are uncountably many possible address table maps of arbitrary computational complexity. 
A natural example of such maps is given by \emph{G\"odelization}, which can be performed effectively.
The framework is however more general and allows us to consider non-r.e.\ sets of addresses like the complement $K^c$ of the halting set
\[
	K = \set{(i, x) \st \textrm{the program $i$ terminates when run on input $x$}}
\]
and a non-computable function $\Lookup\,:\, \mathcal{M}_{K^c} \to K^c$ as a map.
\end{rem} 
In an implementation of EAMs the address table map should be computable---one can choose a fresh address from $\Addrs$ whenever a new machine is constructed, save the correspondence in some table and retrieve it in constant time.

\begin{rems}\label{rem:sillyEAMs}
\bsub\item Since $\cM$ and $\Addrs$ are both countable,  the existence of $2^{\aleph_0}$ ATMs follows from cardinality reasons. This includes effective maps as well as non-computable ones.
\item\label{rem:sillyEAMs2} Depending on the chosen address table map, there might exist infinite (static) chains of EAMs, e.g., EAMs  $(\mM_n)_{n\in\nat}$ satisfying
$\mM_n = \tuple{R_0,\varepsilon,[]}$ with $\val R_0 = \Lookup\mM_{n+1}$.
\esub
\end{rems}
\noindent %FIXED indent
The results presented in this paper are independent from the choice of the ATM.

\begin{exas}\label{ex:extabsmach} The following are examples of EAMs (whose registers are assumed uninitialized where unspecified, i.e.\ $\vec R = \vec \Null$). 
\bsub
\item $\mach{I} := \tuple{R_0 = \Null, \Load 0; \Call 0, []}$,
\item For some $a\in\Addrs,$ $\tuple{R_0 = a, R_1 = \Null, \Apply 1 0 0; \Call 1, []}$
\item $\mach{Succ1} := \tuple{R_0, \Load 0; \Succ 0 0; \Call 0, []}$.
\item $\mach{Succ2} := \tuple{R_0, R_1, \Load 0;\Load 1; \Apply 0 1 1; \Apply 0 1 1; \Call 1,[a_{\mach{S}}]}$, where $a_{\mach{S}} = \Lookup{\mach{Succ1}}$.
\item $\mach{Add\_aux} := \tuple{\vec R,P, []}$ with $\mach{Add\_aux}.r = 5$ and $P = \Load (0, 1, 2); \Pred 1 3; \Succ 2 4;\Apply 0 3 0;$ $\Apply 0 4 0; \Ifz 1 2 0 0; \Call 0$.
\esub
\end{exas}

\subsection{Operational semantics}

The operational semantics of extended addressing machines is given through a small-step rewriting system. 
The reduction strategy is deterministic, since the only applicable rule at every step is univocally determined by the first instruction of the internal program, the contents of the registers and the head of the tape.

\begin{defi}\label{def:eamred}
\bsub
\item Define a reduction strategy $\redh$ on EAMs, representing one step of computation, as the least relation $\redh\ \subseteq\cM\times\cM$ closed under the rules in Figure~\ref{fig:am:small_step}.
\item The \emph{multistep reduction} $\reddh$ is defined as the transitive-reflexive closure~of~$\redh$.
\item The \emph{conversion relation} $\convh$ is the transitive-reflexive-symmetric closure~of~$\redh$.
\item Given $\mM, \mN, \mM\reddh\mN$, we write $|\mM\reddh\mN|\in\nat$ for the length of the (unique) reduction path from $\mM$ to $\mN$.
\item For $n\ge 0$, we write $\mM\reddh^{n}\mN$ whenever $\mM \reddh \mN$ and $\len{\mM \reddh \mN} = n$ hold.
\esub
\end{defi}
\begin{figure*}[t!]
\centering
{\bf Unconditional rewriting rules}
\[
	\bar{rcl}
	\tuple{\vec R,\Call i; P,T}&\redh&\appT{\Lookinv{\val{R_i}}}{T}\\
	\tuple{\vec R,\Load i;P,\Cons a{T}} &\redh& \tuple{\vec R[R_i := a],P,T}\\
	\tuple{\vec R,\Apply i j k; P,T}&\redh&\tuple{\vec R[R_k := \App{\val{R_i}}{\val{R_j}}],P,T}\\
	\ear
\]
{\bf Under the assumption that $\val{R_i} \in \nat$.}
\[
	\bar{rcl}
	\tuple{\vec R,\Pred i j;P,T} &\redh& 
	\tuple{\vec R[R_j := \val{R_i} \ominus 1,P,T},\textrm{ where }n\ominus1 := \max\set{n - 1, 0},\\[1ex]
	\tuple{\vec R,\Succ i j;P,T} &\redh& 
	\tuple{\vec R[R_j := \val{R_i} + 1,P,T}
	\\[1ex]	
	\tuple{\vec R,\Ifz i j k l; P, T}&\redh&\begin{cases}
		\tuple{\vec R[R_l := \val{R_j}], P, T},&\textrm{if }\val{R_i} = 0,\\[3pt]
		\tuple{\vec R[R_l := \val{R_k}], P, T},&\textrm{otherwise}.
		\end{cases}\\
	\ear
\]
{\bf Under the assumption that $\Lookinv{\val{R_i}}\redh \mM$.}
\[
	\bar{rcl}
	\tuple{\vec R,\Pred i j;P,T} &\redh& 
	\tuple{\vec R[R_i := \Lookup{\mM}],\Pred i j;P,T}
	\\
	\tuple{\vec R,\Succ i j;P,T} &\redh& 
	\tuple{\vec R[R_i := \Lookup{\mM}],\Succ i j;P,T}
	\\
	\tuple{\vec R,\Ifz i j k l; P, T}&\redh&\tuple{\vec R[R_i := \Lookup{\mM}],\Ifz i j k l; P, T}\\[-2ex]
	\ear
\]
\caption{Small-step operational semantics for extended addressing machines.}
\label{fig:am:small_step}
\end{figure*}
\noindent %FIXED indent
As a matter of terminology, we say that an EAM $\mM$: 
	\begin{itemize}
	\item is \emph{stuck} %, written $\stuck{\mM}$, 
	if its program has shape $\mM.P = \Load i;P$ but $\mM.T = []$;
		\item \emph{is in final state} if $\mM$ is not in an error state, but cannot reduce further, i.e.\ $\mM\not\redh$;
	\item \emph{is in an error state} if its program has shape $\mM.P = \ins{ins};P'$ for some instruction $\ins{ins}\in\set{\Load i,\Pred i j,\Succ i j,\Ifz i j k l}$, but $\val{R_i}\notin\nat$ and $\Lookinv{\val{R_i}}$ cannot reduce further.
	\item \emph{reaches a final state} (resp.~\emph{raises an error}) if $\mM\reddh\mM'$ for some $\mM'$ in final (resp.\ error) state;  
	$\mM$ \emph{does not terminate}, otherwise.
	\end{itemize}
\noindent %FIXED indent
Given an EAM $\mM$, the first instruction of its program, together with the contents of its registers and tape, univocally determine which rule from Figure~\ref{fig:am:small_step} is applicable (if any).
When $\mM$ tries to perform an arithmetic operation in one of its registers, say $R_i = a$, it needs to wait that the EAM $\Lookinv{a}$ terminates its execution. 
If it does then the success of the operation depends on whether the result is a numeral, otherwise $\mM$ is in an error state. \pagebreak

\begin{lem}\label{lem:redhproperties}
\bsub\item\label{lem:redhproperties1} 
	The strategy $\redh$ is deterministic: $\mN\,{}_\mach{c}\!\!\leftarrow\mM\redh\! \mN'$ implies $ \mN = \mN'$.
\item\label{lem:redhproperties2} 
	The reduction $\redh$ is Church-Rosser: $\mM\convh\mN \Leftrightarrow \exists \mach{Z}\in\cM\,.\,\mM\reddh \mach{Z} {~}_{\mach{c}}\!\!\twoheadleftarrow\mN$.
\item \label{lem:redhproperties3}
 	If $\mM \reddh \mach{M'}$, then $\appT{\mM}{[\Lookup\mN]} \reddh \appT{\mach{M'}}{[\Lookup\mN]}$.
\esub
\end{lem}
\begin{proof}
\bsub
\item Trivial: There is at most one applicable reduction rule for each state.
\item Trivial as a consequence of (\ref{lem:redhproperties1}).
\item By induction on the length of $\mM \reddh \mach{M'}$.\qedhere
\esub
\end{proof}

\begin{exas} See Example~\ref{ex:extabsmach} for the definition of $\mach{I}, \mach{Succ1}$, $\mach{Succ2}$, $\mach{Add\_aux}$. 
\bsub
\item For some $a\in\Addrs$, $\appT{\mach{I}}{[a]} \redh \tuple{R_0 = a, \Call 0, []} \redh \Lookinv{a}$
\item We have $\appT{\mach{Succ1}}{[0]}\reddh \Lookinv{1}$ and $\appT{\mach{Succ2}}{[1]}\reddh \Lookinv{3}$.
\item Define $\mach{Add} = \appT{\mY}{[\Lookup{\mach{Add\_aux}}]}$, an EAM performing the addition. We show:\\
\noindent %FIXED indent
\begin{minipage}{\linewidth} %FIXED Moved into Minipage and reduced spacing
\begin{equation*}
	\bar{l}
	\appT{\mach{Add}}{[1,3]}\!\!
\redh \!\! \Tuple{(R_0 \!=\! \Lookup{\mY},R_1 = \Lookup{\mach{Add\_aux}}), \Apply 0 1 0;\Apply 1 01; \Call 1,[1, 3]}\\[2ex]
\reddh\! \Tuple{\vec R, \Load (0,1,2); \Pred 1 3; \Succ 2 4; \Apply 0 3 0;\Apply 0 4 0;\\\Ifz 1 2 0 0; \Call 0, [\Lookup{\mach{Add}}, 1, 3]}\\[3ex]
	\reddh\!\Tuple{R_0 \!=\! \Lookup{\mach{Add}}, R_1 = 1, R_2 = 3, R_3, R_4, \Pred 1 3; \Succ 2 4; \Apply 0 3 0;\\\Apply 0 4 0; \Ifz 1 2 0 0; \Call 0, []}\\[2ex]
	\reddh\!\Tuple{R_0 \!=\! \Lookup(\appT{\mach{Add}}{[0,4]}), R_1 = 1, R_2 = 3, R_3 = 0, R_4 = 4, \Ifz 1 2 0 0; \Call 0, []}\\[2ex]
	\reddh\!\Tuple{R_0 \!=\! \Lookup(\appT{\mach{Add}}{[0,5]}), R_1 = 0, R_2 = 4, R_3 = 0, R_4 = 5, \Ifz 1 2 0 0; \Call 0, []}\\[2ex]
	\reddh\! \Lookinv{4}.
\ear
\end{equation*}
\end{minipage}
\esub
\end{exas}

% !TEX root = ../LMCS.tex
%!TEX spellcheck = en-US

\subsection{Typing Extended Addressing Machines}

Recall that the set $\Types$ of (simple) types has been introduced in Definition~\ref{def:simpletypes}\eqref{def:simpletypes1}. 
We now show that certain EAMs can be typed, and that typable machines do not raise an error.

\begin{defi}\label{def:EAMtypes} 
\bsub
\item A \emph{typing context} $\Delta$ is a finite set of associations between indices and types, represented as a list $i_1 : \alpha_1,\dots,i_{r} : \alpha_{r}$. The indices $i_1,\dots,i_r$ are not necessarily consecutive.
\item We denote by $\Delta[i : \alpha]$ the typing context $\Delta$ where the type associated with $i$ becomes $\alpha$. Note that $\dom(\Delta[i : \alpha]) = \dom(\Delta)\cup\set{i}$. If $i$ is not present in $\Delta$, then $\Delta[i : \alpha] = \Delta, i : \alpha$.
\item Let $\Delta$ be a typing environment, $\mM\in\cM$, $r\ge 0$, $R_0,\dots,R_r$ be registers, $P$ be a program, $T\in\Tapes$ and $\alpha\in\Types$. We define the typing judgements
\[
	\mM : \alpha\qquad\qquad\qquad	\Delta \Vdash^r (P,T) : \alpha
	\qquad\qquad\qquad R_0,\dots,R_r\models \Delta
\] 
by mutual induction as the least relations closed under the rules of Figure~\ref{fig:eamstyping}.
\item A machine $\mM$ is called \emph{typable} if the judgement $\mM : \alpha$ is derivable for some $\alpha\in\Types$.
\esub
\end{defi}

\begin{figure*}[t!]
\[\bar{ccc}
\infer[(\mathrm{nat})]{\mach{n} : \tint}{}
\qquad
	\infer[(\mathrm{fix})]{\mY{} : (\alpha \to \alpha) \to \alpha}{}&&
	\infer[(\vec R)]{\tuple{R_0,\dots,R_{r},P,T} : \alpha}{R_0,\dots,R_{r} \models \Delta & \Delta\Vdash^r (P,T) : \alpha}
	\\[1ex]
	\infer[(R_\Null)]{R_0,\dots,R_r \models \Delta}{R_0,\dots,R_{r-1} \models \Delta & !R_r = \Null}
	&&
	\infer[(R_\Types)]{R_0,\dots,R_r \models \Delta, r:\alpha}{R_0,\dots,R_{r-1} \models \Delta & \Lookinv{!R_r} : \alpha}
	\\[1ex]
	\infer[(\mathrm{load_{\Null}})]{\Delta \Vdash^r (\Load i;P,[]) : \beta\to\alpha}{\Delta[i : \beta] \Vdash^r (P,[]) : \alpha}
	&&
	\infer[(\mathrm{load_{\Types}})]{\Delta\Vdash^r (\Load i;P, \Cons a {T}) : \alpha}{\Delta[i : \beta] \Vdash^r (P,T) : \alpha & \Lookinv {a} : \beta}
	\\[1ex]
	\infer[(\mathrm{pred})]{\Delta, i : \tint \Vdash^r (\Pred i j;P,T) : \alpha}{(\Delta, i : \tint)[j : \tint] \Vdash^r (P,T) : \alpha}
	&&
	\infer[(\mathrm{succ})]{\Delta, i : \tint \Vdash^r (\Succ i j;P,T) : \alpha}{(\Delta, i : \tint)[j :\tint] \Vdash^r (P,T) : \alpha}
	\\[1ex]
	\multicolumn{3}{c}{\infer[(\mathrm{test})]{\Delta, i : \tint,  j :\beta, k : \beta \Vdash^r (\Ifz i j k l;P,T) : \alpha}{
		(\Delta, i : \tint,  j : \beta, k : \beta)[l : \beta] \Vdash^r (P,T) : \alpha 
		}
	}
	\\[1ex]
	\multicolumn{3}{c}{\infer[(\mathrm{app})]{\Delta, i : \alpha \to \beta, j : \alpha \Vdash^r (\Apply i j k;P,T) : \delta}{
		(\Delta, i : \alpha \to \beta, j : \alpha)[k : \beta] \Vdash^r (P,T) : \delta}	}
	\\[1ex]
	\multicolumn{3}{c}{\infer[(\mathrm{call})]{\Delta, i : \alpha_1\to\dots\to\alpha_n\to\alpha \Vdash^r (\Call i,[\Lookup{\mM_1},\dots,\Lookup{\mM_n}]) : \alpha}{
		 \mM_1 : \alpha_1 & \cdots &  \mM_n : \alpha_n}}
\ear
\]
\caption{Typing rules for extended addressing machines.}\label{fig:eamstyping}
\end{figure*}
\noindent %FIXED indent
The algorithm in Figure~\ref{fig:eamstyping} deserves some discussion. As it is presented as a set of inference rules, one should reason bottom-up.

To give a machine $\mM$ a type $\alpha$, one needs to derive the judgement $\mM : \alpha$. The machines $\mach{n}$ and $\mY$ are recognizable from their addresses and the rules $(\mathrm{nat})$ and $(\mathrm{fix})$ can thus be given higher precedence. For all other cases, one begins by giving all the registers a type using the rule $(\vec R)$, applying $(R_\Types)$ or $(R_\Null)$ for each register. Once this initial step is performed, one needs to derive a judgement of the form $i_1:\beta_{i_1},\dots,i_n:\beta_{i_n}\Vdash^r (P,T) : \alpha$, where $P$ and $T$ are the program and the input tape of the original machine respectively. This is done by verifying the coherence of the instructions in the program with the types of the registers and of the values in the input tape.

As a final consideration, notice that the rules in Figure~\ref{fig:eamstyping} can only be considered as an algorithm when the address table map is effectively given. Otherwise, the algorithm would depend on an oracle deciding $a = \#\mM$.
\begin{rems}\label{rem:abouttypings}
\bsub
\item\label{rem:abouttypings1} For all $\mM\in\cM$ and $\alpha\in\Types$, we have $\mM : \alpha$ if and only if there exists $a\in\Addrs$ such that both $\Lookinv{a} : \alpha$ and $ \Lookup{\mM} = a$ hold.
\item\label{rem:abouttypings2} 
If $\Lookup\mM\notin \nnat[]\cup\set{\Lookup{\mY}}$, then 
\[
	\mM : \alpha \iff \exists\Delta\,.\, [\Delta\models \mM.\vec R\ \land\ \Delta \Vdash^r (\mM.P,\mM.T) : \alpha ]
\]
\item The superscript $r\ge 0$ in $\Vdash^r$ keeps track of the initial amount of registers, i.e.\ $\vec i \in\set{0,\dots,r}$.
\item\label{rem:abouttypings3} Numeral machines are not typable by content, so $(\mathrm{nat})$ having a higher priority is necessary for the consistency of the system. The $(\mathrm{fix})$ rule having a higher priority does not modify the set of typable machines, but guarantees the syntax-directedness of the system.
\esub
\end{rems}

\begin{exas} The following are some examples of derivable typing judgements.
\bsub
\item $\mach{I}$ can be typed with $\alpha\rightarrow\alpha$ for any $\alpha \in \Types$:
\[
\infer{\tuple{R_0 = \Null, \Load 0;\Call 0, []}: \alpha \rightarrow\alpha}{
\infer{R_0\models}{!R_0 = \Null}
&
\infer{\Vdash^1 \tuple{\Load 0;\Call 0,[]}:\alpha\rightarrow\alpha}{
\infer{0 : \tint\Vdash^1 \tuple{\Call 0, []} : \alpha}{}
}
}
\]
\item $\mach{Succ1}$ can be typed with $\tint \rightarrow\tint$:
\[
\infer{\tuple{R_0=\Null, \Load 0; \Succ 0 0; \Call 0, []} : \tint\to\tint}{
	\infer{R_0\models}{!R_0 = \Null}
	&
	\infer{\Vdash^1\tuple{ \Load 0; \Succ 0 0; \Call 0, []} : \tint\to\tint}{\infer{0 : \tint \Vdash^1 \tuple{\Succ 0 0; \Call 0, []} : \tint}{
		\infer{0 : \tint\Vdash^1 \tuple{\Call 0, []} : \tint}{}	
		}
	}
}
\]
\item $\mach{Succ2} : \tint \rightarrow \tint$ and $\mach{Add} : \tint\rightarrow\tint\rightarrow\tint$ are also derivable typing judgements, but the full trees are omitted due to size constraints.
\esub
\end{exas}

\begin{lem} Let $\mM\in\cM$, $\alpha\in\Types$. Assume that $\# : \mM\to\Addrs$ is effectively given.
\bsub
\item If $\mM = \tuple{\vec R = \Null,P,[]}$ then the typing algorithm is capable of deciding whether $\mM : \alpha$ holds.
\item In general, the typing algorithm semi-decides whether $\mM : \alpha$ holds.
\esub
\end{lem}

\begin{proof}(Sketch) (i) In this case, $\mM : \alpha$ holds if and only if $\Vdash^r (\mM.P,[])$ does. By induction on the length of $\mM.P$, one verifies if it is possible to construct a derivation. Otherwise, conclude that $\mM : \alpha$ is not derivable.

(ii) In the rules $(R_\Types)$ and $(\mathrm{load}_\mathbb{T})$, one needs to show that a type for the premises exists.
As the set of types is countable, and effectively given, one can easily design an algorithm constructing a derivation tree (by dovetailing).
However, the algorithm cannot terminate when executed on the machine $\mM_0$ defined in Remark~\ref{rem:sillyEAMs}(\ref{rem:sillyEAMs2}) because it would require an infinite derivation tree.
\end{proof}
 
\begin{prop}\label{prop:typing} 
Let $\mM,\mach {M'},\mN,\in\cM$ and $\alpha,\beta\in\Types$.
\bsub
\item\label{prop:typing1}  If $\mM : \beta \to \alpha$ and $\mN : \beta$ then $\appT{\mM}{[\Lookup{\mN}]} : \alpha$.
\item\label{prop:typing2}  If $ \mM : \alpha$ and $\mM\redh \mN$ then $ \mN : \alpha$.
\item\label{prop:typing4}  If $ \mM : \tint$ then either $\mM$ does not terminate or $\mM\reddh \mach{n}$, for some $n\ge 0$.
\item\label{prop:typing3}  If $ \mM : \alpha$ then $\mM$ does not raise an error.
\esub
\end{prop}

\begin{proof} 
\bsub \item Simultaneously, one proves that if both $\Delta\Vdash^r (P,T) : \beta\to\alpha$ and $ \mN : \beta$ hold, then so does $\Delta\Vdash^r(P,\appT{T}{[\Lookup \mN]}) : \alpha$.
We proceed by induction on a derivation of $ \mM : \beta\to\alpha$ (resp.\ $\Delta\Vdash^r (P,T) : \beta\to\alpha$). 

Case ($\mathrm{nat}$) is vacuous.

Case ($\mathrm{fix}$). By definition of $\mY$ (see Definition~\ref{def:bijectivelookup+Y}\eqref{def:ymachine}), we have: %FIXED Adjusted Spacing
\[
	\appT{\mY}\!{[\Lookup\mN]} \!=\! \Tuple{R_0 \!=\! \Null, R_1 \!=\! \Null, \Load 0 ; \Apply 1 0 1;\Apply 0 1 0; \Call \!0,[\Lookup{\mY}, \Lookup{\mN}]}.
\]
Notice that, in this case, $\beta = \alpha\to\alpha$. We derive (omitting the $(R_\Null)$ rule usages) :
{
\begin{minipage}{\linewidth} %FIXED Put into minipage for alignment
\[
	\infer{\appT{\mY}{[\Lookup\mN]} : \alpha}{\infer{
	\Vdash^2 \tuple{\mY.P,[\Lookup{\mY},\Lookup{\mN}]} : \alpha
	}{\infer{0:(\alpha\to\alpha)\to\alpha\Vdash^2 \tuple{\Load 1;\dots,[\Lookup{\mN}]} : \alpha}{
	\infer{0 : (\alpha\to\alpha)\to\alpha, 1 : \alpha\to\alpha\Vdash^2 (\Apply 0 1 0;\dots,[]) : \alpha
	}{\infer{0 : \alpha, 1 : \alpha\to\alpha\Vdash^2 (\Apply 1 0 1;\dots,[]) : \alpha}{
				\infer{0 : \alpha,1 :\alpha\Vdash^2 ( \Call 1,[]) : \alpha}{}
				}
		}
		&
		\!\vdash \mN\! :\! \alpha\!\to\!\alpha}
		&
		\infer{ \mY \!:\! (\alpha\!\to\!\alpha)\!\to\!\alpha}{}
	}
	}
\]
\end{minipage}
}
\vspace{0.5\baselineskip}\\\indent %FIXED Corrected Vertical Spacing towards the bottom
Case $\mathrm{load}_\Null$. Then $P = \Load i;P'$, $T = []$ and $\Delta[i : \beta] \Vdash^r (P',[]) : \alpha$.
By assumption $ \mN : \beta$, so we conclude $\Delta \Vdash^r (\Load i;P',[\Lookup{\mN}]) : \alpha$ by applying $\mathrm{load}_\Types$.

All other cases derive straightforwardly from the IH.

\item The cases $\mM = \mY$ or $\mM = \mach{n}$ for some $n\in\nat$ are vacuous, as these machines are in final state.
Otherwise, by Remark~\ref{rem:abouttypings}\eqref{rem:abouttypings2},  $\Delta \Vdash^r (\mM.P,\mM.T) : \alpha$ for some $\Delta\models \mM.\vec R$.
By cases on the shape of $\mM.P$. 

Case $P = \Load i;P'$. Then $\mM.T = \Cons a T'$ otherwise $\mM$ would be in final state, and $\mN = \tuple{\vec R[R_i := a],P',T'}$. 
From $(\mathrm{Load}_\Types)$ we get $\Delta[i : \beta] \Vdash^r (P',T') : \alpha$ for some $\beta\in\Types$ satisfying $\Lookinv {a} : \beta$. 
As $\vec R\models \Delta$ we derive $\vec R[R_i := a]\models\Delta[i : \beta] $, so as $N = \Tuple{\vec R[R_i := a],P',T'}$, by Remark~\ref{rem:abouttypings}\eqref{rem:abouttypings2}, $N : \alpha$.

Case $P = \Call i$. Then $i : \alpha_1\to\cdots\to\alpha_n\to\alpha$, $T = [\Lookup{\mM_1},\dots,\Lookup{\mM_n}]$ and $\mM_j : \alpha_j$, for all $j\le n$.
In this case, $\mN = \appT{\Lookinv{\val{(\mM.R_i)}}}{T}$ with $\Lookinv{\val{(\mM.R_i)}} : \alpha_1\to\cdots\to\alpha_n\to\alpha$, so we conclude by \eqref{prop:typing1}.

All other cases follows easily from the IH.

  \item Assume that $\mM : \tint$ and $\mM\reddh \mN$ for some $\mN$ in final state. By \eqref{prop:typing2}, we obtain that $\mN: \tint$ holds, thus $\mN= \mach{n}$ since numerals are the only machines in final state typable with $\tint$.

 \item The three cases from Figure~\ref{fig:am:small_step} where a machine can raise an error are ruled out by the typing rules ($\mathrm{pred}$), ($\mathrm{succ}$) and ($\mathrm{test}$), respectively. Therefore, no error can be raised during the execution.\qedhere
\esub
\end{proof}

\section{Simulating (E)PCF}\label{sec:Simulation}
% !TEX root = ../LMCS.tex
%!TEX spellcheck = en-US

In this section we define a translation from \EPCF{} terms to EAMs. 
We prove that both the typing and the operational semantics of \EPCF{} programs are preserved under this translation.
In Theorem~\ref{thm:eqPCFandEPCF}, we show that a \PCF{} program having type $\tint$ computes a numeral $\num n$ exactly when the corresponding EAM reduces to $\mach n$.
This result builds upon \cite[Thm.~4.12]{IntrigilaMM22}, where only one implication is shown. It can be seen as a first step towards full abstraction.

\subsection{EAMs implementing PCF instructions}

We start by defining some auxiliary EAMs implementing the main instructions of \PCF.

\begin{defi}\label{def:fixedmachines} Define the following EAMs (for $k>0$ and $n\ge 0$):
\[
	\bar{rl}
	\mProj{k}{i} =& \Tuple{R_0, (\Load 1)^{i-1};\Load 0; (\Load 1)^{k-i}; \Call 0; []},\textrm{ for }1\le i \le k;\\[4pt]
	\mPred =& \Tuple{R_0, \Load 0;\Pred 00;\Call 0; []};\\[4pt]
	\mSucc =& \Tuple{R_0, \Load 0;\Succ 00;\Call 0; []};\\[4pt]
	\ear
	\]
	\[
	\bar{rl}
	\mIfz =& \Tuple{R_0,R_1,R_2, \Load 0;\Load 1;\Load 2;\Ifz 0120; \Call 0; []};\\[4pt]
	\mAppn{0}{k} =&\mProj{1}{1},\textrm{ for $k > 0$. Recall that the EAM $\mProj{1}{1} = \mach{I}$ represents the identity};\\[4pt]
	\mAppn{n+1}{k} =& \Tuple{R_0 = \Lookup{\mAppn{n}{k}},R_1,\dots,R_{k+2},\Load (1,\dots,k+2);\\\Apply {2}{k+2}{2};\cdots;\Apply {k+1}{k+2}{k+1};\\ \Apply 010;\cdots;\Apply 0{k+1}0;\Call 0, []}.\\
	\ear
\]
The registers whose values are not specified are assumed to be uninitialized.
\end{defi}

The EAM $\mAppn{n}{k}$ deserves an explanation. 
This machine is stuck, waiting for $n+k+1$ arguments $a,d_1,\ldots,d_k,e_1,\ldots,e_n$, where $k$ is the arity of $a$ and $n$ is the arity of each $d_i$. Once all arguments are available in the tape, $\mAppn{n}{k}$ applies $\vec e$ to each $d_i$ and then feeds the machine $\Lookinv{a}$ the resulting list of arguments $d_1\cdot\vec e,\dots,d_k\cdot\vec e$.

\begin{lem}\label{lem:existenceofeams} 
For all $a,b,c,d_1,\dots,d_k,e_1,\dots,e_n\in\Addrs$, the EAMs below reduce as follows:
\bsub
\item\label{lem:existenceofeams1} 
	$\appT{\mProj{k}{i}}{[d_1,\dots,d_k]} \reddh^{k+1} d_i$, for $1\le i \le k>0$;
\item\label{lem:existenceofeams2}  
	$\appT{\mAppn{n}{k}}{[a,d_1,\ldots,d_k,e_1,\ldots,e_n]} \reddh^{(3k+4)n + 2} \appT{\Lookinv a}{[d_1\cdot e_1\cdots e_n,\ldots,d_k\cdot e_1\cdots e_n]}\!$;
\item\label{lem:existenceofeams3}  
	 $\appT{\mPred}{[a]}\redh \Tuple{R_0 = a, \Pred 00;\Call 0,[]}$;
\item\label{lem:existenceofeams4}   
	$\appT{\mSucc}{[a]}\redh \Tuple{R_0 = a, \Succ 00;\Call 0,[]}$;
\item\label{lem:existenceofeams5}
	$\appT{\mIfz}{[a,b,c]}\reddh^3\Tuple{R_0 = a,R_1 = b, R_2 = c, \Ifz 0120;\Call 0,[]}$;
\item\label{lem:existenceofeams6}
	$\appT{\mY}{[a]}\reddh^5 \appT{\Lookinv{a}}{[\#(\appT{\mY}{[a]})]}$.
\esub
\end{lem}
\begin{proof} The only interesting cases are the application and fixed point combinator.

\eqref{lem:existenceofeams2} Concerning $\mAppn{n}{k}$, we proceed by induction on $n$.
\begin{itemize}
\item
Base case. If $n=0$ then $\mAppn{0}{k} = \mach{I}$, and $\appT{\mach{I}}{[a,d_1,\dots,d_k]} \reddh^2 \appT{\Lookinv{a}}{[d_1,\dots,d_k]}$.

\item Induction case. Easy calculations give:\\
\begin{minipage}{\linewidth} %FIXED Put into Minipage for alignment
\[
	\appT{
	\mAppn{n+1}{k}
	}{
	[a,d_1,\dots,d_k,e_1,\dots,e_{n+1}]
	} \reddh^{3k+4} \appT{\mAppn{n}{k}}{[a,d_1\cdot e_1,\dots,d_k\cdot e_1,e_2,\dots,e_{n+1}]}
\]
This case follows from the IH since $3k+4+(3k+4)n + 2 = (3k+4)(n+1) + 2$ .
\end{minipage}

\end{itemize}
\noindent %FIXED indent
\eqref{lem:existenceofeams6} Recall that $\mY$ has been introduced in Definitions \ref{def:AM}\eqref{def:AM6} and \ref{def:bijectivelookup+Y}\eqref{def:ymachine}.
\[\bar{lcl}
	\mY &=& \tuple{R_0,R_1 ,\Load (0, 1);\Apply  0 1 0;\Apply 101; \Call 1,[\Lookup\mY,a]}\\
	&\reddh^2&\tuple{R_0 = \Lookup\mY,R_1 =a,\Apply  0 1 0;\Apply 101; \Call 1,[]}\\
	&\redh&\tuple{R_0 = \Lookup{\mY}\cdot a,R_1 =a,\Apply 101; \Call 1,[]}\\	
	&\redh&\tuple{R_0 = \Lookup{\mY}\cdot a,R_1 =a\cdot(\Lookup{\mY}\cdot a), \Call 1,[]}\\	
	&\redh&\Lookinv{a\cdot(\Lookup{\mY}\cdot a)} = \appT{\Lookinv a}{[\appT{\Lookup\mY}{[a]}]}.\hfill 
\ear
\]
This concludes the proof.
\end{proof}

We show that the above machines are actually typable using the type assignment system of Figure~\ref{fig:eamstyping}.
This property is needed to show that the translation is type-preserving.

\begin{lem}\label{lem:welltypedeams}  
The EAMs introduced in Definition~\ref{def:fixedmachines} can be typed as follows (for all $\alpha,\beta_i,\delta_i\in\Types$, using the notation $\vec \delta \to\beta_i= \delta_1\to\cdots\to\delta_n\to\beta_i$):
\bsub
\item\label{lem:welltypedeams1}
	$ \mProj{k}{i} : \beta_1\to\cdots\to\beta_k\to\beta_i$, for $1\le i \le k>0$;
\item\label{lem:welltypedeams2} 
	$\mAppn{n}{k} : (\beta_1\to\dots\to\beta_k\to\alpha)\to(\vec \delta\to\beta_1)\to\dots\to(\vec \delta\to\beta_k)\to\vec\delta\to\alpha$;
\item\label{lem:welltypedeams3}
	$ \mPred : \tint\to\tint$;
\item\label{lem:welltypedeams4} 
	$\mSucc : \tint\to\tint$;
\item\label{lem:welltypedeams5} 
	$\mIfz : \tint\to\alpha\to\alpha\to\alpha$.
\esub
\end{lem}
\begin{proof}
It follows easily from Definition~\ref{def:fixedmachines}. For $\mAppn{n}{k}$, proceed by induction on $n$.
\end{proof}

\subsection{Translating (E)PCF programs to EAMs}

Using the auxiliary EAM introduced above, we can translate an $\EPCF$ term  $M$ having $x_1,\dots,x_n$ as free variables as an EAM $\mach{M}$ loading their values from the input tape. 

\begin{defi}[Translation]\label{def:trans}
Let $M$ be an \EPCF{} term such that $\FV{M}\subseteq\set{x_1,\dots,x_n}$.
The \emph{translation} of $M$ (w.r.t.\ $\vec x$) is an EAM $\trans[\vec x]{M}$ defined by structural induction on $M$:
\[
	\bar{lcll}
	\trans[\vec x]{x_i} &=& \mProj{n}{i},\textrm{ where }i\in\set{1,\dots,n};\\[4pt]		
	\trans[\vec x]{\lambda y.M}&= & \trans[\vec x,y]{M}, \textrm{ where wlog }y\notin\vec x;\\[4pt]
		\trans[\vec x]{M\langle N/y\rangle} &=&  \appT{\trans[y,\vec x]{M}}{[\Lookup{\trans{N}}]};\\[4pt]		
	\trans[\vec x]{M\cdot N} &=&  \appT{\mAppn{n}{2}}{[\Lookup{\mProj{1}{1}},\Lookup{\trans[\vec x]{M}},\Lookup {\trans[\vec x]{N}}]};\\[4pt]
	\trans[\vec x]{\mathbf{0}} &=& \appT{\mProj{n+1}{1}}{[0]};\\[4pt]
	\trans[\vec x]{\pred M} &=& \appT{\mAppn{n}{1}}{[\Lookup\mPred,\Lookup{\trans[\vec x]{M}}]};\\[4pt]
	\trans[\vec x]{\succ M} &=& \appT{\mAppn{n}{1}}{[\Lookup\mSucc,\Lookup{\trans[\vec x]{M}}]};\\[4pt]
	\trans[\vec x]{\ifterm LMN} &=& \appT{\mAppn{n}{3}}{[\Lookup\mIfz,\Lookup{\trans[\vec x]{L}},\Lookup{\trans[\vec x]{M}},\Lookup{\trans[\vec x]{N}}]};\\[4pt]
	\trans[\vec x]{\fix M} &=& \begin{cases}\appT{\mY}{[\Lookup{\trans[\vec x]{M}}]},&\textrm{if }n = 0,\\
	\appT{\mAppn{n}{1}}{[\Lookup\mY,\Lookup{\trans[\vec x]{M}}]},& \mathrm{otherwise}.\end{cases}
	\ear
\]
\end{defi}

\begin{exas} Recall the \EPCF{} terms introduced in Example~\ref{ex:PCFterms}. 
The translation of said \EPCF{} terms produces the following machines:
\bsub
\item $\trans{\comb{I}} = \trans[x]{x} = \mProj{1}{1}$.
\item $\trans{\Om} = \appT{\mY}{[\Lookup{\trans{\comb{I}}}]} = \appT{\mY}{[\Lookup{\mProj{1}{1}}]}$.
\item $\trans{\mathbf{succ1}} = \trans{\lam x.\succ x} = \trans[x]{\succ x} = \appT{\mAppn{1}{1}}{[\Lookup{\mSucc}, \Lookup{\mProj{1}{1}}]}$,
\item $\!\!\!\bar[t]{{lcl}}
	\trans{\mathbf{succ2}}&=&\appT{\mAppn{2}{0}}{[\Lookup{\trans[s,n]{s\cdot(s\cdot n)}}\cdot \Lookup{\trans{\mathbf{succ1}}}]}\\[4pt]
					 &= &\appT{\mProj{1}{1}}{[\Lookup{\mAppn{2}{2}}\cdot\Lookup{\mProj{2}{1}}\cdot\Lookup{\trans[s,n]{s\cdot n}}, \Lookup{\trans{\mathbf{succ1}}}]}\\[4pt]
					 &= &\appT{\mProj{1}{1}}{[\Lookup{\mAppn{2}{2}}\cdot\Lookup{\mProj{2}{1}}\cdot(\Lookup{\mAppn{2}{2}}\cdot \mProj{2}{1}\cdot \mProj{2}{2}), \Lookup{\trans{\mathbf{succ1}}}]}\\[4pt]
					&= &\appT{\mProj{1}{1}}{[\Lookup{\mAppn{2}{2}}\cdot\Lookup{\mProj{2}{1}}\cdot(\Lookup{\mAppn{2}{2}}\cdot \mProj{2}{1}\cdot \mProj{2}{2}), \Lookup{\mAppn{1}{1}}\cdot \Lookup{\mSucc} \cdot \Lookup{\mProj{1}{1}}]}.
    \ear$
\esub
\end{exas}
\noindent %FIXED indent
In the translation above typing information was deliberately ignored for the sake of generality. We now show that our translation preserves the typings in the following sense.

\begin{thm}\label{thm:transtyping} Let  $M$ be an $\mathsf{(E)PCF}$ term and $\Gamma = x_1:\delta_1,\dots,x_n:\delta_n$. Then
\[  
	\Gamma \vdash M:\alpha \imp \trans[\vec x]{M} : \delta_1\rightarrow\cdots\rightarrow\delta_n\rightarrow\alpha.
\]
\end{thm}
\begin{proof}[Proof (Appendix~\ref{app:Sim})]
Note that the type assignment systems of \EPCF{} and of \PCF{} coincide on \PCF{} terms.
Proceed by induction on a derivation of $\Gamma \vdash M:\alpha$, using Proposition~\ref{prop:typing}\eqref{prop:typing1} and Lemma~\ref{lem:welltypedeams}.
\end{proof}

Ideally, one would like that an \EPCF{} step $M \redwh N$ becomes a reduction $\trans{M} \reddh \trans{N}$ in the corresponding EAMs. Unfortunately, the situation is more complicated---the translation $\trans{N}$ may contain auxiliary EAMs that are not generated by $\trans{M}$ along reduction. 
The property that actually holds is that the two EAMs are interconvertible $\trans{M}\convh \trans{N}$, and $\trans{N}$ is `closer' to their common reduct.
The next definition captures this intuition.

\begin{defi} For $\mM,\mN\in\cM$, define the relation:
\[
	\mM\convg\mN \iff \exists \mach{Z}\in\cM\,.\,[\,(\mM\reddh \mach{Z} {~}_{\mach{c}}\!\!\twoheadleftarrow\mN) \land (\len{\mM\reddh \mach{Z}} > \len{\mN\reddh \mach{Z}})\,].
\]
Note that the reduction path $\mM\reddh \mach{Z}$ above must be non-empty.
Moreover, recall that $\mM\convg\mN$ entails $\mM\convh\mN$.
\end{defi}

\begin{lem}\label{lemma:convgstrictness}
\bsub\item\label{lemma:convgstrictness1} $\len{\mM\reddh \mach{Z}} > \len{\mN\reddh \mach{Z}}$ and $\mach{Z}\reddh \mach{Z}'$ imply $\len{\mM\reddh \mach{Z}'} > \len{\mN\reddh \mach{Z}'}$.
\item\label{lemma:convgstrictness2}
The relation $\convg$ is transitive.
\esub
\end{lem}
\begin{proof}
\eqref{lemma:convgstrictness1} By confluence and determinism of $\redh$ (Lemma~\ref{lem:redhproperties}). 

\eqref{lemma:convgstrictness2} follows from \eqref{lemma:convgstrictness1}.
\end{proof}

\begin{prop}\label{prop:smallstep} Given \EPCF{} program $M$ of type $\tint$, we have $M \redwh N \Rightarrow \trans{M}\convg \trans{N}$.
\end{prop}
\begin{proof}[Proof (Appendix~\ref{app:Sim})] By induction on a derivation of $M \redwh N$, applying Lemma~\ref{lem:existenceofeams}.
\end{proof}

Since every \PCF{} program $P$ of type $\tint$ is also an \EPCF{} program of the same type, and in this case the operational semantics of \PCF{} and \EPCF{} coincide (Theorem~\ref{thm:eqPCFandEPCF}), we can use the above proposition to prove that the EAM $\trans{P}$ faithfully simulates the behavior of $P$.

\begin{thm}\label{thm:simulation}  For a \PCF{} program $P$ having type $\tint$, the following are equivalent:
\begin{enumerate}\item
 $P \redd[\PCF] \num n $;
 \item
 	$ \trans{P} \reddh \mach{n}$.
 \end{enumerate}
\end{thm}
\begin{proof} $(1 \Rightarrow 2)$ Let $P \redd[\PCF] \num n $. 
Equivalently, $P \reddwh \num n$ (by Theorem~\ref{thm:eqPCFandEPCF}).
By Proposition~\ref{prop:smallstep}, we get $\trans{P} \convh \mach{n}$. 
Since $\mach{n}$ is in final state, this entails $\trans{P} \reddh \mach{n}$ by confluence (Lemma~\ref{lem:redhproperties}\eqref{lem:redhproperties2}).

 $(2 \Rightarrow 1)$ We prove the contrapositive. Assume that $\vdash P:\tint$, but $P$ does not reduce to a numeral. 
 As \PCF{} enjoys subject reduction~\cite{Ong95}, $P$ must have an infinite $\to_\PCF$ reduction path.
By Theorem~\ref{thm:eqPCFandEPCF}, this generates an infinite w.h.\ reduction path $P=P_0\redwh P_1\reddwh P_k\reddwh\cdots$.
By Proposition~\ref{prop:smallstep}, this translates to an infinite $\convg$-chain $\trans{P_k}\convg \trans{P_{k+1}}$ for common reducts $\trans{P_k}\reddh \mach{Z}_k {~}_{\mach{c}}\!\!\twoheadleftarrow\trans{P_{k+1}}$. By confluence and Lemma~\ref{lemma:convgstrictness}\eqref{lemma:convgstrictness1}, there are $\mach{Z}'_k$ and $i_k> j_k$ such that
\[{}\hspace{20pt}
\begin{tikzpicture}
\node (P1) at (-51pt,20pt) {$\trans{P_1}$};
\node (P2) at (-15pt,20pt) {$\trans{P_2}$};
\node (P3) at (19pt,20pt) {$\trans{P_3}$};
\node (P4) at (70pt,20pt) {$\cdots$};
\draw[->>] (P1) -- node [above,midway] {\tiny $j_0$} ($(P1)+(30pt,-13pt)$);
\draw[->>] (P2) -- node [above,midway] {\tiny $j_1$} ($(P2)+(30pt,-13pt)$);
\draw[->>] (P3) -- node [above,midway] {\tiny $j_2$} ($(P3)+(30pt,-13pt)$);
\node at (5pt,0){$\trans{P} = \trans{P_0}\reddh^{i_0} \mach{Z}'_1\reddh^{i_1} \mach{Z}'_2\reddh^{i_2} \mach{Z}'_3\reddh \cdots$};
\end{tikzpicture}
\]
This is only possible if the EAM $\trans{P}$ does not terminate.
\end{proof}

\section{A Fully Abstract Computational Model}\label{sec:Model}
% !TEX root = ../LMCS.tex
%!TEX spellcheck = en-US

We construct a model of \PCF{} based on EAMs and prove that it is fully abstract.
The main technical tools used to achieve this result are logical relations~\cite{AmadioC98,Pitts97} and definability~\cite{Curien07}.
We remind the definition of a fully abstract model of \PCF, together with some basic results about its observational equivalence.

Recall that \PCF{} contexts have been defined in Definition~\ref{def:contexts}\eqref{def:contexts1}, the relation $\to_\PCF$ represents a weak head reduction step and generates the interconvertibility relation $\conv[\PCF]$.

\begin{defi}
\bsub
	\item Given a \PCF{} context $\C\square$ an environment $\Gamma$ and two types $\alpha,\beta\in\Types$, we write $\C\square : (\Gamma,\alpha)\to\beta$ if $\vdash \C[P] : \beta$ holds, for all \PCF{} terms $P$ having type $\alpha$ in $\Gamma$, i.e.\ satisfying $\Gamma\vdash P : \alpha$.
\item The \emph{observational equivalence} is defined on \PCF{} terms $P,P'$ having type $\alpha$ in $\Gamma$ by:
\[
P\obseq P'\iff	\forall \C\square : (\Gamma,\alpha) \to\tint\,.\, [\,\C[P] \redd[\PCF]\num n\iff \C[P']\redd[\PCF] \num n\,]
\]
\item  The \emph{applicative equivalence} is defined on \PCF{} programs $P,P'$ having the same type $\alpha_1\to\cdots\to\alpha_k\to\tint$, for some $k\ge 0$, as follows:
\[
	P\appeq P'\iff \forall Q_1 : \alpha_1,\dots,Q_k:\alpha_k \,.\, [\,PQ_1\cdots Q_k\redd[\PCF]\num n\iff P'Q_1\cdots Q_k \redd[\PCF] \num n\,]
\]
\esub
\end{defi}

\begin{rem}
Note that \PCF{} contexts are different from \PCF{} evaluation contexts---unlike evaluation contexts, the `hole' can be found at any location in the term.
\end{rem}

The main difference between observational equivalence in the untyped setting and the one for \PCF{} is that, in the latter, we observe termination at ground type.
The following result shows that applicative and observational equivalence coincide on closed \PCF{} terms.

%FIXED Changed to PropC to omit double brackets
\begin{propC}[Abramsky's equivalence~\cite{Ong95}]\label{prop:Abramsky}~\\
For all $\PCF$ programs $P,P'$ of type $\alpha$, we have:
\[
	P\obseq P' \iff P\appeq  P'
\]
\end{propC}

\begin{defi}
	A \emph{model of \PCF{}} is a triple $\mathscr{M} = \tuple{(\MM_\alpha)_{\alpha\in\Types},(\ \cdot^{\alpha,\beta})_{\alpha,\beta\in\Types},\int{-}}$ where:
	\begin{itemize}
	\item $(\MM_\alpha)_{\alpha\in\Types}$ is a type-indexed collection of sets;
	\item $(\ \cdot^{\alpha,\beta}) : \MM_{\alpha\to\beta}\times\MM_\alpha\to\MM_\beta$ is a well-typed operation called \emph{application};
	\item $\int{-}$ is an \emph{interpretation function} mapping a derivation of $x_1:\beta_1,\dots,x_n:\beta_n\vdash P : \alpha$ to an element $\int{x_1:\beta_1,\dots,x_n:\beta_n\vdash P : \alpha}\in\MM_{\beta_1\to\cdots\to\beta_n\to\alpha}$.
	\end{itemize}
Due to syntax-directedness of \PCF{} type system, we can simply write $\int[\Gamma,\alpha]{P}$ for $\int{\Gamma\vdash P : \alpha}$.
\end{defi}

\begin{defi}
\bsub
\item A model $\mathscr{M}$ is \emph{sound} if, for all \PCF{} programs $P,P'$ of type $\alpha$, the following holds:
\[
	P\conv[\PCF] P' \imp \int[\alpha]{P} = \int[\alpha]{P'}
\]
\item A model $\mathscr{M}$ of \PCF{} is \emph{fully abstract} if, for all \PCF{} programs $P,P'$ of type $\alpha$:
\[
	\int[\alpha]{P} = \int[\alpha]{P'}\iff P\obseq P'
\]
The left-to-right implication is called \emph{adequacy}, the converse implication \emph{completeness}.
\esub
\end{defi}

\subsection{An adequate model of PCF}  

In this section we are going to construct a model of \PCF{} based on EAMs, and we prove that it enjoys adequacy.
The idea is to focus on the subset $\cD\subseteq\Addrs$ containing addresses of typable EAMs, which is then stratified $(\cD_\alpha)_{\alpha\in\Types}$ following the inductive syntax of simple types.
In other words, the set $\cD_\alpha$ contains the addresses of those EAM having type $\alpha$.
In particular, this means that the sets $\cD_\alpha$ are not pairwise disjoint: e.g., the address $\Lookup{\mach{I}}$ of the identity machine belongs to the set $\cD_{\alpha\to\alpha}$, for all types $\alpha$.
The well-typedness condition allows us to get rid of those EAMs containing addresses of infinite chains of `pointers'---a phenomenon described in Remark~\ref{rem:sillyEAMs}\eqref{rem:sillyEAMs2}---that might exist in $\cM$, but cannot be typed.
Subsequently, we are going to impose that two addresses $a,b\in\cD_\alpha$ are equal in the model whenever the corresponding EAMs exhibit the same applicative behavior:
at ground type ($\alpha=\tint$) this simply means that either $\Lookinv{a}$ and $\Lookinv{b}$ compute the same numeral, or they are both non-terminating.
This equality is then lifted at higher order types following the well-established tradition of logical relations~\cite{AmadioC98,Pitts97}.

The above intuitions are formalized in the following definition.

\begin{defi}\label{def:model}
\bsub
\item For all types $\alpha\in\Types$, define
	$
	\cD_\alpha = \set{ a\in\Addrs \st \Lookinv a : \alpha}.
	$
\item For all EAMs $\mM,\mN$ of type $\alpha$, define $\mM \equiv_\alpha \mN$ by structural induction on $\alpha$:
	\[
	\bar{lcl}
	\mM \equiv_\tint \mN &\iff& \forall n\in\nat\,.\,\big[\,\mM \reddh \mach{n} \iff \mN \reddh \mach{n}\,\big] \\[1ex]
	\mM \equiv_{\alpha\to\beta} \mN&\iff&\forall a,b\in\cD_\alpha\,.\,\big[\,\Lookinv a \equiv_\alpha \Lookinv b \imp\appT{\mM}{[a]} \equiv_\beta \appT{\mN}{[b]}\,\big]\\
	\ear
	\]
\item For $\alpha\in\Types$ and $a,b\in\cD_\alpha$, we write $a\simeq_\alpha b$ whenever $\Lookinv a\equiv_\alpha \Lookinv b$ holds. 
\item We let $\cD_\alpha/_{\simeq_\alpha} {=} \set{ [a]_{\simeq_\alpha}\st a \in\cD_\alpha}$.
\esub
\end{defi}
\noindent %FIXED indent
The model which is shown to be fully abstract is constructed as follows.
\begin{defi}\label{def:modelD} 
Define the model $\DD = \tuple{(\cD_\alpha/_{\simeq_\alpha})_{\alpha\in\Types},(\ \cdot^{\alpha,\beta})_{\alpha,\beta\in\Types},\int{-}}$ where
\[
	\bar{rcl}
	[a]_{\simeq_{\alpha\to\beta}}\cdot^{\alpha,\beta} [b]_{\simeq_\alpha}&=& [a\cdot b]_{\simeq_\beta}\\
	\int{x_1:\beta_1,\dots,x_n:\beta_n\vdash P : \alpha} &=& [\Lookup{\trans[\vec x]{P}}]_{\simeq_{\beta_1\to\cdots\to\beta_n\to\alpha}}
	\ear
\]
\end{defi}

The application $\cdot^{\alpha,\beta}$ is well defined by Proposition~\ref{prop:typing}\eqref{prop:typing1} and the interpretation function $\int{-}$ by Theorem~\ref{thm:transtyping}. 
Note that two \PCF{} programs $P,Q$ of type $\alpha$ have the same interpretation in the model $\DD$, i.e.\ $\int[\alpha]{P} = \int[\alpha]{Q}$, exactly when $\trans{P} \simeq_{\alpha}\trans{Q}$. 
Hence, we mainly work with translations of \PCF{} terms modulo $\simeq_{\alpha}$ and draw conclusions for $\DD$ at the end (Theorem~\ref{thm:FA}).

\begin{lem}\label{lem:equivprops}
\bsub
\item\label{lem:equivbinrel} 
	The relation $\equiv_\alpha$ is an equivalence.
\item\label{lem:equivconvh} 
	Let $\mM,\mN\in\cM$ and $\alpha\in\Types$ be such that $\mM:\alpha $ and $\mN:\alpha$. If $\mM \redh \mN$ then $\mM \equiv_\alpha \mN$.
\item\label{lem:equivtape}
	Assume $\mM : \alpha\to\beta,\mN_1 : \alpha$ and $\mN_2 : \alpha$. If $\mN_1 \redh \mN_2$ then $\appT{\mM}{[\Lookup{\mN_1}]} \equiv_\beta \appT{\mM}{[\Lookup{\mN_2}]}$.
\esub
\begin{proof}
\bsub
\item Reflexivity, symmetry, and transitivity are proven easily by induction on $\alpha$.
\item By induction on $\alpha$, using the reflexivity and transitivity properties from \eqref{lem:equivbinrel}.
\item By \eqref{lem:equivconvh} we get $\mN_1 \equiv_\beta \mN_2$, so this item follows from reflexivity and definition of $\equiv_{\beta\to\alpha}$.\qedhere
\esub
\end{proof}
\end{lem}
\noindent %FIXED indent
Everything is now in place to prove that the model is sound and adequate.

\begin{cor}[Soundness]\label{cor:sound}
For all $\PCF{}$ programs $P_1,P_2$ of type $\alpha$, we have
\[
	\trans{P_1}\convh\trans{P_2}\imp\trans{P_1}\equiv_\alpha\trans{P_2}
\]
\end{cor}
\begin{proof} By definition of interconvertibility $\convh$ and Lemma~\ref{lem:equivprops}\eqref{lem:equivconvh}.
\end{proof}

\begin{thm}[Adequacy]\label{thm:adequacy} 
Given two \PCF{} programs $P_1,P_2$ of type $\alpha$, we have 
\[
	 \trans{P_1} \equiv_\alpha \trans{P_2} \imp P_1 \obseq P_2.
\]
\end{thm}
\begin{proof}
By Proposition~\ref{prop:Abramsky} it is sufficient to show $P_1\appeq P_2$. Proceed by induction on $\alpha$.

Base case $\alpha = \tint$. For $i=1,2$, we have $P_i\redd[\PCF] \num n\iff\trans{P_i}\reddh\mach n$, by Theorem~\ref{thm:simulation}.
Then, this case follows from the assumption $\trans{P_1} \equiv_\tint \trans{P_2}$, that is $\trans{P_1}\reddh\mach n\iff\trans{P_2}\reddh\mach n$.

Case $\alpha = \beta_1\to\beta_2$. Take any \PCF{} program $\vdash Q_1 : \beta_1$.
By Definitions~\ref{def:trans} and~\ref{def:fixedmachines}, we get $\trans{P_1Q_1} = \appT{\mProj{1}{1}}{[\Lookup{\mProj{1}{1}},\Lookup{\trans{P_1}},\Lookup{\trans{Q_1}}]}\reddh \appT{\trans{P_1}}{[\Lookup{\trans{Q_1}}]}$
whence $\trans{P_1Q_1}\equiv_{\beta_2}\appT{\trans{P_1}}{[\Lookup{\trans{Q_1}}]}$.
Similarly, $\trans{P_2Q_1}\equiv_{\beta_2}\appT{\trans{P_2}}{[\Lookup{\trans{Q_1}}]}$.
From $\trans{P_1} \equiv_{\beta_1\to\beta_2} \trans{P_2}$ and $\trans{Q_1} \equiv_{\beta_1} \trans{Q_1}$ (reflexivity), we get $\appT{\trans{P_1}}{[\Lookup{\trans{Q_1}}]} \equiv_{\beta_2} \appT{\trans{P_2}}{[\Lookup{\trans{Q_1}}]}$. By transitivity of $\equiv_{\beta_2}$, we get $\trans{P_1Q_1}\equiv_{\beta_2}\trans{P_2Q_1}$ and by IH $P_1Q_1 \appeq P_2Q_1$. 
Thus, for all $Q_2,\dots,Q_k$ of the appropriate types, we get $P_1 Q_1\cdots Q_n\redd[\PCF] \num n {\iff} P_2 Q_1\cdots Q_n\redd[\PCF] \num n$. 
As $Q_1$ is arbitrary, conclude $P_1 \appeq P_2$.
\end{proof}

The following technical lemmas turn out to be useful in the remainder of the article.

\begin{lem}\label{lem:applicative} 
For \emph{($\E$)}\PCF{} programs $M,N_1,\dots,N_n$ such that $\vdash M:\alpha_1\to\cdots\to\alpha_n\to\beta$ and $\vdash N_i:\alpha_i$ for all $i\,(1\le i\le n)$, we have 
\[
	\appT{\trans{M}}{[\Lookup{\trans{N_1}},\dots,\Lookup{\trans{N_n}}]} \equiv_\beta \trans{M\cdot N_1\cdots N_n}. 
\]
\end{lem}
\begin{proof} By induction on $n$, the case $n=0$ being trivial.

Case $n>0$. By Definitions~\ref{def:trans} and~\ref{def:fixedmachines}, we get 
\[
	\trans{MN_1\cdots N_{n}} = \appT{\mProj{1}{1}}{[\Lookup{\mProj{1}{1}},\Lookup{\trans{MN_1\cdots N_{n-1}}},\Lookup{\trans{N_n}}]}\reddh \appT{\trans{MN_1\cdots N_{n-1}}}{[\Lookup{\trans{N_n}}]},
\]
whence $\trans{MN_1\cdots N_{n}}\equiv_{\beta}\appT{\trans{MN_1\cdots N_{n-1}}}{[\Lookup{\trans{N_n}}]}$ by Lemma~\ref{lem:equivprops}\eqref{lem:equivconvh}. 
By IH, we have $\appT{\trans{M}}{[\Lookup{\trans{N_1}},\dots,\Lookup{\trans{N_{n-1}}}]} \equiv_\beta \trans{MN_1\cdots N_{n-1}}$, so we conclude by Lemma~\ref{lem:equivprops}\eqref{lem:equivbinrel}.
\end{proof}

\begin{lem}[Strengthening for EAMs]\label{lem:eamstr}~\\
Let $M\in\Lam^\E$ be such that  $x_1:\alpha_1,\dots,x_n : \alpha_n \vdash M : \beta$. 
Assume that $x_i \notin\FV{M}$ for some $i\,(1\le i \le n)$. 
Then for all $a_1 \in \cD_{\alpha_1}, \dots, a_n \in \cD_{\alpha_n}$:
\[
	\appT{\trans[x_1,\dots,x_n]{M}}{[a_1,\dots,a_n]} \equiv_{\beta} 
	\appT{\trans[x_1,\dots,x_{i-1},x_{i+1},\dots,x_n]{M}}{[a_1,\dots,a_{i-1},a_{i+1},\dots,a_n]}
\]
\end{lem}
 
\begin{proof}[Proof (Appendix~\ref{app:model})] By induction on a derivation of $x_1:\alpha_1,\dots,x_n : \alpha_n \vdash M : \beta$.
\end{proof}

\begin{cor}\label{cor:closedargs} Let $M \in \Prog{\E}$ and $\alpha_1,\dots,\alpha_n,\beta\in\Types$.
If $~\vdash M : \beta$ then, for all $x_1,\dots,x_n\in\Var$ and $a_i\in\cD_{\alpha_i}$ $(1\le i \le n)$, we have \[\appT{\trans[x_1,\dots,x_n]{M}}{[a_1,\dots,a_n]\equiv_\beta \trans{M}}.\]
\end{cor}

\subsection{Definability and full abstraction}

The adequacy result established above gives one implication of the Full Abstraction property.
In order to prove the converse implication, namely completeness, we need to show that the model does not contain any undefinable element (\emph{junk}). 
This amounts to associate with any typable EAM, a \PCF{} program of the same type exhibiting the same observational behavior (\emph{mutatis mutandi}). 
The problem is that an EAM might perform useless computations (e.g., updating the value of a register that is subsequently never used) that could however prevent the machine from terminating.
To simulate the same behavior in the corresponding \PCF{} term, we use a `convergency' test $\ifc{x}{y}$ (\emph{if-converges-then}) defined by: $\ifc{x}{y} = \ifterm{x}{y}{y}$. Indeed, given \PCF{} programs $P$ and $Q$, the program $\ifc{P}{Q}$ is observationally indistinguishable from $Q$ exactly when $P$ is terminating, independently from its result.

Recall that an EAM typing context $\Delta = i_1:\beta_{i_1},\dots,i_k:\beta_{i_k}$ is a list of associations between indices of registers and types.
Moreover, the judgments $\mM : \alpha$, $\Delta\Vdash^r (P,T) : \alpha$ and $\vec R\models\Delta$ have been defined in Figure~\ref{fig:eamstyping}.

\begin{defi}[Reverse Translation] \label{def:revtrans}
Let $\mM\in\cM$, $P$ be an EAM program, $T\in\Tapes$, $\alpha\in\Types.$
Given $\mM : \alpha$ (resp.\ $\Delta\Vdash^r (P,T) : \alpha$), we associate a \PCF{} term $\rtrans[]{\mM}{\alpha}$ (resp.\ $\rtrans{P,T}{\alpha}$) defined by induction on the type-derivation as follows:
\[
\bar{l}
\rtrans[]{\mach n}{\tint} = \num n;\\[3pt]
\rtrans[]{\mY{}}{(\alpha\to\alpha)\to\alpha} = \lam x.\fix x;\\[3pt]
\rtrans[]{\tuple{ \vec R, P, T }}{\alpha} =(\lambda x_{i_1}\dots x_{i_k}.\rtrans[i_1:\beta_{i_1},\dots,i_k:\beta_{i_k}]{P,T}{\alpha})\cdot\rtrans[]{\Lookinv{\val R_{i_1}}}{\beta_{i_1}}\cdots\rtrans[]{\Lookinv{\val R_{i_k}}}{\beta_{i_k}},\\[3pt]
\hfill\textrm{where $\vec R \models i_1:\beta_{i_1},\dots,i_k:\beta_{i_k}$, for $1\le k \le r$;}\\[3pt]
\rtrans{\Load i;P,[]}{\beta\to\alpha} = \lambda x_i. \rtrans[{\Delta[i:\beta]}]{ P,[]}{\alpha};\\[3pt]
\rtrans{\Load i;P,a::T}{\alpha} = \big(\lambda x_i.\rtrans[\Delta{[i:\beta]}]{ P,T}{\alpha}\big)\cdot\rtrans[]{\Lookinv{a}}{\beta};\\[3pt]
\rtrans[\Delta, i:\tint]{\Pred ij;P,T}{\alpha} = \ifc{x_i}{(\lambda x_j.
\rtrans[{(\Delta,i:\tint)[j:\tint]}]{ P,T}{\alpha})\cdot (\pred x_i)};\\[3pt]
\rtrans[\Delta, i:\tint]{\Succ ij;P,T}{\alpha} = \ifc{x_i}{(\lambda x_j.\rtrans[{(\Delta,i:\tint)[j:\tint]}]{ P,T}{\alpha})\cdot (\succ x_i)};  \\[3pt]
\rtrans[\Delta, i:\tint,j:\beta,k:\beta]{\Ifz ijkl;P,T}{\alpha} = \ifc{x_i}{(\lambda x_l.\rtrans[{(\Delta,i:\tint,j:\beta,k:\beta)[l:\beta]}]{ P,T}{\alpha})\cdot \ifterm {x_i}{x_j}{x_k}};\\[3pt]
\rtrans[\Delta, i:\beta\to\alpha,j:\beta]{\Apply i j k ;P,T}{\gamma} = (\lambda x_k.\rtrans[{(\Delta,i:\beta\to\alpha,j:\beta)[k:\alpha]}]{ P,T}{\gamma})\cdot(x_i \cdot x_j);\\[3pt]
\rtrans[\Delta, i:\alpha_1\to\cdots\to\alpha_n\to\alpha]{\Call i,[a_1,\dots,a_n]}{\alpha} = x_i\cdot \rtrans []{\Lookinv{a_1}}{\alpha_1} \cdots \rtrans []{\Lookinv{a_n}}{\alpha_n}.
\ear
\]
\end{defi}

It is easy to check that the \PCF{} term associated with $\mM:\alpha$ is actually a \PCF{} program.

\begin{exas} Consider the EAMs introduced in Example~\ref{ex:extabsmach}. The reverse translation applied to said machines produces the following \PCF{} programs: 
\bsub
\item $\rtrans[]{\mach{I}}{\alpha\to\alpha} =\rtrans[]{\Load 0; \Call 0, []}{\alpha\to\alpha} = \lambda x_0.\rtrans[0:\alpha]{\Call 0, []}{\alpha} = \lambda x_0.x_0$.
\item $\!\!\!\bar[t]{l}\rtrans[]{\mach{\appT{Y}{[\Lookup {\mach I}]}}}{\alpha} = \rtrans[]{\Load 0;\Load 1; \Apply 0 1 0; \Apply 1 0 1; \Call 1; [\Lookup{\mY}, \Lookup{\mach I}]}{\alpha}\\= (\lam x_0.(\lam x_1. \rtrans[]{\Apply 0 1 0; \Apply 1 0 1; \Call 1; []}{\alpha})\cdot \rtrans[]{\mach I}{\alpha\to\alpha})\cdot \rtrans[]{\mY}{(\alpha\to\alpha)\to\alpha}\\
=(\lam x_0.(\lam x_1. (\lambda x'_0.(\lambda x'_1.x'_1) \cdot (x_1 \cdot x'_0))\cdot (x_0 \cdot x_1))\cdot (\lambda y_0.y_0))\cdot (\lam x. \fix x).\ear$
\item $\rtrans[]{\mach{Succ1}}{\tint\to\tint} = \lambda x_0.\rtrans[0:\tint]{\Succ 0 0; \Call 0, []}{\tint} = \lambda x_0.\ifc{x_0}{(\lambda x'_0.x'_0)\cdot (\succ x_0)}$.
\item $\rtrans[]{\mach{Succ2}}{\tint\to\tint} = \big(\lambda x_0.(\lambda x_1.(\lambda x'_1.(\lambda x''_1.x''_1) \cdot (x_0 \cdot x'_1)) \cdot (x_0 \cdot x_1))\big) \cdot \rtrans[]{\mach{Succ1}}{\tint\to\tint}$.
\esub
Note that, for all \PCF{} programs $P$ of type $\tint$, the program $\rtrans[]{\mach{Succ1}}{\tint\to\tint}\cdot P$ is well typed and converges to a natural number exactly when $P$ does.
\end{exas}

\begin{prop}\label{prop:revtranstyping}
\bsub\item
Let $\mM \in \cM$ and $\alpha\in\Types$. If $\mM:\alpha$ then $\vdash \rtrans[]{\mM}{\alpha} : \alpha$. 
\item Let $P$ be an EAM program, $T \in \mathcal{T}_\Addrs$, $\Delta = {i_1}:\alpha_{i_1},\dots,{i_k}:\alpha_{i_k}$ be a type environment and $\alpha\in\Types.$ Then, 
\[
	\Delta\Vdash^r (P,T) : \alpha\imp x_{i_1}:\alpha_{i_1},\dots,x_{i_k}:\alpha_{i_k}\vdash \rtrans[{i_1}:\alpha_{i_1},\dots,{i_k}:\alpha_{i_k}]{P,T}{\alpha} : \alpha
\]
\esub
\end{prop}
\begin{proof}[Proof (Appendix~\ref{app:model})]
Both items are proved by mutual induction on the height of a derivation of $\mM:\alpha$ and of $\Delta\Vdash^r (P,T) : \alpha$.
\end{proof}

\begin{thm}\label{thm:transrevtrans}
For all EAMs $\mM : \alpha$, we have $\trans{\rtrans[]{\mM}{\alpha}}{}{}\equiv_\alpha \mM$.
\end{thm}
\begin{proof}
One needs to consider the additional statement: 

``For all $i_1:\alpha_{i_1},\dots,i_n:\alpha_{i_n} \Vdash^r (P,T):\alpha$, $a_{i_1}\in \cD_{\alpha_{i_1}},\dots,a_{i_n}\in \cD_{\alpha_{i_n}},$ 
\[ 
	\appT{\trans[x_{i_1},\dots,x_{i_n}]{\rtrans[i_1:\alpha_{i_1},\dots,i_n:\alpha_{i_n}]{P,T}{\alpha}}}{[a_{i_1},\dots,a_{i_n}]} \equiv_\alpha \tuple{\vec{R}^r_{a_{i_1},\dots,a_{i_n}},P,T},
\]
where $\vec{R}^r_{a_{i_1},\dots,a_{i_n}}$ denotes the list of registers $R_0,\dots,R_{r}$ such that, for all $j\,(0 \leq j \leq r)$, $\val R_j = a_j$ if $j \in \set{i_1,\dots,i_n}$, and $\val R_j = \Null$ otherwise.''

The two statements are proven by mutual induction on a derivation of $\mM : \alpha$ and $\Delta \Vdash^r (P,T):\alpha$, respectively.
(The details are worked out in the Technical Appendix~\ref{app:model}.)
\end{proof}

\begin{cor}\label{corr:machinepcfequiv}
For all $a \in \cD_\alpha$, there is a \PCF{} program $\vdash P_a:\alpha$ such that $\trans{P_a} \equiv_\alpha \Lookinv{a}$.
\end{cor}

\begin{thm}[Completeness]\label{thm:completeness} 
Given two \PCF{} programs $P_1,P_2$ of type $\alpha$, we have
\[
	P_1 \obseq P_2 \Rightarrow \trans{P_1} \equiv_\alpha \trans{P_2}
\]
\end{thm}
\begin{proof} Assume $P_1 \obseq P_2$.
Let $\alpha {=} \alpha_1\to\cdots\to\alpha_n\to\tint$ and $a_1,b_1 \in \cD_{\alpha_1},\dots,$ $a_n,b_n\in \cD_{\alpha_n}$ be such that $\Lookinv{a_i} \equiv_{\alpha_i} \Lookinv {b_i}$. 
By Corollary~\ref{corr:machinepcfequiv}, there are \PCF{} programs $Q_1,\dots,Q_n$, and $Q'_1,\dots,Q'_n$ such that $Q_i,Q'_i$ have type $\alpha_i$ and both $\Lookinv{a_i} \equiv_{\alpha_i} \trans{Q_i}$ and $\Lookinv{b_i} \equiv_{\alpha_i} \trans{Q'_i}$ hold, for every such $i$. We get
\[
\bar{rcll}
\appT{\trans{P_1}}{[a_1,\dots,a_n]}&\equiv_\tint&\appT{\trans{P_1}}{[\Lookup{\trans{Q_1}},\dots,\Lookup{\trans{Q_n}}]},&\textrm{as $\trans{P_1}\equiv_\alpha \trans{P_1}$ by Lemma~\ref{lem:equivprops}\eqref{lem:equivbinrel},}\\
&\equiv_\tint&\trans{P_1\cdot Q_1\cdots Q_n},&\textrm{by Lemma~\ref{lem:applicative}},\\
\trans{P_1\cdot Q_1\cdots Q_n}\reddh \mach k &\Leftrightarrow&P_1\cdot Q_1\cdots Q_n \redd[\PCF] \num k,&\textrm{by Theorem~\ref{thm:simulation},}\\
&\Leftrightarrow& P_2\cdot Q_1\cdots Q_n \redd[\PCF] \num k,&\textrm{as }P_1 \obseq P_2,\\
&\Leftrightarrow& \trans{P_2\cdot Q_1\cdots Q_n}\reddh \mach k,&\textrm{by Theorem~\ref{thm:simulation},}\\
\trans{P_2\cdot Q_1\cdots Q_n}&\equiv_\tint&\appT{\trans{P_2}}{[\Lookup{\trans{Q_1}},\dots,\Lookup{\trans{Q_n}}]},&\textrm{by Lemma~\ref{lem:applicative}},\\
&\equiv_\tint&\appT{\trans{P_2}}{[b_1,\dots,b_n]},&\textrm{as $\trans{P_2}\equiv_\alpha \trans{P_2}$ by Lemma~\ref{lem:equivprops}\eqref{lem:equivbinrel}.}\\
\ear
\]
Conclude by transitivity (Lemma~\ref{lem:equivprops}\eqref{lem:equivbinrel}).
\end{proof}

Adequacy and completeness together yield the main result of the paper: full abstraction.

\begin{thm}[Full Abstraction]\label{thm:FA} The model $\DD$ is fully abstract for \PCF.
\end{thm}
\begin{proof} By Definition~\ref{def:modelD}, the interpretation in $\DD$ is defined by $\int[\alpha]{P} = [\Lookup{\trans{P}}]_{\simeq_{\alpha}}$.
Therefore, the full abstraction property follows directly from Theorems~\ref{thm:adequacy} and \ref{thm:completeness}.
\end{proof}

\section*{Acknowledgment}
  \noindent The authors wish to acknowledge fruitful discussions with Flavien Breuvart and Damiano Mazza.
  %% the following bibliography is generated manually for the sake of brevity
  %% only; please use a separate .bib file in your submission
\bibliographystyle{alphaurl}
\bibliography{biblio}
\newpage
\appendix
\section{Technical Appendix}
% !TEX root = ../LMCS.tex
%!TEX spellcheck = en-US
This technical appendix is devoted to providing some omitted proofs.
We often abbreviate `induction hypothesis' as IH.

\subsection{PCF and omitted proofs from Section~\ref{sec:PCF}}\label{app:regPCF}
% !TEX root = ../../LMCS.tex
%!TEX spellcheck = en-US

% !TEX root = ../../LMCS.tex
%!TEX spellcheck = en-US

In order to prove Proposition~\ref{prop:epcfbigsmallstep}, we need the following auxiliary lemma concerning w.h.-normalizing $\EPCF$ programs, that is, $\EPCF$ programs reducing to some value (by Remark~\ref{rem:EPCFprops}\eqref{rem:EPCFprops4}). 

\begin{lem}\label{lem:epcffactorization} Let $M,N,N_1,N_2\in\Prog{\E}$
\begin{enumerate}
\item\label{lem:epcffactorization1} 
	If $M\cdot N \reddwh V$, then there exist $M', \sigma$ such that $M\reddwh (\lam x.M')^\sigma$ is a shorter reduction;
\item\label{lem:epcffactorization2} 
	If $\pred M \reddwh V$ then there is a shorter reduction $M\reddwh \num n$, for some $\num n \in \nat$; 
\item\label{lem:epcffactorization3} 
	If $\succ M \reddwh V$ then there is a shorter reduction $M\reddwh \num n$, for some $\num n \in \nat$; 
\item\label{lem:epcffactorization4} 
	If $\ifterm{M}{N_1}{N_2} \reddwh V$, then there is a shorter reduction $M\reddwh \num n$, for some $\num n \in \nat$.
\end{enumerate}
\end{lem}
\begin{proof} (1) Since $V$ cannot be an application, the reduction $M\cdot N \reddwh V$ is non-empty. 
If $M\cdot N$ is a redex then $M$ must have shape $(\lam x.M')^\sigma$ for some $M',\sigma$ and we are done.
Otherwise, the reduction has shape $M\cdot N \redwh M'\cdot N\reddwh V$ for some $M\redwh M'$.
By repeating this reasoning we must obtain $M\reddwh V'$, as the global reduction is finite.
We conclude since $V'\cdot N$ must be a redex, and this is only possible if $V'  = (\lam x.M')^\sigma$.

(1-3) Analogous.
\end{proof}

\begin{proof}[Proof of Proposition~\ref{prop:epcfbigsmallstep}] 
$(\Rightarrow)$ By induction on the height of a derivation of $M\ToE V$:
\begin{itemize}
\item $\num n \ToE \num n$: Base case, by reflexivity $\num n \reddwh \num n$ holds.
\item $(\lambda x.M)^\sigma \ToE (\lambda x.M)^\sigma$: Base case, as above.
\item $x^\sigma \ToE V$ with $\sigma(x) = N$: Then $x^\sigma\redwh N$. Conclude by applying the IH to $N\ToE V$.
\item $(\pred M)^\sigma\ToE \mathbf{0}$: We have $(\pred M)^\sigma\redwh \pred(M^\sigma)$. 
By IH we obtain $M^\sigma\reddwh \mathbf{0}$, from which it follows $\pred(M^\sigma) \reddwh\pred \mathbf{0} \redwh \mathbf{0}$.
\item $(\pred M)^\sigma\ToE\num n$: proceed as above, using $\pred(\num{n+1})\redwh \num n$.
\item $(\succ M)^\sigma\ToE\num {n+1}$: We have $(\succ M)^\sigma \redwh \succ(M^\sigma)$. 
By IH we obtain $M^\sigma\reddwh \num n$, so we conclude, remembering that $\num n = \succ^n(\mathbf{0})$.
\item $(\ifterm LMN)^\sigma\ToE V_1$: Then, we have $(\ifterm LMN)^\sigma \redwh \ifterm{L^\sigma}{M^\sigma}{N^\sigma}$. 
By IH we get $L^\sigma\reddwh \num 0$, so $ \ifterm{L^\sigma}{M^\sigma}{N^\sigma}\reddwh M^\sigma$. 
Conclude since $M^\sigma \reddwh V_1$ holds by IH.
\item $(\ifterm LMN)^\sigma\ToE V_2$: Proceed as above.
\item $(\fix M)^\sigma\ToE V$: We have $(\fix M)^\sigma\redwh \fix (M^\sigma)\redwh M^\sigma \cdot (\fix (M^\sigma))$. Conclude by IH.
\item $(M\cdot N)^\sigma\ToE V$: We have $(M\cdot N)^\sigma\redwh M^\sigma\cdot N^\sigma$. 
From the IH, we obtain $M^\sigma\reddwh (\lam x.M')^{\sigma'}$. 
Since $ (\lam x.M')^{\sigma'}\cdot N^\sigma\redwh (M')^{\sigma'}\esubst{x}{N^\sigma}$, we may conclude using the IH.
\end{itemize}

$(\Leftarrow)$ We proceed by induction on the length $k = \len{M\reddwh V}$.

Case $k=0$. Then $M = V$. Apply either $\Enumrule$ or $\Elamrule$, depending on the shape of $V$.

Case $k > 0$. Then $M \redwh N \reddwh V$, for some $N\in\Lame$ with $\len{N \reddwh V} = k-1 < k$. By Remark~\ref{rem:EPCFprops}\eqref{rem:EPCFprops3}, there exist a unique evaluation context $\E\square$ and $M',N'\in\Lame$ such that $M = \E[M']$, $N = \E[N']$ and $M' \to N'$, where $\to$ represents the contraction of  a computation or a percolation redex.
Proceed by cases on the structure of $\E\square$.

Subcase $\E\square = \square$ ($M=M',N = N'$). First, consider the computation reduction cases:
\begin{itemize}
\item $M = (\lam x.M_1)^\sigma M_2$ and $N = M_1^\sigma\esubst{x}{M_2}$:  
By $\Elamrule$ we obtain $(\lam x.M_1)^\sigma\ToE (\lam x.M_1)^\sigma$ and, from the IH, $M_1^\sigma\esubst{x}{M_2}\ToE V$.
We apply $\Ebetarule$ to infer that $(\lam x.M_1)^\sigma M_2 \ToE V$.
\item $M = \pred \mathbf{0}$ and $N = \mathbf{0}$: By $\Enumrule$ we obtain $\mathbf{0}\ToE\mathbf{0}$. We apply $\Epredzrule$ to infer $\pred\mathbf{0}\ToE\mathbf{0}$.
\item The other computation cases for $\pred$ and $\succ$ are analogous to the above case.
\item $M=\ifterm{\mathbf{0}}{L_1}{L_2} $ and $N =  L_1$: By $\Enumrule$ we obtain $\mathbf{0}\ToE\mathbf{0}$, and by the IH $L_1 \ToE V$ holds. We apply $\Eifzzrule$ to infer that $\ifterm{\mathbf{0}}{L_1}{L_2}\ToE V$.
\item $M=\ifterm{\num {n+1}}{L_1}{L_2} $ and $N =  L_2$: Analogous to the above case.
\item $M = \fix L$ and $N = L\cdot (\fix L) $: By IH we obtain $L\cdot (\fix L) \ToE V$. We apply $\Efixrule$ to infer $\fix L \ToE V$.
\end{itemize}
Second, the percolation reduction steps:
\begin{itemize}
\item $M = x^\sigma $ and $\sigma(x) = N$: By IH we obtain $N\ToE V$. We apply $\Evarrule$ to infer $x^\sigma\ToE V$.
\item $M = \mathbf{0}^\sigma$ and $N = \mathbf{0}$: We apply $\Enumrule$ to infer $\mathbf{0}^\sigma\ToE \mathbf{0}$.
\item $M = (\pred M_1)^\sigma$ and $N = \pred (M_1^\sigma)$: By Lemma~\ref{lem:epcffactorization}\eqref{lem:epcffactorization2}, $\pred (M_1^\sigma)\reddwh V$ entails that there is a shorter reduction $M_1^\sigma\reddwh \num n$, for some $n\in\nat$. 
By IH, we get $M_1^\sigma\ToE \num n$ and we conclude by applying $\Epredzrule$ or $\Epredrule$ depending on whether $n=0$ or not.
\item Proceed as in the above case for all other percolation reduction steps.
\end{itemize}
\noindent %FIXED indent
Subcase $\E\square = \pred (\E')$, for some context $\E'\square$. As $M = \pred (\E'[M'])$ reduces to a value, by Lemma~\ref{lem:epcffactorization}\eqref{lem:epcffactorization2} we must have a shorter reduction $\E'[M']\reddwh \num n$, for some $n\in\nat$. By IH we obtain $\E'[M']\ToE \num n$, so we conclude by applying either $\Epredzrule$ or $\Epredrule$.

Subcase $\E\square = \succ (\E')$, for some $\E'\square$. It follows analogously from IH and $\Esuccrule$.

Subcase $\E\square = \E'L$, for some $\E'\square$. Since $M\reddwh V$, the reduction must factorize as 
\[
	M = (\E'[M'])L\reddwh (\lam x.M_1)^\sigma L \redwh M_1^\sigma\subst{x}{L}\reddwh V
\]
with both $\E'[M']\reddwh (\lam x.M_1)^\sigma$ and $M_1^\sigma\subst{x}{L}\reddwh V$ shorter than $k$ (see Lemma~\ref{lem:epcffactorization}\eqref{lem:epcffactorization1}).
By IH, we get $\E'[M']\ToE (\lam x.M_1)^\sigma$ and $M_1^\sigma\subst{x}{L}\ToE V$, so we conclude by $\Ebetarule$.

Subcase $\E\square = \ifterm{\E'}{M_1}{M_2}$, for some $\E'\square$. Since $M\reddwh V$, Lemma~\ref{lem:epcffactorization}\eqref{lem:epcffactorization4} entails that the reduction must factorize as
\[
	M = \ifterm{\E'[L]}{M_1}{M_2}\reddwh \ifterm{\num n}{M_1}{M_2}\redwh
	\begin{cases}
	M_1 \reddwh V, & \textrm{if } n = 0,\\ 
	M_2\reddwh V, &\textrm{otherwise},\\
	\end{cases}
\]
with $\E'[L]\reddwh \num n$ and $\len{M_i\reddwh V} < k$. By IH, we get either [$\E'[L]\ToE \mathbf{0}$ and $M_1\ToE V$] or [$\E'[L]\ToE \num n+1$ and $M_2\ToE V$] (for some $n \in \nat$). We conclude by either $\ifzzrule$ or $\ifzrule$.
\end{proof}

\begin{proof}[Proof of Lemma~\ref{lem:epcftyping}]
\bsub
\item (Syntax directedness) Trivial proof by inspection.
\item (Strengthening) An easy proof by induction on the shape of $M$ proves this in both directions - the statement holds for $\Gamma\vdash \mathbf{0}:\tint$ and $\Gamma, x:\alpha\vdash x:\alpha$, and all other cases are derived from those two. 
\item (Subject reduction) We will prove this by induction on the shape of $M = \E[M']$ where $M'$ is the contracted redex. 
The only interesting cases are when $M'$ is in head position, i.e. $\E\square = \square$. For cases where $\sigma$ is present, we will assume that $\dom(\sigma) = \{y_1,\dots, y_n\}$.
\begin{itemize}
\item Cases $M = x$ and $M=\mathbf{0}$ do not apply, as neither can reduce.
\item Case $M = (\lam x.M_1)^\sigma\cdot M_2$ and $N = M_1^\sigma\esubst{x}{M_2}$. 
From $\vdash (\lam x.M_1)^\sigma\cdot M_2 : \alpha$ we obtain the judgements 
$\vdash \sigma(y_1) : \gamma_1,\dots,\vdash \sigma(y_n):\gamma_n$, 
$y_1 : \gamma_1,\dots, y_n:\gamma_n, x:\beta \vdash M_1 : \alpha$ 
and $\vdash M_2 : \beta$, via a single application of the $(\to_\mathrm{E})$ rule followed by repeated applications of the $(\sigma)$ rule. 
Using all of these judgements we can then derive the judgement $\vdash M_1\esubst{x}{M_2} : \alpha$, through $n+1$ applications of the $(\sigma)$ rule.
\item All other computation reduction cases are trivial, as they are identical to $\PCF$ where subject reduction holds.
\item Case $M = x^\sigma$, where $\sigma(x) = N$: Using $(\sigma)$ we obtain the judgement $\vdash N : \alpha$.
\item Case $M = y^\sigma$, where $y \notin \dom(\sigma)$: Does not apply, as this term is not closed.
\item Case $M = \mathbf{0}^\sigma$: Trivial, as this has the type $\tint$ and $\mathbf{0}$ has the type $\tint$.
\item Case $M = (\pred M_1)^\sigma$: We obtain the judgements $\vdash \sigma(y_1) : \gamma_1,\dots,\vdash \sigma(y_n):\gamma_n$, $y_1:\gamma_1,\dots,y_n:\gamma_n\vdash M_1 : \tint$. We can reassemble these to form the judgement $\vdash M_1^\sigma:\tint$, from which we derive $\vdash \pred(M_1^\sigma):\tint$. 
\item All other percolation reduction cases proceed identically to $(\pred M_1)^\sigma$.
\end{itemize}
\noindent %FIXED indent
The cases where $\E\square\neq \square$ are trivial applications of the induction hypothesis. \qedhere
\esub
\end{proof}

\begin{proof}[Proof of Proposition~\ref{prop:MreddNthenMreddN}] 
\bsub
\item 
By induction on a derivation of $M \redcr N$.
\begin{itemize}\item Base cases.
	\begin{itemize}
	\item Case $M = (\lam x.L_1)^\sigma\cdot L_2$ and $N = L_1^\sigma \esubst{x}{L_2}$ where, say, $\sigma = \esubst{\vec y\,}{\vec Y}$ with $x\notin \vec y$. Then, we have:
	\[
		\bar{lclcl}
		\uf M &=& (\lam x.\uf L_1)\subst{\vec y\,}{\uf{\vec Y}}\cdot \uf L_2
		= (\lam x.\uf L_1\subst{\vec y\,}{\uf{\vec Y}})\cdot \uf L_2&\to_\PCF& \uf L_1\subst{\vec y\,}{\uf{\vec Y}}\subst{x}{\uf L_2}\\
		&=& \uf{(L_1^\sigma\esubst{x}{L_2})}= \uf N
		\ear
	\]

	\item Case $M = \ifterm{\num 0}{N}{L}$. Then $\uf L = \ifterm{\num 0}{\uf{N}}{\uf L}\to_\PCF \uf{N}$.
	\item Case $M = \ifterm{\num{n+1}}{L}{N}$. Then $\uf L = \ifterm{\num{n+1}}{\uf L}{\uf{N}}\to_\PCF\uf{N}$.
	\item All other base cases hold trivially.
	\end{itemize}
\item Case $M = M_1\cdot M_2$ and $N = N_1\cdot M_2$ with $M_1\redcr N_1$. By induction hypothesis, we have $\uf M_1 \to_\PCF \uf N_1$. Therefore $\uf{M} = \uf{M_1}\cdot \uf{M_2}\to_\PCF \uf{N_1}\cdot \uf{M_2} = \uf{(N_1\cdot M_2)} = \uf N$.
\item Case $M = \ifterm{M_1}{L_1}{L_2}$ and $N = \ifterm{N_1}{L_1}{L_2}$ with $M_1\redwh N_1$. 
By induction hypothesis, we have $\uf{M_1}\to_\PCF \uf{N_1}$. Thus $\uf{M} = \ifterm{\uf M_1}{\uf L_1}{\uf L_2} \to_\PCF \ifterm{\uf N_1}{\uf L_1}{\uf L_2} = \uf{(\ifterm{N_1}{L_1}{L_2})} = \uf N$.
\item Case $M =\pred M_1$ and $N = \pred N_1$ with $M_1\redcr N_1$. By induction hypothesis, we have $\uf{M_1} \to_\PCF \uf{N_1}$. Thus $\uf M = \pred (\uf{M_1}) \to_\PCF \pred (\uf{N_1}) = \uf{(\pred N_1)} = \uf N$.
\item Case $M =\succ M_1$ and $N = \succ N_1$ with $M_1\redcr N_1$.  Analogous.
\end{itemize}

\item By induction on a derivation of $M \redsc N$.
\begin{itemize}\item Base cases.
	\begin{itemize}
	\item Case $M = (x\esubst{x}{N})^\sigma$. Since $N\in\Prog{\E}$, we have $\FV{N} = \emptyset$. Thus $\uf M = x\subst{x}{\uf{N}} = \uf{N}$.
	\item Case $M = (y\esubst{x}{L})^\sigma$ and $N = y^\sigma$. Since $M\in\Prog{\E}$ we must have $y\in\dom(\sigma)$ with, say, $\sigma(y) = L'\in\Prog{\E}$. Conclude since $\uf M = \uf{L'} = \uf{(y^\sigma)} = \uf{N}$.
	\item All other base cases hold trivially.
	\end{itemize}
\item Case $M = M_1\cdot M_2$ and $N = N_1\cdot M_2$ with $M_1\redsc N_1$. By induction hypothesis, we have $\uf M_1 = \uf N_1$. Therefore $\uf{M} = \uf{M_1}\cdot \uf{M_2}= \uf{N_1}\cdot \uf{M_2} = \uf{(N_1\cdot M_2)} = \uf N$.
\item Case $M = \ifterm{M_1}{L_1}{L_2}$ and $N = \ifterm{N_1}{L_1}{L_2}$ with $M_1\redsc N_1$. 
By induction hypothesis, we have $\uf{M_1}= \uf{N_1}$. Thus $\uf{M} = \ifterm{\uf M_1}{\uf L_1}{\uf L_2} = \ifterm{\uf N_1}{\uf L_1}{\uf L_2} = \uf{(\ifterm{N_1}{L_1}{L_2})} = \uf N$.
\item Case $M =\pred M_1$ and $N = \pred N_1$ with $M_1\redsc N_1$. By induction hypothesis, we have $\uf{M_1} = \uf{N_1}$. Thus $\uf M = \pred (\uf{M_1}) = \pred (\uf{N_1}) = \uf{(\pred N_1)} = \uf N$.
\item Case $M =\succ M_1$ and $N = \succ N_1$ with $M_1\redsc N_1$.  Analogous.\qedhere
\end{itemize}
\esub
\end{proof}

\subsection{Omitted proofs from Section~\ref{sec:Simulation}}\label{app:Sim}
% !TEX root = ../../LMCS.tex
%!TEX spellcheck = en-US

\begin{proof}[Proof of Theorem~\ref{thm:transtyping}.]
Proceed by induction on a derivation of $\Gamma \vdash M:\alpha$. We split into cases depending on the last applied rule.

\begin{itemize}
\item Case $(0)$. In this case $M = \mathbf{0}$ and $\alpha = \tint$. 
By Definition~\ref{def:trans}, we have $\trans[\vec x]{\mathbf{0}} = \appT{\mProj{n+1}{1}}{[0]}$.
By Lemma~\ref{lem:welltypedeams}\eqref{lem:welltypedeams1}, we know that $\mProj{n+1}{1} : \tint\to \vec\delta\to\tint$.
By rule $(\mathrm{nat})$, we get $ \mach{0} : \tint$.
By Proposition~\ref{prop:typing}\eqref{prop:typing1}, we conclude $ \appT{\mProj{n+1}{1}}{[0]} : \vec\delta\rightarrow\tint$;

\item Case $\mathrm{(ax)}$. Then $M = x_i$ for some $x_i\in\vec x$ and $\trans[\vec x]{x_i} = \mProj{n}{i}$. By Lemma~\ref{lem:welltypedeams}\eqref{lem:welltypedeams1}.
\item Case $(\to_\mathrm{E})$. Then $M = M_1\cdot M_2$ and there is $\beta\in\Types$ such that $\Gamma\vdash M_1 : \beta\rightarrow\alpha, \Gamma\vdash M_2:\beta$.
From the induction hypothesis, we obtain that $\trans[\vec x]{M_1} : \vec\delta\to\beta\to\alpha$ and $\trans[\vec x]{M_2} : \vec\delta\to\beta$.
By Lemma~\ref{lem:welltypedeams}\eqref{lem:welltypedeams1}, we get $\mProj{1}{1} : (\beta\to\alpha)\to\beta\to\alpha$. 
By Lemma~\ref{lem:welltypedeams}\eqref{lem:welltypedeams2}, we know that $\mAppn{n}{2} : ((\beta\to\alpha)\to\beta\to\alpha)\to(\vec\delta\to\beta\to\alpha)\to\vec\delta\to\beta$. 
By Definition~\ref{def:trans}, we have $\trans[\vec x]{M_1\cdot M_2} = \appT{\mAppn{n}{2}}{[\Lookup{\mProj{1}{1}}, \Lookup{\trans[\vec x]{M_1}},\Lookup{\trans[\vec x]{M_2}}]}$ so we conclude by Proposition~\ref{prop:typing}\eqref{prop:typing1}.
\item Case $(\sigma)$. Then $M = M'\esubst{y}{N}$ with $\Gamma,y:\beta\vdash M' : \alpha$ and $\vdash N : \beta$, for some $\beta\in\Types$.
By induction hypothesis, we get $\trans[\vec x,y]{M'} : \beta\to\vec \delta\to\alpha$ and $ \trans{N} : \beta$.
By Definition~\ref{def:trans}, we have $\trans[\vec x]{M'\esubst{y}{N}} = \appT{\trans[y,\vec x]{M'}}{[\Lookup{\trans{N}}]}$. Conclude by Proposition~\ref{prop:typing}\eqref{prop:typing1}.
\item Case $(\to_\mathrm{I})$. Then $M = \lam y.N$ and $\alpha = \alpha_1\to\alpha_2$, with $\Gamma,y:\alpha_1\vdash N : \alpha_2$.
By Definition~\ref{def:trans}, we have $\trans[\vec x]{\lambda y.N} = \trans[\vec x,y]{N}$ so the case follows from the induction hypothesis.
\item Case $(+)$. Then $M = \succ N$ and $\alpha = \tint$, with $\Gamma\vdash N : \tint$. 
By induction hypothesis $ \trans[\vec x]{N} : \vec\delta\to\tint$ and, by Lemma~\ref{lem:welltypedeams}\eqref{lem:welltypedeams4}, we have $ \mSucc : \tint\to\tint$.
By Definition~\ref{def:trans}, $\trans[\vec x]{\succ N} = \appT{\mAppn{n}{1}}{[\Lookup{\mSucc},\Lookup{\trans[\vec x]{N}}]}$. Conclude by Lemma~\ref{lem:welltypedeams}\eqref{lem:welltypedeams2} and Proposition~\ref{prop:typing}\eqref{prop:typing1}.

\item Case $(-)$. Analogous to the previous case, using Lemma~\ref{lem:welltypedeams}\eqref{lem:welltypedeams3} instead of Lemma~\ref{lem:welltypedeams}\eqref{lem:welltypedeams2}.
\item Case $(\mathrm{ifz})$. Then $M = \ifterm L{N_1}{N_2}$ with $\Gamma\vdash L : \tint$ and $\Gamma\vdash N_i : \alpha$, for each $i=1,2$.
By induction hypothesis, we get $ \trans[\vec x]{L} : \vec\delta\to \tint$ and $ \trans[\vec x]{N_i} : \vec\delta\to\alpha$, for every such $i$.
By Lemma~\ref{lem:welltypedeams}\eqref{lem:welltypedeams5}, we have $ \mIfz : \tint\to\alpha\to\alpha\to\alpha$.
Since, by Definition~\ref{def:trans}, $\trans[\vec x]{\ifterm L{N_1}{N_2}} = \appT{\mAppn{n}{3}}{[\Lookup{\mIfz},\Lookup{\trans[\vec x]{L},\Lookup{\trans[\vec x]{N_1}},\Lookup{\trans[\vec x]{N_2}}}]}$ we conclude by applying Lemma~\ref{lem:welltypedeams}\eqref{lem:welltypedeams2} and Proposition~\ref{prop:typing}\eqref{prop:typing1}.
\item Case $(\mathrm{Y})$. Then $M = \fix N$ with $\Gamma\vdash N : \alpha\to\alpha$.
By induction hypothesis, we have $\trans[\vec x]{N} : \vec\gamma\to\alpha\to\alpha$. 
By rule $(\mathrm{fix})$ we know that $\mY : (\alpha \to \alpha) \to \alpha$.
The result follows from $\trans[\vec x]{\fix N} = \appT{\mAppn{n}{1}}{[\Lookup{\mY}, \Lookup{\trans[\vec x]{N}}]}$ using Lemma~\ref{lem:welltypedeams}\eqref{lem:welltypedeams2} and Proposition~\ref{prop:typing}\eqref{prop:typing1}.\qedhere
\end{itemize}
\end{proof}
% !TEX root = ../../LMCS.tex
%!TEX spellcheck = en-US

\begin{proof}[Proof of Proposition~\ref{prop:smallstep}.]

By induction on a derivation of $M \redwh N$.
\begin{itemize}
\item Base cases. In the following $\sigma = \esubst{x_1}{N_1}\cdots\esubst{x_n}{N_n}$ for some $n>0$. 
As a matter of notation, we introduce the abbreviations $\vec{x} = x_1,\dots,x_n$ and ${\Lookup{\vec\sigma}} = \Lookup{\trans{N_1}{}{}},\dots,\Lookup{\trans{N_n}{}{}}$.
	\begin{itemize}
	\item Case $M = ((\lambda y.M_1)^\sigma)M_2$ and $N = M_1^\sigma\esubst{x}{M_2}$. Wlog $y\notin\dom(\sigma)$ and since $M_1\in\Prog{\E}$, we must have $\FV{M_1}\subseteq\set{y}$. Therefore
	\[
		\bar{lcll}
		\trans{M}\!&=&\trans{((\lambda y.M_1)^\sigma)M_2}\\
		&=&\appT{\mAppn{0}{2}}{[\Lookup{\trans[y]{M_1^\sigma}}, \Lookup{\trans{M_2}}]},&\textrm{by Definition~\ref{def:trans}},\\
		&=&\appT{\mProj{1}{1}}{[\Lookup{\trans[y]{M_1^\sigma}}, \Lookup{\trans{M_2}}]},&\textrm{since $\mAppn{0}{2} = \mProj{1}{1}$ by Definition~\ref{def:fixedmachines}},\\
		&\reddh^{2}& \appT{\trans[y]{M_1^\sigma}}{[\Lookup{\trans{M_2}}]},&\textrm{by Lemma~\ref{lem:existenceofeams}\eqref{lem:existenceofeams1}},\\
		&=& \trans{M_1^\sigma\esubst{y}{M_2}},&\textrm{by Definition~\ref{def:trans}}.\\
		\ear
	\]
	\item Case $M = \mathbf{0}^\sigma$ and $N = \mathbf{0}$. 
	On the one side, $\trans{\mathbf{0}^\sigma} = \trans[\vec x]{\mathbf{0}} = \appT{\mProj{n+1}{1}}{[0,\Lookup{\vec\sigma}]}$ whence, by Lemma~\ref{lem:existenceofeams}\eqref{lem:existenceofeams1}, we get $\appT{\mProj{n+1}{1}}{[0,\Lookup{\vec\sigma}]}\reddh^{n + 2} \Lookinv{0} =  \mach{0}$. 
	On the other side, by Lemma~\ref{lem:existenceofeams}\eqref{lem:existenceofeams1}, we obtain $\trans{\mathbf{0}} = \appT{\mProj{1}{1}}{[0]}\reddh^{2} \mach{0}$. Since $n>0$, we conclude $\trans{\mathbf{0}^\sigma}\convg\trans{\mathbf{0}}$.
	\item Case $M = \ifterm{\mathbf{0}}{M_1}{M_2}$ and $N = M_1$. Therefore $\trans{M}$ is equal to
	\[
		\bar{lcll}
		 \trans{\ifterm{\mathbf{0}}{M_1}{M_2}}&=&\appT{\mAppn{0}{3}}{[\Lookup{\mIfz},0,\Lookup{\trans{M_1}},\Lookup{\trans{M_2}}]},&\textrm{by Definition~\ref{def:trans}},\\
			&=&\appT{\mProj{1}{1}}{[\Lookup{\mIfz},0,\Lookup{\trans{M_1}},\Lookup{\trans{M_2}}]},&\textrm{by Definition~\ref{def:fixedmachines}},\\
			&\reddh^{2}&\appT{\mIfz}{[0,\Lookup{\trans{M_1}},\Lookup{\trans{M_2}}]},&\textrm{by Lemma~\ref{lem:existenceofeams}\eqref{lem:existenceofeams1}},\\
			&\reddh&\trans{M_1}.
		\ear
	\]
	\item Case $M = \ifterm{\num{k+1}}{M_1}{M_2}$, for some $k\ge0$, and $N = M_2$. Therefore $\trans{M}$ is equal to
	\[
		\bar{llll}
			&\mathrlap{\trans{\ifterm{\num{k+1}}{M_1}{M_2}}}\\&=&\appT{\mAppn{0}{3}}{[\Lookup{\mIfz},k+1,\Lookup{\trans{M_1}},\Lookup{\trans{M_2}}]},&\textrm{by Definition~\ref{def:trans}},\\
			&=&\appT{\mProj{1}{1}}{[\Lookup{\mIfz},k+1,\Lookup{\trans{M_1}},\Lookup{\trans{M_2}}]},&\textrm{by Definition~\ref{def:fixedmachines}},\\
			&\reddh^{2}&\appT{\mIfz}{[k+1,\Lookup{\trans{M_1}},\Lookup{\trans{M_2}}]},&\textrm{by Lemma~\ref{lem:existenceofeams}\eqref{lem:existenceofeams1}},\\
			&\reddh&\trans{M_2}.
		\ear
	\]	
	\item Case $M = y^\sigma$ for $y\notin\dom(\sigma)$. Vacuous, since $M\in\Prog{\E}$.
	\item Case $M = x_i^\sigma$ and $N = \sigma(x_i) = N_i$. So $\trans{x^\sigma} = \appT{\mProj{n}{i}}{[\Lookup{\vec\sigma}]}\reddh^{n+1}\trans{N_i}$ by Lemma~\ref{lem:existenceofeams}\eqref{lem:existenceofeams1}.
	\item Case $M = (M_1\cdot M_2)^\sigma$ and $N = M_1^\sigma\cdot M_2^\sigma$. Since $M\in\Prog\E$, each $M_i^\sigma$ is closed, whence $\FV{M_i}\subseteq\set{x_1,\dots,x_n}$. On the one side, we get
	\[
	\bar{llll}
		&\mathrlap{\trans{(M_1\cdot M_2)^\sigma} = \mAppn{n}{2}@[\Lookup{\mProj{1}{1}},\Lookup{\trans[\vec x]{M_1}},\Lookup{\trans[\vec x]{M_2}},\Lookup{\vec\sigma}]}\\
		&\reddh^{10n + 2}& \appT{\mProj{1}{1}}{[\Lookup{(\trans[\vec x]{M_1}@[\Lookup{\vec\sigma}])},\Lookup{(\trans[\vec x]{M_2}@[\Lookup{\vec\sigma}])}]},&\textrm{by Lemma~\ref{lem:existenceofeams}\eqref{lem:existenceofeams2}},\\
		&\reddh^{2}& \trans{M_1^\sigma} @[\Lookup{\trans{M_2^\sigma}}],&\textrm{by Lemma~\ref{lem:existenceofeams}\eqref{lem:existenceofeams1}}.\\
	\ear
	\]
	On the other side, we get $\trans{M_1^\sigma\cdot M_2^\sigma} = \mProj{1}{1}@[\Lookup{\mProj{1}{1}},\Lookup{\trans{M_1^\sigma}},\Lookup{\trans{M_2^\sigma}}]\reddh^{4} \trans{M_1^\sigma}@[\Lookup{\trans{M_2^\sigma}}]$ by applying Lemma~\ref{lem:existenceofeams}\eqref{lem:existenceofeams1} (twice). 
	Conclude since $10n+4 > 4$ when $n>0$.
	\item Case $M = \pred \mathbf{0}$ and $N=\mathbf{0}$. 
	By Lemma~\ref{lem:existenceofeams}\eqref{lem:existenceofeams1} and \eqref{lem:existenceofeams3}, we have
	\[
	\bar{rl}
	\trans{\pred \mathbf{0}} =&\appT{\mProj{1}{1}}{[\Lookup{\mPred},0]}\reddh^2\appT{\mPred}{[0]}\\\redh& \Tuple{R_0 = 0, \Pred 00;\Call 0,[]}\\
	\redh&\Tuple{R_0 = 0, \Call 0,[]} \redh \mach{0}=\trans{ \mathbf{0}}
	\ear
	\]
	\item Case $M = \pred(\succ \num n)$ and $N=\num n$.
	By applying Lemma~\ref{lem:existenceofeams}\eqref{lem:existenceofeams1} and \eqref{lem:existenceofeams4}, we obtain $\trans{\succ \num n} = \appT{\mProj{1}{1}}{[\Lookup{\mSucc},n]} \reddh^2 \appT{\mSucc}{[n]}\redh$\\$\Tuple{R_0 = n, \Succ 00;\Call 0,[]}\redh\Tuple{R_0 = n+1,\Call 0,[]}\redh \mach{n+1}$.\\
	Using Lemma~\ref{lem:existenceofeams}\eqref{lem:existenceofeams3} as well, we conclude
	\[
	\bar{rl}	
	\trans{\pred(\succ \num n)} = &\appT{\mProj{1}{1}}{[\Lookup{\mPred},\Lookup{\trans{\succ \num n}}]} \reddh^2 \appT{\mPred}{[\Lookup{\trans{\succ \num n}}]}\\
	\redh&\Tuple{R_0 = \Lookup{\trans{\succ \num n}, \Pred 00;\Call 0,[]}}\\
	 \reddh& \Tuple{R_0 = n+1, \Pred 00;\Call 0,[]}\\
	 \reddh& \Tuple{R_0 = n, \Call 0,[]}	\redh \mach{n}=\trans{\num n}
	\ear
	\]
	\item Case $M = \fix M'$ and $N = M'(\fix M')$. In this case, we get
	\[
	\bar{lcll}
	\trans{\fix M'} &=&\appT{\mY}{[\Lookup{\trans{M}}]},&\textrm{by Definition~\ref{def:trans}},\\
	&\reddh^5&\appT{\trans{M'}}{[\Lookup{\appT{\mY}{[\Lookup{\trans{M'}}]}}]},&\textrm{by Lemma~\ref{lem:existenceofeams}\eqref{lem:existenceofeams6}},\\
	&=&\appT{\trans{M'}}{[\Lookup{\trans{\fix M'}}]},&\textrm{by Definition~\ref{def:trans}},\\
	&{}^2_\mach{c}\!\!\twoheadleftarrow&\appT{\mProj{1}{1}}{[\Lookup{\trans{M'}},\Lookup{\trans{\fix M'}}]},&\textrm{by Lemma~\ref{lem:existenceofeams}\eqref{lem:existenceofeams1}},\\
	&{}^2_\mach{c}\!\!\twoheadleftarrow&\appT{\mProj{1}{1}}{[\Lookup{\mProj{1}{1}},\Lookup{\trans{M'}},\Lookup{\trans{\fix M'}}]},&\textrm{by Lemma~\ref{lem:existenceofeams}\eqref{lem:existenceofeams1}},\\	
	&=&\trans{M'(\fix M')},&\textrm{by Definition~\ref{def:trans}}.\\
	\ear
	\]
	\item Case $M = (\ifterm{L}{M_1}{M_2})^\sigma$ and $N = \ifterm{L^\sigma}{M_1^\sigma}{M_2^\sigma}$. Note that $M\in\Prog{\E}$ entails $L^\sigma,M_1^\sigma,M_2^\sigma$ closed, thus $\FV{L}\subseteq\set{\vec x}$ and $\FV{M_i}\subseteq\set{\vec x}$.  By Lemma~\ref{lem:existenceofeams}\eqref{lem:existenceofeams2}, we get
	\[
	\bar{rcl}
	\trans{(\ifterm{L}{M_1}{M_2})^\sigma}&=&\mAppn{n}{3}@[\Lookup{\mIfz},\Lookup{\trans[\vec x]{L}},\Lookup{\trans[\vec x]{M_1}},\Lookup{\trans[\vec x]{M_2}},\Lookup{\vec\sigma}]\\
	&\reddh^{13n+2}& \appT{\mIfz}{[\Lookup{\trans{L^\sigma}},\Lookup{\trans{M_1^\sigma}},\Lookup{\trans{M_2^\sigma}}]}\\
	&{}^2_\mach{c}\!\!\twoheadleftarrow&\mProj{1}{1}@[\Lookup{\mIfz},\Lookup{\trans{L^\sigma}},\Lookup{\trans{M_1^\sigma}},\Lookup{\trans{M_2^\sigma}}]\\&=&\trans{\ifterm {L^\sigma}{M_1^\sigma}{M_1^\sigma}}
	\ear
	\]
	Conclude since $13n+2 > 2$ when $n>0$.
	\item Case $M = (\pred M')^\sigma$ and $N = \pred{((M')^\sigma)}$. Analogous to the above.
\item Case $M = (\succ M')^\sigma$ and $N = \succ{((M')^\sigma)}$. Analogous to the above.
\item Case $M = (\fix M')^\sigma$ and $N = \trans{\fix ((M')^\sigma)}$. By Lemma~\ref{lem:existenceofeams}\eqref{lem:existenceofeams2}, we obtain $\trans{(\fix M')^\sigma} =  \appT{\mAppn{n}{1}}{[\mY,\Lookup{\trans[\vec x]{M'}},\Lookup{\vec\sigma}]}\reddh^{7n + 2} \appT{\mY}{\left[\Lookup{\trans{(M')^\sigma}}\right]} = \trans{\fix (M'^\sigma)}$.
	\end{itemize}	
\item Case $M = M_1 \cdot M_2$ and $N = M_1' \cdot M_2$, where $M_1 \redwh M'_1$. By IH $\trans{M_1}\convg \trans{M'_1}$, i.e.\ there exists an EAM $\mach{X}$ such that $\trans{M_1} \reddh^k \mach X$ and $\trans{M_1'}\reddh^{k'} \mach X$, for some $k > k'$. Then we get 
\[
\bar{rcll}
\trans{M_1\cdot M_2}& =& \mProj{1}{1}@[\Lookup{\mProj{1}{1}},\Lookup{\trans{M_1}},\Lookup{\trans{M_2}}],&\textrm{by Definitions~\ref{def:trans} and~\ref{def:fixedmachines}},\\
&\reddh^{4}& \trans{M_1}@[\Lookup{\trans{M_2}}],&\textrm{by Lemma~\ref{lem:existenceofeams}\eqref{lem:existenceofeams2}}\\
&\reddh^{k}& \mach {X}@[\Lookup{\trans{M_2}}]\\
&{}^{k'}_\mach{c}\!\!\twoheadleftarrow&\trans{M_1'}@[\Lookup{\trans{M_2}}]\\
&{}^{4}_\mach{c}\!\!\twoheadleftarrow&\mProj{1}{1}@[\Lookup{\mProj{1}{1}},\Lookup{\trans{M_1'}},\Lookup{\trans{M_2}}],&\textrm{by Lemma~\ref{lem:existenceofeams}\eqref{lem:existenceofeams2}},\\
&=&\trans{M_1\cdot M_2},&\textrm{by Definition~\ref{def:trans}}.
\ear
\]
\item Case $M = \ifterm L{M_1}{M_2}$ and $N = \ifterm {L'}{M_1}{M_2}$, where $L \redwh L'$. By IH $\trans{L}\convg\trans{L'}$, i.e.\ there is an EAM $\mach{X}$ such that $\trans{L} \reddh^k \mach X$ and $\trans{L'}\reddh^{k'} \mach X$ for some $k > k'$. Then we get 
\[
\bar{rlll}
&\mathrlap{\trans{\ifterm L{M_1}{M_2}}}\\& =& \mProj{1}{1}@[\Lookup{\mIfz},\Lookup{\trans{L}},\Lookup{\trans{M_1}},\Lookup{\trans{M_2}}],&\textrm{by Def.~\ref{def:trans} and~\ref{def:fixedmachines}},\\[1ex]
&\reddh^{2}& \mIfz@[\Lookup{\trans{L}},\Lookup{\trans{M_1}},\Lookup{\trans{M_2}}],&\textrm{by Lemma~\ref{lem:existenceofeams}\eqref{lem:existenceofeams2}}\\[1ex]
&\reddh^{3}& \Tuple{R_0 = \Lookup{\trans{L}{}{}},R_1=\Lookup{\trans{M_1}{}{}},R_2=\Lookup{\trans{M_2}{}{}},\\ \Ifz 0120; \Call 0,[]},&\textrm{by Lemma~\ref{lem:existenceofeams}\eqref{lem:existenceofeams5}}\\[1ex]
&\reddh^{k}& \Tuple{R_0 = \Lookup{\mach X},R_1=\Lookup{\trans{M_1}{}{}},R_2=\Lookup{\trans{M_2}{}{}},\\ \Ifz 0120; \Call 0,[]},&\\[1ex]
&{}^{k'}_\mach{c}\!\!\twoheadleftarrow&\Tuple{R_0 = \Lookup{\trans{L'}{}{}},R_1=\Lookup{\trans{M_1}{}{}},R_2=\Lookup{\trans{M_2}{}{}},\\ \Ifz 0120; \Call 0,[]}\\[1ex]
&{}^{3}_\mach{c}\!\!\twoheadleftarrow&\mIfz@[\Lookup{\trans{L}},\Lookup{\trans{M_1}},\Lookup{\trans{M_2}}],&\textrm{by Lemma~\ref{lem:existenceofeams}\eqref{lem:existenceofeams5}},\\[1ex]
&{}^{2}_\mach{c}\!\!\twoheadleftarrow&\mProj{1}{1}@[\Lookup{\mProj{1}{1}},\Lookup{\trans{M_1'}},\Lookup{\trans{M_2}}],&\textrm{by Lemma~\ref{lem:existenceofeams}\eqref{lem:existenceofeams2}},\\[1ex]
&=&\trans{M_1'\cdot M_2},&\textrm{by Definition~\ref{def:trans}}.
\ear
\]
\item  Case $M = \pred M'$ and $N = \pred N'$, where $M' \redwh N'$. Analogous.
\item  Case $M = \succ M'$ and $N = \succ N'$, where $M' \redwh N'$. Analogous.\qedhere
\end{itemize}
\end{proof}

\subsection{Omitted proofs from Section~\ref{sec:Model}}\label{app:model}
% !TEX root = ../../LMCS.tex

\begin{proof}[Proof of Lemma~\ref{lem:eamstr}.]
By induction on a derivation of $x_1:\alpha_1,\dots,x_n : \alpha_n \vdash M : \beta$. 
As a matter of notation, we let $\Gamma = x_1:\alpha_1,\dots,x_n : \alpha_n$ and introduce the abbreviations 
\[
	\bar{lclclcl}
	\vec{a} &=& a_1,\dots,a_n;&&\vec{x} &=& x_1,\dots,x_n;\\ 
 	\vec a^- &=& a_1,\dots,a_{i-1},a_{i+1},\dots,a_n;&&	\vec x^- &=& x_1,\dots,x_{i-1},x_{i+1},\dots,x_n.\\
	\ear
\]
 
\begin{itemize}
\item Case $\Gamma \vdash \num 0 : \tint$. We get
\[
\bar{rcll}
\trans[\vec x]{\num 0}@[\vec a] &=& \appT{\mProj{n+1}{1}}{[0, \vec{a}]},&\textrm{by Definition~\ref{def:trans}},\\
&\reddh& \Lookinv{0},&\textrm{by Lemma~\ref{lem:existenceofeams}\eqref{lem:existenceofeams1}},\\
&{}_\mach{c}\!\!\twoheadleftarrow&\appT{\mProj{n}{1}}{[0, \vec a^-]},&\textrm{by Lemma~\ref{lem:existenceofeams}\eqref{lem:existenceofeams1}},\\
&=&\trans[\vec x^-]{\num 0}@[\vec a^-],&\textrm{by Definition~\ref{def:trans}}.\\
\ear
\]
Conclude by Lemma~\ref{lem:equivprops}\eqref{lem:equivconvh}.
\item Case $\Gamma, y:\beta\vdash y:\beta$. 
Let $b \in \cD_{\beta}$. 
By Definition~\ref{def:trans} and Lemma~\ref{lem:existenceofeams}\eqref{lem:existenceofeams1}, we have 
$\appT{\trans[\vec x,y]{y}}{[\vec a, b]} = \appT{\mProj{n+1}{n+1}}{[\vec a, b]}\reddh \Lookinv{b}$ and $\trans[\vec x^-,y]{y}@[\vec a^-, b] = \appT{\mProj{n}{n}}{[\vec a^-, b]}\reddh \Lookinv{b}$. Conclude by Lemma~\ref{lem:equivprops}\eqref{lem:equivconvh}.
\item Case $\Gamma \vdash M \cdot N : \beta$ since $\Gamma \vdash M : \alpha\to\beta$ and $\Gamma\vdash N : \alpha$.  Then, we have
\[
\bar{rcll}
\appT{\trans[\vec x]{M\cdot N}}{[\vec a]} &=& \appT{\mAppn{n}{2}}{[\Lookup{\mProj{1}{1}},\Lookup{\trans[\vec{x}]{M}},\Lookup {\trans[\vec  {x}]{N}}, \vec{a}]},&\textrm{by Definition~\ref{def:trans}},\\
&\reddh& \appT{\trans[\vec{x}]{M}}{[\vec{a}, \Lookup {(\appT{\trans[\vec{x}]{N}}{[\vec a]})}]},&\textrm{by Lemma~\ref{lem:existenceofeams}\eqref{lem:existenceofeams2}}.\\
\ear
\]

By IH, we have $\appT{\trans[\vec{x}]{N}}{[\vec a]} \equiv_\alpha\! \appT{\trans[\vec x^-]{N}}{[\vec a^-]}$ and \\
$\appT{\trans[\vec{x}]{M}}{[\vec a]} \equiv_{\alpha\to\beta}\! \appT{\trans[\vec x^-]{M}}{[\vec a^-]}$, so by Lemma~\ref{lem:equivprops}\eqref{lem:equivbinrel} (reflexivity), we derive \\$\appT{\trans[\vec{x}]{M}}{[\vec{a}, \Lookup {(\trans[\vec{x}]{N}@[\vec a])}]} \equiv_\beta\appT{\trans[\vec{x}]{M}}{[\vec{a}, \Lookup {(\trans[\vec x^-]{N}@[\vec a^-])}]}$. Conclude by Lemmas~\ref{lem:equivprops}\eqref{lem:equivbinrel}-\eqref{lem:equivconvh} and Definition~\ref{def:trans}.
\item Case $\Gamma \vdash M\langle N/z \rangle :\beta$ with $\Gamma,z:\alpha\vdash M : \beta$ and $\vdash N : \alpha$. By Definition~\ref{def:trans} we get $\trans[\vec{x}]{M\langle N/z\rangle}@[\vec a]=  \appT{\trans[z,\vec{x}]{M}}{[\Lookup{\trans{N}},\vec a]}$. Conclude by applying the IH.
\item Case $\Gamma \vdash \lam z.M:\gamma\to\delta$ since $\Gamma,z:\gamma\vdash M : \delta$. By Definition~\ref{def:trans} we get $\trans[\vec{x}]{\lambda z.M}=\trans[\vec{x},z]{M}$. 
This case follows straightforwardly from the IH.
\item Case $\Gamma \vdash \succ M : \tint$, since $\Gamma \vdash M : \tint$. We get
\[
\bar{rcll}
\appT{\trans[\vec x]{\succ M}}{[\vec a]} &=& \appT{\mAppn{n}{1}}{[\Lookup\mSucc,\Lookup{\trans[\vec{x}]{M}}, \vec{a}]},&\textrm{by Definition~\ref{def:trans}},\\
&\reddh& \appT{\mSucc}{[\Lookup{(\appT{\trans[\vec{x}]{M}}{[\vec a]})}]},&\textrm{by Lemma~\ref{lem:existenceofeams}\eqref{lem:existenceofeams2}}.\\
\ear
\]
From the IH we obtain $\appT{\trans[\vec{x}]{M}}{[\vec a]} \equiv_{\tint} \appT{\trans[\vec x^-]{M}}{[\vec a^-]}$, so by Lemma~\ref{lem:equivprops}\eqref{lem:equivbinrel} (reflexivity), \[\appT{\mSucc}{[\Lookup{(\trans[\vec{x}]{M}@[\vec a])}]} \equiv_\tint\appT{\mSucc}{[\Lookup{(\trans[\vec x^-]{M}@[\vec a^-])}]}\] 
Conclude by Lemmas~\ref{lem:equivprops}-\eqref{lem:equivbinrel}\eqref{lem:equivconvh} and Definition~\ref{def:trans}.
\item Case $\Gamma \vdash \pred M:\tint$. Analogous.
\item Case $\Gamma \vdash \ifterm LMN:\beta$. Analogous.
\item Case $\Gamma \vdash \fix M:\beta$ since $\Gamma\vdash M : \beta\to\beta$. Then,
\[
\bar{rcll}
\appT{\trans[\vec x]{\fix M}}{[\vec a]} &=& \appT{\mAppn{n}{1}}{[\Lookup\mY,\Lookup{\trans[\vec{x}]{M}}, \vec{a}]},&\textrm{by Definition~\ref{def:trans}},\\
&\reddh& \appT{\mY}{[\Lookup{(\appT{\trans[\vec{x}]{M}}{[\vec a]})}]},&\textrm{by Lemma~\ref{lem:existenceofeams}\eqref{lem:existenceofeams2}}.\\
\ear
\]

By IH, we have $\appT{\trans[\vec{x}]{M}}{[\vec a]}\equiv_{\beta\to\beta}\! \appT{\trans[\vec x^-]{M}}{[\vec a^-]}$, so by Lemma~\ref{lem:equivprops}\eqref{lem:equivbinrel}(reflexivity\!) we get 
\[
	\appT{\mY}{[\Lookup{(\appT{\trans[\vec{x}]{M}}{[\vec a]})}]} \equiv_\beta \appT{\mY}{[\Lookup{(\appT{\trans[\vec x^-]{M}}{[\vec a^-]})}]}
\] 
Conclude by Definition~\ref{def:trans}, applying Lemma~\ref{lem:equivprops}\eqref{lem:equivbinrel}\eqref{lem:equivconvh} in case $n>1$ and Lemma~\ref{lem:equivprops}\eqref{lem:equivbinrel} when $n=1$.\qedhere
\end{itemize}
\end{proof}

% !TEX root = ../../LMCS.tex
%!TEX spellcheck = en-US
\begin{proof}[Proof of Proposition~\ref{prop:revtranstyping}.]
Both (1) and (2) follow by mutual induction on a derivation of $\mM:\alpha$ and $\Delta\Vdash^r (P,T) : \alpha$ and call IH1, IH2 the respective induction hypothesis.

As a matter of notation, if $\Delta = {i_1}:\alpha_{i_1},\dots,{i_k}:\alpha_{i_k}$ we let $\Delta^* = x_{i_1}:\alpha_{i_1},\dots,x_{i_k}:\alpha_{i_k}$.
\bsub
\item Case $(\mathrm{nat})$. Then $\mM = \mach{n}$ and $\alpha = \tint$. Conclude since $\rtrans[]{\mach{n}}{\tint} = \num n$ and $\vdash \num n : \tint$ holds.

 Case $(\mathrm{fix})$. Then $\mM = \mach{Y}$ and $\alpha = (\beta\to\beta)\to\beta$. 
We type $\rtrans{\mY}{\alpha} = \lam x.\fix x$ as follows:
\[
	\infer[(\to_\mathrm{I})]{\vdash \lambda x.\fix x : (\beta\rightarrow\beta)\rightarrow\beta}{\infer[(\mathrm{Y})]{x : \beta \rightarrow \beta \vdash \fix x : \beta}{\infer[(\mathrm{ax})]{x : \beta \rightarrow \beta \vdash x: \beta \rightarrow \beta}{}}}
\]

Case $(\vec R)$. Then $\mM = \tuple{R_0,\dots,R_r,P,T}$ with $\vec R\models \Delta$ and $\Delta\Vdash^r (P,T) : \alpha$, for some $\Delta = {i_1}:\alpha_{i_1},\dots,{i_k}:\alpha_{i_k}$.
From the former condition, by rules $(R_\Null)$ and $(R_\Types)$, we get a derivation of $\Lookinv{\val{R_{i_j}}} : \alpha_{i_j}$ having smaller size, for all $1\le j \le k$. By applying IH1, we obtain a derivation of $\vdash\rtrans[]{\Lookinv{!R_{i_j}}}{\alpha_{i_j}} : \alpha_{i_j}$.
From the latter condition and IH2, we have $\Delta^*\vdash\rtrans[\Delta]{P,T}{\alpha} : \alpha$. Therefore, we construct a derivation
\[
\infer=[]{
	\vdash(\lambda x_{i_1}\dots x_{i_k}.\rtrans[\Delta]{P,T}{\alpha})\cdot\rtrans[]{\Lookinv{!R_{i_1}}}{\beta_{i_1}}\cdots\rtrans[]{\Lookinv{!R_{i_k}}}{\beta_{i_k}}:\alpha
	}{\infer=[]{\vdash \lambda x_{i_1}\dots x_{i_k}.\rtrans[\Delta]{P,T}{\alpha} : \alpha_{i_1}\rightarrow\cdots\rightarrow\alpha_{i_k}\rightarrow\alpha}{\Delta^*\vdash\rtrans[\Delta]{P,T}{\alpha} : \alpha} 
	&
	\vdash \rtrans[]{\Lookinv{!R_{i_j}}}{\alpha_{i_j}} : \alpha_{i_j}
	& 1\!\le\! j \!\le\! k}
\]

\item Case $(\mathrm{load_{\Null}})$. Then $P = \Load j;P'$, $T = []$, $\alpha = \beta_1\to\beta_2$ and $\Delta[j : \beta_1] \Vdash^r (P',[]) : \beta_2$ has a derivation of smaller size. There are two subcases.
	\begin{itemize}
	\item Case $j\notin\dom(\Delta)$, whence $\Delta[j : \beta_1] = \Delta,j : \beta_1$.
	By IH2 we get $\Delta^*\vdash \rtrans[{\Delta,j : \beta_1}]{P',[]}{\beta_2} : \beta_2$. Simply apply rule $(\to_{\mathrm{I}})$.
	\item Case $j\in\dom(\Delta)$, say, $j = i_k$. From the IH2 we obtain $\Gamma,x_{i_k} : \beta_1\vdash \rtrans[{\Delta[i_k : \beta_1]}]{P',[]}{\beta_2} : \beta_2$ for $\Gamma = x_{i_1} : \alpha_{i_1},\dots,x_{i_{k-1}}:\alpha_{i_{k-1}}$.
	By applying the rule $(\to_{\mathrm{I}})$, we obtain a derivation of $\Gamma\vdash \lam x_{i_k}.\rtrans[{\Delta[i_k : \beta_1]}]{P',[]}{\beta_2} : \beta_1\to\beta_2$ whence, by strengthening (Lemma~\ref{lem:epcftyping}\eqref{lem:epcftyping2}), we conclude 
	$\Gamma,x_{i_k} : \alpha_{i_k}\vdash \lam x_{i_k}.\rtrans[{\Delta[i_k : \beta_1]}]{P',[]}{\beta_2} : \beta_1\to\beta_2$.
	\end{itemize}
\noindent %FIXED indent
Case $(\mathrm{load_{\Types}})$. Then $P = \Load j;P'$, $T = \Cons a{T'}$. Moreover, $\Delta[j : \beta] \Vdash^r (P',T') : \alpha$ and $\Lookinv{a} : \beta$ have a derivation of smaller size, for some $\beta$.
From the former, one obtains a derivation of $\Delta^*\vdash \lambda x_j.\rtrans[{\Delta[j:\beta]}]{ P',T'}{\alpha}:\beta\rightarrow\alpha$ proceeding as above.
By the IH1 applied to the latter, we obtain a derivation of $\vdash\rtrans[]{\Lookinv{a}}{\beta}:\beta$, whence $\Delta^*\vdash\rtrans[]{\Lookinv{a}}{\beta}:\beta$ holds by strengthening (Lemma~\ref{lem:epcftyping}\eqref{lem:epcftyping2}). By applying rule $(\to_{\mathrm{E}})$, we conclude that $\Delta^*\vdash\big(\lambda x_j.\rtrans[{\Delta[j:\beta]}]{P',T'}{\alpha}\big)\cdot\rtrans[]{\Lookinv{a}}{\beta}:\alpha$ is derivable.

Case $(\mathrm{pred})$. Then $P = \Pred i j;P'$ and $(\Delta, i : \tint)[j : \tint] \Vdash^r (P',T') : \alpha$ has a smaller derivation. We assume $j\notin\dom(\Delta,i:\tint)$, otherwise proceed as in case $(\mathrm{load_{\Null}})$.
By IH2, we obtain a derivation of $\Delta^*,x_i : \tint,x_j : \tint\vdash \rtrans[{\Delta,i:\tint,j:\tint}]{(P',T')}{\alpha} : \alpha$, thus:
\begin{minipage}{\linewidth} %FIXED Minipage for margins
\[
\infer{\Delta^*, x_i\! :\! \tint \vdash \ifc{x_i}{(\lam x_j.\rtrans[{\Delta,i:\tint,j:\tint}]{P',T'}{\alpha})\cdot (\pred x_i)} : \alpha}{
	\Delta^*, x_i : \tint \vdash x_i : \tint
	&\!\!\!\!\!\!\!\!\!
	\infer{\Delta^*, x_i \vdash(\lambda x_j.\rtrans[{\Delta,i:\tint,j:\tint}]{P',T'}{\alpha})\cdot (\pred x_i): \alpha}{
		\infer{\Delta^*, x_i : \tint \vdash \lambda x_j.\rtrans[{\Delta,i:\tint,j:\tint}]{P',T'}{\alpha} : \alpha}{
			\Delta^*, x_i:\tint, x_j:\tint\vdash\rtrans[{\Delta,i:\tint,j:\tint}]{P',T'}{\alpha}:\alpha
			}
			&
			\infer{\Delta^*, x_i : \tint \vdash \pred x_i : \tint}{
			\Delta^*, x_i : \tint \vdash x_i : \tint
			}
	}
	}
\]
\end{minipage}
Case $(\mathrm{succ})$. Analogous.

Case $(\mathrm{call})$. Then $P = \Call i$ and $T = [a_1,\dots,a_n]$ with $\Lookinv{a_j} : \beta_j$, for $j\,(1\le j \le n)$. 
Call $\Gamma = \Delta^*,x_i : \beta_1\to\cdots\to\beta_n\to\alpha$.
By IH1, we get $\vdash \rtrans[]{\Lookinv{a_j}}{\beta_j} : \beta_j$ whence $\Gamma\vdash\rtrans[]{\Lookinv{a}}{\beta}:\beta$ holds by strengthening (Lemma~\ref{lem:epcftyping}\eqref{lem:epcftyping2}). Derive
\[
\infer={\Gamma\vdash x_i\cdot \rtrans []{\Lookinv{a_1}}{\beta_1} \cdots \rtrans []{\Lookinv{a_n}}{\beta_n} : \alpha}{
	\Gamma\vdash x_i : \beta_1\!\rightarrow\cdots\rightarrow\beta_n\!\rightarrow\alpha
	&
	\Gamma\vdash \rtrans []{\Lookinv{a_1}}{\beta_1}\!:\beta_1
	\cdots
	\Gamma\vdash \rtrans []{\Lookinv{a_n}}{\beta_n}\!:\beta_n
	}
\]
\item Case $(\mathrm{app})$. Then $P = \Apply i j l;P'$ and $(\Delta,i:\beta_1{\to}\beta_2,j:\beta_1)[l : \beta_2] {\Vdash^r}(P',T) : \alpha$ has a smaller derivation, for some $\beta_1,\beta_2$. Assume that $l\notin\dom(\Delta,i:\beta_1{\to}\beta_2,j:\beta_1)$, otherwise proceed as in case $(\mathrm{load_{\Null}})$.
Setting $\Gamma' = \Delta^*, x_i:\beta_1\to\beta_2,x_j:\beta_1$, we get:
\[
	\infer{\Gamma\vdash (\lambda x_l.\rtrans[{\Delta, i:\beta_1\rightarrow\beta_2,j:\beta_1,l:\beta_2}]{P',T}{\alpha})\cdot(x_i \cdot x_j) : \alpha
	}{
		\infer{\Gamma\vdash \lambda x_l.\rtrans[{\Delta, i:\beta_1\rightarrow\beta_2,j:\beta_1,l:\beta_2}]{P',T}{\alpha} : \beta_2 \rightarrow \alpha}{\Gamma, x_l : \beta_2\vdash \rtrans[{\Delta, i:\beta_1\rightarrow\beta_2,j:\beta_1,l:\beta_2}]{P',T}{\alpha} : \alpha}
		&
			\infer{
	\Gamma\vdash x_i \cdot x_j :\beta_2
	}{
		\Gamma \vdash x_i : \beta_1 \rightarrow\beta_2
		&
		\Gamma\vdash x_j : \beta_1
	}
	}
\]

Case $(\mathrm{test})$. Then $P =\Ifz i j l m;P' $ and $(\Delta,i : \tint,j:\beta,l:\beta)[m:\beta] \Vdash^r (P,T) : \alpha$ has a smaller derivation, for some $\beta$. Assume $m\notin\dom(\Delta,i : \tint,j:\beta,l:\beta)$, otherwise proceed as in case $(\mathrm{load_{\Null}})$. Let $\Gamma = \Delta^*,x_i : \tint,x_j:\beta,x_l:\beta,x_m:\beta$, we get:

\[
\infer{\Gamma\vdash \ifc{x_i}{(\lambda x_m.\rtrans[{\Delta,i:\tint,j:\beta,l:\beta,m:\beta}]{ P',T}{\alpha})\cdot \ifterm {x_i}{x_j}{x_l}}:\alpha
	}{
		\Gamma\vdash x_i:\tint
		&
		\infer{\Gamma\vdash (\lambda x_m.\rtrans[{\Delta,i:\tint,j:\beta,l:\beta,m:\beta}]{ P',T}{\alpha})\cdot \ifterm {x_i}{x_j}{x_l}:\alpha
		}{
			\infer{\Gamma\vdash\lambda x_m.\rtrans[{\Delta,i:\tint,j:\beta,l:\beta,m:\beta}]{ P',T}{\alpha}:\beta\rightarrow\alpha
			}{
				\Gamma, x_m : \beta \vdash \rtrans[{\Delta,i:\tint,j:\beta,l:\beta,m:\beta}]{ P',T}{\alpha}:\alpha
				} 
				& 
				\infer{\Gamma\vdash \ifterm {x_i}{x_j}{x_l} : \beta
				}{
					\Gamma\vdash x_n : \tint, n \in \{i,j,k\}
				}
		}
	}
\]
\esub
This concludes the proof.
\end{proof}

% !TEX root = ../../LMCS.tex
%!TEX spellcheck = en-US

\begin{proof}[Proof of Theorem~\ref{thm:transrevtrans}.]

Consider the following statements: 
\bsub
\item
	Let $\mM\in\cM$. If $\mM : \alpha$ then $\trans{\rtrans[]{\mM}{\alpha}}\equiv_\alpha \mM$.
\item 
	For all $i_1:\alpha_{i_1},\dots,i_n:\alpha_{i_n} \Vdash^r (P,T):\alpha$, $a_{i_1}\in \cD_{\alpha_{i_1}},\dots,a_{i_n}\in \cD_{\alpha_{i_n}},$ 
\[ 
	\appT{\trans[x_{i_1},\dots,x_{i_n}]{\rtrans[i_1:\alpha_{i_1},\dots,i_n:\alpha_{i_n}]{P,T}{\alpha}}}{[a_{i_1},\dots,a_{i_n}]} \equiv_\alpha \tuple{\vec{R}^r_{a_{i_1},\dots,a_{i_n}},P,T},
\]
where $\vec{R}^r_{a_{i_1},\dots,a_{i_n}}$ denotes the list of registers $R_0,\dots,R_{r}$ such that, for all $j\,(0 \leq j \leq r)$, $\val R_j = a_j$ if $j \in \set{i_1,\dots,i_n}$, and $\val R_j = \Null$ otherwise.
\esub
Both statements are proven using simultaneous induction on a derivation of $\mM : \alpha$ and $\Delta \Vdash^r (P,T):\alpha$. We refer to the former induction hypothesis as IH1 and to the latter as IH2. As a matter of notation, we let $\Delta = i_1:\beta_{i_1},\dots,i_n:\beta_{i_n}$, $\vec x = x_{i_1},\dots,x_{i_n}$, and $\vec a = a_{i_1},\dots,a_{i_n}$ such that for all $j \in \set{i_1,\dots,i_n},\,a_j \in \cD_{\beta_j}$.

\begin{itemize}
\item Case $\mach k : \tint$. We prove this case by induction on $k\in\nat$ (and call this $\mathrm{IH1}'$).
\begin{itemize}
\item Case $k = 0$: By Definition~\ref{def:trans} we get $\trans{\rtrans[]{\mach 0}{\tint}} = \trans{\num 0} = \mProj{1}{1}@[0]$.  Conclude by Lemmas~\ref{lem:existenceofeams}\eqref{lem:existenceofeams1} and~\ref{lem:equivprops}\eqref{lem:equivconvh}.
\item Case $k = m + 1$, for some $m\in\nat$. Then by Definitions~\ref{def:trans} and~\ref{def:revtrans}, we have 
\[\bar{rcl}
	\trans{\rtrans[]{\mach k}{\tint}} &=& \trans{\num{m+1}} = \trans{\succ \num m} \\&=& \appT{\mAppn{0}{1}}{[\Lookup{\mSucc}, \Lookup{\trans{\num m}}]} = \mProj{1}{1}@[\Lookup{\mSucc}, \Lookup{\trans{\num m}}].
	\ear
\] 
By $\mathrm{IH1}'$ we have $\trans{\num m} \equiv_\tint \mach m$, so $\mProj{1}{1}@[\Lookup{\mSucc}, \Lookup{\trans{\num m}}] \equiv_\tint \mProj{1}{1}@[\Lookup{\mSucc}, m]$, by reflexivity. Conclude by Lemma~\ref{lem:existenceofeams}\eqref{lem:existenceofeams1},\eqref{lem:existenceofeams4}, Lemma~\ref{lem:equivprops}\eqref{lem:equivconvh}, and transitivity.
\end{itemize}
\noindent %FIXED indent
\item Case $\mY : (\alpha\to\alpha)\to\alpha$. Let $a,b \in \cD_{\alpha\to\alpha}$ such that $a \simeq_{\alpha\to\alpha} b$. 
Since by Lemma~\ref{lem:welltypedeams}\eqref{lem:welltypedeams1} and Lemma~\ref{lem:equivprops}\eqref{lem:equivconvh} we have $\appT{\mProj{1}{1}}{[a]} \equiv_{\alpha\to\alpha} \Lookinv{a}$, we obtain
\[
\bar{rcll}
\appT{\trans{\rtrans[]{\mY{}}{(\alpha\to\alpha)\to\alpha}}}{[a]} &=&\appT{\trans{\lam x.\fix x}}{[a]},&\textrm{by Definition~\ref{def:revtrans}},\\
&=&\appT{\mAppn{1}{1}}{ [\Lookup{\mY},\Lookup{\mProj{1}{1}},a]},&\textrm{by Definition~\ref{def:trans}},\\
&\reddh&\appT{\Lookinv{a}}{[\Lookup{\mY}\cdot (\Lookup{\mProj{1}{1}}\cdot a)]},&\textrm{by Lemma~\ref{lem:existenceofeams}\eqref{lem:existenceofeams6}},\\
&\equiv_\alpha& \appT{\Lookinv{a}}{[\Lookup{\mY}\cdot a]},&\textrm{by Lemma~\ref{lem:equivprops}\eqref{lem:equivtape}},\\
&\equiv_\alpha& \appT{\Lookinv{a}}{[\Lookup{\mY}\cdot b]},&\textrm{by reflexivity},\\
&\equiv_\alpha& \appT{\Lookinv{b}}{[\Lookup{\mY}\cdot b]},&\textrm{by Definition~\ref{def:model}},\\
&{}_\mach{c}\!\!\twoheadleftarrow&\appT{\mY}{[b]},&\textrm{by Lemma~\ref{lem:existenceofeams}\eqref{lem:existenceofeams6}.}
\ear
\]
Conclude by Lemma~\ref{lem:equivprops}\eqref{lem:equivconvh} and transitivity.
\item Case $\tuple{\vec R, P, T}:\alpha$. Let $\vec R \models \Delta$. By Definition~\ref{def:revtrans}, Lemma~\ref{lem:existenceofeams}, and Definition~\ref{def:trans} we get 
\[
\bar[t]{rcl}
\trans{\rtrans[]{\langle \vec R, P, T \rangle}{\alpha}}&=&\trans{(\lambda x_{i_1}\dots x_{i_n}.\rtrans[\Delta]{P,T}{\alpha})\cdot\rtrans[]{\Lookinv{!R_{i_1}}}{\beta_{i_1}}\cdots\rtrans[]{\Lookinv{!R_{i_n}}}{\beta_{i_n}}}  \\
&=& \appT{\trans[\vec x]{\rtrans[\Delta]{P,T}{\alpha}}}{[\Lookup{\trans{\rtrans[]{\Lookinv{!R_{i_1}}}{\beta_{i_1}}}},\dots,\Lookup{\trans{\rtrans[]{\Lookinv{!R_{i_n}}}{\beta_{i_n}}}}]}\\
\ear
\]
By IH1, for all $k\in \dom(\Delta)$, we have $\Lookup{\trans{\rtrans[]{\Lookinv{!R_{k}}}{\beta_{k}}}}\simeq_{\alpha}!R_k$. Then by reflexivity, \[\bar{l}\appT{\trans[\vec x]{\rtrans[\Delta]{P,T}{\alpha}} }{ [\Lookup{\trans{\rtrans[]{\Lookinv{!R_{i_1}}}{\beta_{i_1}}}},\dots,\Lookup{\trans{\rtrans[]{\Lookinv{!R_{i_n}}}{\beta_{i_n}}}}]}\\\equiv_\alpha \appT{\trans[\vec x]{\rtrans[\Delta]{P,T}{\alpha}}}{[!R_{i_1},\dots,!R_{i_n}]}\ear\] Conclude by IH2, Lemma~\ref{lem:equivprops}\eqref{lem:equivconvh}, and transitivity.
\item Case $\Delta\Vdash^r (\Load k;P,[]):\beta\to\alpha$. There are two subcases.
\begin{itemize}
\item Subcase $k \notin \dom(\Delta)$. Let $b,c \in \cD_{\beta}$ such that $b \simeq_\beta c$. We get 
\[
\bar{rlll}
&\mathrlap{\appT{\trans[\vec x]{\rtrans[\Delta]{\Load k;P,[]}{\beta\to\alpha}}}{[\vec a, b]}} \\&=& 
\appT{\trans[\vec x]{\lam x_k.\rtrans[\Delta,k:\beta]{P,[]}{\alpha}}}{[\vec a, b]},&\textrm{by Definition~\ref{def:revtrans}},\\
 &=& 
 \appT{\trans[\vec x,x_k]{\rtrans[{\Delta,k:\beta}]{P,[]}{\alpha}}}{[\vec a, b]},&\textrm{by Definition~\ref{def:trans}},\\
&\equiv_\alpha& 
\tuple{\vec{R}^r_{\vec a,b},P,[]},&\textrm{by IH2},\\
&{}_\mach{c}\!\!\leftarrow&\tuple{\vec{R}^r_{\vec a},\Load k;P,[b]},&\textrm{by Definition~\ref{def:eamred}},\\
&\equiv_\alpha&\tuple{\vec{R}^r_{\vec a},\Load k;P,[c]},&\textrm{by reflexivity}.\\
\ear
\]
Conclude by Lemma~\ref{lem:equivprops}\eqref{lem:equivconvh} and transitivity.
\item Subcase $k \in\dom(\Delta)$. 
	Let $k = i_m$, and let $b,c\in \cD_{\beta}$ such that $b \simeq_\beta c$. 
	We also fix $\vec{x'} = x_{i_1},\dots,x_{i_{m-1}},x_{i_{m+1}},\dots,x_{i_n}$ and $\vec{a'} = a_{i_1},\dots,a_{i_{m-1}},a_{i_{m+1}},\dots,a_{i_n}$. We get 
\[
\bar{rlll}
&\mathrlap{\appT{\trans[\vec x]{\rtrans[\Delta]{\Load k;P,[]}{\beta\to\alpha}}}{[\vec a, b]}} \\&=& \appT{\trans[\vec x]{\lam x_k.\rtrans[\Delta{[k:\beta]}]{ P,[]}{\alpha}}}{[\vec a, b]},&\textrm{by Definition~\ref{def:revtrans}},\\
&\equiv_\alpha&\appT{\trans[\vec{x'}]{\lam x_{k}.\rtrans[\Delta{[k:\beta]}]{ P,[]}{\alpha}}}{[\vec{a'}, b]},&\textrm{by Lemma~\ref{lem:eamstr}},\\
 &=& \appT{\trans[\vec {x'},x_k]{\rtrans[{\Delta,k:\beta}]{ P,[]}{\alpha}}}{[\vec {a'}, b]},&\textrm{by Definition~\ref{def:trans}},\\
&\equiv_\alpha& \tuple{\vec{R}^r_{\vec {a'},b},P,[]},&\textrm{by IH2},\\
&{}_\mach{c}\!\!\twoheadleftarrow&\tuple{\vec{R}^r_{\vec {a}},\Load k;P,[b]},&\textrm{by Definition~\ref{def:eamred}},\\
&\equiv_\alpha&\tuple{\vec{R}^r_{\vec a},\Load k;P,[c]},&\textrm{by reflexivity}.\\
\ear
\]
Conclude by Lemma~\ref{lem:equivprops}\eqref{lem:equivconvh} and transitivity.
\end{itemize} 
\end{itemize}
\noindent %FIXED indent
In the cases following, we assume that $k \notin\dom(\Delta)$. If $k \in\dom(\Delta)$, one proceeds as above.
\begin{itemize}
\item Case $\Delta \Vdash^r (\Load k;P,b::T): \alpha$. By Definition~\ref{def:revtrans}, Lemma~\ref{lem:existenceofeams}, and Definition~\ref{def:trans}, we get 
\[
\bar{rll}
&\mathrlap{\appT{\trans[\vec x]{\rtrans[\Delta]{\Load k;P,b::T}{\alpha}}}{[\vec a]}} \\&=& \appT{\trans[\vec x]{(\lambda x_k.\rtrans[\Delta,k:\beta]{ P,T}{\alpha})\cdot\rtrans[]{\Lookinv{b}}{\beta}}}{[\vec a]}\\
&=&\appT{\mAppn{n}{2}}{[\Lookup{\mProj{1}{1}},\Lookup{\trans[\vec x,x_k]{\rtrans[\Delta,k:\beta]{ P,T}{\alpha}}}, \Lookup{\trans[\vec x]{\rtrans[]{\Lookinv{b}}{\beta}}}, \vec a]}\\
&\reddh&\appT{\trans[\vec x, x_k]{\rtrans[\Delta,k:\beta]{ P,T}{\alpha}}}{[\vec a, \Lookup{(\appT{\trans[\vec x]{\rtrans[]{\Lookinv{b}}{\beta}}}{[\vec a]})}]}.
\ear
\] 
By Proposition~\ref{prop:revtranstyping}, $\vdash\rtrans[]{\Lookinv{b}}{\beta} : \beta$, so by Corollary~\ref{cor:closedargs} and IH1, $\appT{\Lookup{\trans{\rtrans[]{\Lookinv{b}}{\beta}}}}{\vec a}\simeq_{\alpha} b$. By reflexivity, we obtain
\[
	\appT{\trans[\vec x, x_k]{\rtrans[\Delta,k:\beta]{ P,T}{\alpha}}}{[\vec a, \Lookup{(\appT{\trans[\vec x]{\rtrans[]{\Lookinv{b}}{\beta}}}{[\vec a]})}]}\equiv_\alpha \appT{\trans[\vec x, x_k]{\rtrans[\Delta,k:\beta]{ P,T}{\alpha}}}{[\vec a, b]}. 
\]
Conclude by IH2, Lemma~\ref{lem:equivprops}\eqref{lem:equivconvh}, and transitivity.
\item Case $x_{i_1}:\beta_{i_1},\dots,x_{i_m}:\tint,\dots,x_n:\beta_{i_n} \Vdash^r (\Pred {i_m}k;P,T) : \alpha$. Fix the notation $\mM = \trans[\vec x]{{(\lambda x_k.\rtrans[{\Delta,k:\tint}]{ P,T}{\alpha})\cdot (\pred x_{i_m})}}$. By Definition~\ref{def:revtrans} and Lemma~\ref{lem:existenceofeams} we get
\[
\bar{rll}
&\mathrlap{\appT{\trans[\vec x]{\rtrans[\Delta]{\Pred {i_m}k;P,T}{\alpha}}}{[\vec a]}} \\&=& \appT{\trans[\vec x]{\ifc{x_{i_m}}{(\lambda x_k.\rtrans[{\Delta,k:\tint}]{ P,T}{\alpha})\cdot (\pred x_{i_m})}}}{[\vec a]}\\
&=&\appT{\mAppn{n}{3}}{[\Lookup{\mIfz},\Lookup{\mProj{n}{i_m}}, \Lookup{\mM}, \Lookup{\mM},\vec a]}\\
&\reddh&\appT{\mIfz}{[\Lookup{(\appT{\mProj{n}{i_m}}{[\vec a]})}, \Lookup{(\appT{\mM}{[\vec a]})}, \Lookup{(\appT{\mM}{[\vec a]})},\vec a]}\\
&\equiv_\alpha& \appT{\mIfz}{[a_{i_m}, \Lookup{(\appT{\mM}{[\vec a]})}, \Lookup{(\appT{\mM}{[\vec a]})},\vec a]}.\\
\ear
\]
There are two subcases from this point.
\begin{itemize}
\item Subcase $\Lookinv{a_{i_m}}$ does not terminate. 
Then, by Definition~\ref{def:eamred}, we have that the machine $\appT{\mIfz}{[a_{i_m}, \Lookup{(\appT{\mM}{[\vec a]})}, \Lookup{(\appT{\mM}{[\vec a]})},\vec a]}$ and $\tuple{\vec{R}^r_{\vec a},\Pred kl; P,T}$ cannot terminate either. 
By Definition~\ref{def:EAMtypes}, we have $\tuple{\vec{R}^r_{\vec a},\Pred kl; P,T} : \alpha$. Conclude by Lemma~\ref{lem:equivprops}\eqref{lem:equivconvh} and transitivity.
\item Subcase $\Lookinv{a_{i_m}}\reddh \mach{t} = \Lookinv{t}$, for some $t \in \nat$. 
Now, easy calculations give $\appT{\trans[\vec x]{\rtrans{\pred x_{i_m}}{\tint}}}{[\vec a]} \reddh \mach{t}':= \Lookinv{t\ominus 1},$ from which $\appT{\trans[\vec x]{\rtrans{\pred x_{i_m}}{\tint}}}{[\vec a]} \equiv_\tint \mach{t}'$ follows by Lemma~\ref{lem:equivprops}\eqref{lem:equivconvh}. Then we get
\[
\bar{rcll}
&&\appT{\mIfz}{[a_{i_m}, \Lookup{(\appT{\mM}{[\vec a]})}, \Lookup{(\appT{\mM}{[\vec a]})},\vec a]}\\
&\reddh& \appT{\trans[\vec x]{(\lambda x_k.\rtrans[{\Delta, k:\tint}]{ P,T}{\alpha})\cdot (\pred x_{i_m})}}{[\vec a]},&\textrm{by Lem.~\ref{lem:existenceofeams}\eqref{lem:existenceofeams6}},\\
&=&\appT{\mAppn{n}{2}}{\left[\bar{l}\Lookup{\mProj{1}{1}}, \Lookup{\trans[\vec x, x_k]{\rtrans[{\Delta, k:\tint}]{ P,T}{\alpha}}},\\\Lookup{\trans[\vec x]{\rtrans{\pred x_{i_m}}{\tint}}},\vec a\ear\right]},&\textrm{by Def.~\ref{def:trans}},\\
&\reddh&\appT{\trans[\vec x, x_k]{\rtrans[{\Delta, k:\tint}]{P,T}{\alpha}}}{[\vec a, \Lookup{(\appT{\trans[\vec x]{\rtrans{\pred x_{i_m}}{\tint}}}{[\vec a]})}]},&\textrm{by Lem.~\ref{lem:existenceofeams}\eqref{lem:existenceofeams2}},\\
&\equiv_\alpha&\appT{\trans[\vec x, x_k]{\rtrans[{\Delta, k:\tint}]{ P,T}{\alpha}}}{[\vec a,t\ominus 1, 0)]},&\textrm{by reflexivity},\\
&\equiv_\alpha&\tuple{\vec{R}^r_{\vec a}[R_k := t\ominus 1],P,T},&\textrm{by IH2},\\
&\equiv_\alpha&\tuple{\vec{R}^r_{\vec a},\Pred {i_m}k;P,T},&\textrm{by Lem.~\ref{lem:equivprops}\eqref{lem:equivconvh}}.\\
\ear
\]
Conclude by Lemma~\ref{lem:equivprops}\eqref{lem:equivconvh} and transitivity.
\end{itemize}
\item Case $x_{i_1}:\beta_{i_1},\dots,x_{i_m}:\tint,\dots,x_n:\beta_{i_n} \Vdash^r (\Succ {i_m}k;P,T) : \alpha$. Analogous.
\item Case $\Delta \Vdash^r (\Ifz {i_l}{i_{m_1}}{i_{m_2}}k;P,T) : \alpha$, where $\Delta(l) = \tint,\,\Delta(i_{m_1}) = \beta,\,\Delta(i_{m_2}) = \beta$. There are two subcases.
\begin{itemize}
\item Subcase $\Lookinv{a_{i_l}}$ does not terminate. Proceed as in case $x_{i_1}\!\!:\beta_{i_1},\dots,x_{i_m}\!:\tint,\dots,x_n:\beta_{i_n} \Vdash^r (\Pred {i_m}k;P,T) : \alpha$.
\item Subcase $\Lookinv{a_{i_l}}\reddh \Lookinv{t}, t \in \nat$. Proceed as in case $x_{i_1}:\beta_{i_1},\dots,x_{i_m}:\tint,\dots,x_n:\beta_{i_n} \Vdash^r (\Pred {i_m}k;P,T) : \alpha$ to get \\$\appT{\trans[\vec x]{\rtrans[\Delta]{\Ifz {i_l}{i_{m_1}}{i_{m_2}}k;P,T}{\alpha}}}{[\vec a]}\equiv_\alpha $\\$\appT{\trans{\rtrans[{\Delta, k:\beta}]{ P,T}{\alpha}}{}{\vec x,x_k}}{[\vec a, \Lookup{(\appT{\trans[\vec x]{\ifterm {x_{i_l}}{x_{i_{m_1}}}{x_{i_{m_2}}}}}{[\vec a]})}]}
$. There are two subcases.
\begin{itemize}
\item Subcase $t = 0$. By Lemma~\ref{lem:equivprops}\eqref{lem:equivconvh} we have $\appT{\trans[\vec x]{\ifterm {x_{i_l}}{x_{i_{m_1}}}{x_{i_{m_2}}}}}{[\vec a]} \equiv_\beta \Lookinv{a_{i_{m_{1}}}}$. Then we get
\[
\bar{cll}
&\mathrlap{\appT{\trans[\vec x,x_k]{\rtrans[{\Delta, k:\beta}]{ P,T}{\alpha}}}{[\vec a, \Lookup{(\appT{\trans[\vec x]{\ifterm {x_{i_l}}{x_{i_{m_1}}}{x_{i_{m_2}}}}}{[\vec a]})}]}}&\\
\equiv_\alpha&\appT{\trans[\vec x,x_k]{\rtrans[{\Delta, k:\beta}]{ P,T}{\alpha}}}{[\vec a, a_{i_{m_1}}]},&\textrm{by reflexivity,}\\
\equiv_\alpha&\tuple{\vec{R}^r_{\vec a}[R_k := a_{i_{m_1}}],P,T},&\textrm{by IH2},\\
\equiv_\alpha&\tuple{\vec{R}^r_{\vec a},\Ifz {i_l}{i_{m_1}}{i_{m_2}}k;P,T},&\textrm{by Lemma~\ref{lem:equivprops}\eqref{lem:equivconvh}.}
\ear
\]
Conclude by Lemma~\ref{lem:equivprops}\eqref{lem:equivconvh} and transitivity.
\item Subcase $t > 0$. By Lemma~\ref{lem:equivprops}\eqref{lem:equivconvh} we have $\appT{\trans[\vec x]{\ifterm {x_{i_l}}{x_{i_{m_1}}}{x_{i_{m_2}}}}}{[\vec a]} \equiv_\beta \Lookinv{a_{i_{m_{2}}}}$. Then we get
\[
\bar{cll}
&\mathrlap{\appT{\trans[\vec x,x_k]{\rtrans[{\Delta, k:\beta}]{ P,T}{\alpha}}}{[\vec a, \Lookup{(\appT{\trans[\vec x]{\ifterm {x_{i_l}}{x_{i_{m_1}}}{x_{i_{m_2}}}}}{[\vec a]})}]}}&\\
\equiv_\alpha&\appT{\trans[\vec x,x_k]{\rtrans[{\Delta, k:\beta}]{ P,T}{\alpha}}}{[\vec a, a_{i_{m_2}}]},&\textrm{by reflexivity,}\\
\equiv_\alpha&\tuple{\vec{R}^r_{\vec a}[R_k := a_{i_{m_2}}],P,T},&\textrm{by IH2},\\
\equiv_\alpha&\tuple{\vec{R}^r_{\vec a},\Ifz {i_l}{i_{m_1}}{i_{m_2}}k;P,T},&\textrm{by Lemma~\ref{lem:equivprops}\eqref{lem:equivconvh}.}
\ear
\]
Conclude by Lemma~\ref{lem:equivprops}\eqref{lem:equivconvh} and transitivity.
\end{itemize}
\end{itemize} 
\item Case $\Delta\Vdash^r (\Apply {i_l}{i_m}k;P,T):\gamma$, where $\Delta(i_l) = \beta\to\alpha$ and $\Delta(i_m) = \beta$. By Definition~\ref{def:revtrans}, Definition~\ref{def:trans}, and Lemma~\ref{lem:existenceofeams}\eqref{lem:existenceofeams2}, we get
\[
\bar{rcl}
&&\appT{\trans[\vec x]{\rtrans{\Apply {i_l}{i_m}m ;P,T}{\gamma}}}{[\vec a]} \\&=&\appT{\trans[\vec x]{(\lambda x_k.\rtrans[\Delta, k:\beta]{P,T}{\gamma})\cdot(x_{i_l} \cdot x_{i_m})}}{[\vec a]}\\
&=&\appT{\mAppn{n}{2}}{[\Lookup{\mProj{1}{1}},\Lookup{\trans[\vec x, x_k]{\rtrans[\Delta, k:\beta]{ P,T}{\gamma}}},\Lookup{\trans[\vec x]{x_{i_l} \cdot x_{i_m}}},\vec a]}\\
&\reddh&\appT{\trans[\vec x, x_k]{\rtrans[\Delta, k:\beta]{ P,T}{\gamma}}}{[\vec a, \Lookup{(\appT{\trans[\vec x]{x_{i_l} \cdot x_{i_m}}}{[\vec a]})}]}\\
\ear
\]
By Lemma~\ref{lem:equivprops}\eqref{lem:equivconvh} we have $\appT{\Lookup{\mProj{n}{i_m}}}{[\vec a]}\equiv_\beta \Lookinv{a_{i_m}}$. We then have
\[
\bar{rcll}
\appT{\trans[\vec x]{x_{i_l} \cdot x_{i_m}}}{[\vec a]} &=& \appT{\mAppn{n}{2}}{[\Lookup{\mProj{1}{1}}, \Lookup{\mProj{n}{i_l}},\Lookup{\mProj{n}{i_m}},\vec a]},&\textrm{by Definition~\ref{def:trans}},\\
&\reddh&\appT{\mProj{n}{i_l}}{[\vec a, \Lookup{(\appT{\Lookup{\mProj{n}{i_m}}}{[\vec a]})}]},&\textrm{by Lemma~\ref{lem:existenceofeams}\eqref{lem:existenceofeams2}},\\
&\equiv_\alpha& \appT{\mProj{n}{i_l}}{[\vec a, a_{i_m}]},&\textrm{by reflexivity},\\
&\equiv_\alpha& \appT{\Lookinv{a_{i_l}}}{[a_{i_m}]},&\textrm{by Lemma~\ref{lem:equivprops}\eqref{lem:equivconvh}}.\\
\ear
\]
Thus we get 
\[
\bar{cll}
&\appT{\trans[\vec x, x_k]{\rtrans[\Delta, k:\beta]{ P,T}{\gamma}}}{[\vec a, \Lookup{(\appT{\trans[\vec x]{x_{i_l} \cdot x_{i_m}}}{[\vec a]})}]}\\
\equiv_\gamma &\appT{\trans[\vec x, x_k]{\rtrans[\Delta, k:\beta]{ P,T}{\gamma}}}{[\vec a, a_{i_l}\cdot a_{i_m}]},&\textrm{by reflexivity,}\\
\equiv_\gamma & \tuple{\vec{R}^r_{\vec a}[R_k := a_{i_l} \cdot a_{i_m}],P,T},&\textrm{by IH2},\\
\equiv_\gamma & \tuple{\vec{R}^r_{\vec a},\Apply {i_l}{i_m}k;P,T},&\textrm{by Lemma~\ref{lem:equivprops}\eqref{lem:equivconvh}}.\\
\ear
\]
Conclude by Lemma~\ref{lem:equivprops}\eqref{lem:equivconvh} and transitivity.
\item Case $\Delta \Vdash^r (\Call k,[b_1,\dots,b_m]):\alpha$, where $\Delta(k) = \gamma_1\to\cdots\to\gamma_m\to\alpha$ and for all $0 < j \leq m,\, b_j \in \cD_{\gamma_j}$. Let $\vec b = b_1,\dots,b_m$. By Definition~\ref{def:revtrans} we get $\appT{\trans[\vec x]{\rtrans{\Call k,[\vec b]}{\alpha}}}{[\vec a]} = \appT{\trans[\vec x]{x_k \cdot \rtrans[]{\Lookinv{b_1}}{\gamma_1}\cdots \rtrans[]{\Lookinv{b_m}}{\gamma_m}}}{[\vec a]}$. 
By an easy induction on $m$, one shows the following: 
\[
	\appT{\trans[\vec x]{x_k \cdot \rtrans[]{\Lookinv{b_1}}{\gamma_1}\cdots \rtrans[]{\Lookinv{b_m}}{\gamma_m}}}{[\vec a]} \equiv_\alpha \appT{\trans[\vec x]{x_k}}{[\vec a, \vec b]}\textrm{ (cf.\ Lemma~\ref{lem:applicative}).}
\]
By Definition~\ref{def:trans} and Lemma~\ref{lem:existenceofeams}\eqref{lem:existenceofeams1}, we get $ \appT{\trans[\vec x]{x_k}}{[\vec a, \vec b]} = \appT{\mProj{n}{k}}{[\vec a, \vec b]} \reddh \appT{\Lookinv{a_k}}{[\vec b]}$. 
Conclude by Lemma~\ref{lem:equivprops}\eqref{lem:equivconvh} and transitivity, as $\tuple{\vec{R}^r_{\vec a},\Call k,[\vec b]} \redh \appT{\Lookinv{a_k}}{[\vec b]} $.\qedhere
\end{itemize}
\end{proof}

\end{document}